\documentclass[a4paper,onecolumn,11pt,accepted=2022-04-22]{quantumarticle}
\pdfoutput=1
\usepackage[utf8]{inputenc}
\usepackage[english]{babel}
\usepackage[T1]{fontenc}
\usepackage{amsmath}
\usepackage{hyperref}

\usepackage{tikz}
\usepackage{lipsum}

\usepackage{microtype}
\usepackage[numbers]{natbib}

\usepackage{etoolbox}
\patchcmd{\thebibliography}{\section*{\refname}}{}{}{}
\usepackage{todonotes}
\usepackage{amssymb}
\usepackage{pgfplots}
\usepackage{amsthm}
\usepackage{mathrsfs}
\usepackage{dsfont}
\usepackage{mathtools}
\usepackage{physics}
\usepackage{bm}
\usepackage{thmtools}
\usepackage{verbatim}
\pgfplotsset{compat=1.17} 

\usepackage[capitalize]{cleveref}

\usepackage{xcolor}
\usepackage{qcircuit}
\newcommand{\Ngate}[2]{*+<.6em>{#1} \POS ="i","i"+UR;"i"+UL **\dir{-};"i"+DL **\dir{-};"i"+DR **\dir{-};"i"+UR **\dir{-},"i"  #2}
\newcommand{\Nmultigate}[3]{*+<1em,.9em>{\hphantom{#2}} \POS [0,0]="i",[0,0].[#1,0]="e",!C *{#2},"e"+UR;"e"+UL **\dir{-};"e"+DL **\dir{-};"e"+DR **\dir{-};"e"+UR **\dir{-},"i" #3}
\newcommand{\Nghost}[2]{*+<1em,.9em>{\hphantom{#1}} #2}
\newcommand{\ww}[1][-1]{\ar@{~} [0,#1]}
\newcommand{\dw}[1][-1]{\ar@{.} [0,#1]}
\newcommand{\dbw}[1][-1]{\ar@{=} [0,#1]}
\newcommand{\dbww}[1][-1]{\ar@{$\approx$} [0,#1]}
\usepackage{tikz}
\usetikzlibrary{positioning}
\usetikzlibrary{calc}
\usepackage{contour}
\usepackage{ulem}

\contourlength{0.8pt}
\newcommand{\myuline}[1]{%
	\uline{\phantom{#1}}%
	\llap{\contour{white}{#1}}%
}
\makeatletter
\newcommand\footnoteref[1]{\protected@xdef\@thefnmark{\ref{#1}}\@footnotemark}
\makeatother
\makeatletter
\def\Ddots{\mathinner{\mkern1mu\raise\p@
		\vbox{\kern7\p@\hbox{.}}\mkern2mu
		\raise4\p@\hbox{.}\mkern2mu\raise7\p@\hbox{.}\mkern1mu}}
\makeatother

\theoremstyle{plain}
\newtheorem{Thm}{Theorem}[section]
\newtheorem{Prop}[Thm]{Proposition}
\newtheorem{Cor}[Thm]{Corollary} 
\newtheorem{Lem}[Thm]{Lemma}

\newtheorem{OP}[Thm]{Open Problem}
\declaretheorem[style=definition,qed=$\blacksquare$, sibling=Thm]{Definition}
\declaretheorem[style=definition,qed=$\blacklozenge$,sibling=Thm]{Example}
\declaretheorem[style=definition,qed=$\maltese$,sibling=Thm]{Remark}

\newcommand{\N}{\mathbb{N}}

\newcommand{\C}{\mathbb{C}}

\newcommand{\cG}{\mathcal{G}}
\newcommand{\cF}{\mathcal{F}}
\newcommand{\cL}{\mathcal{L}}

\newcommand{\cR}{\mathcal{R}}

\newcommand{\cH}{\mathcal{H}}
\newcommand{\cP}{\mathcal{P}}

\newcommand{\cE}{\mathcal{E}}
\newcommand{\cK}{\mathcal{K}}

\newcommand{\cX}{\mathcal{X}}
\newcommand{\cY}{\mathcal{Y}}

\newcommand{\fB}{\mathfrak{B}}

\newcommand{\fF}{\mathfrak{F}}

\newcommand{\scrT}{\mathscr{T}}

\newcommand{\scrD}{\mathscr{D}}

\newcommand{\scrF}{\mathscr{F}}

\newcommand{\scrE}{\mathscr{E}}

\newcommand{\scrC}{\mathscr{C}}

\newcommand{\sfP}{\mathsf{P}}

\newcommand{\sfx}{\mathsf{x}}
\newcommand{\sfy}{\mathsf{y}}

\newcommand{\sfe}{\mathsf{e}}

\newcommand{\bone}{\mathds{1}}
\newcommand{\id}{\textup{id}}

\newcommand{\up}[1]{\textup{#1}}

\newcommand{\what}[1]{\widehat{#1}}

\newcommand{\End}[1]{\textup{End}(#1)}
\newcommand{\supp}[1]{\textup{supp}(#1)}

\newcommand{\sfA}{\mathsf{A}}
\newcommand{\sfB}{\mathsf{B}}

\newcommand{\Theory}{\mathbf{\Theta}}
\newcommand{\CIT}{\mathbf{CIT}}
\newcommand{\QIT}{\mathbf{QIT}}

		\newcommand{\triv}{\mathbf{1}}

		\usepackage{stmaryrd}

\usepackage{caption}
\usepackage{float}

\usepackage{tikz-cd}
\usepackage{adjustbox}
\usepackage[toc,page]{appendix}

\usepackage{epsdice}

\newcommand{\Bhv}[1]{\up{Bhv}(#1)}
\newcommand{\Strat}[1]{\up{Strat}(#1)}

\usepackage{stackengine,scalerel}
\newcommand{\cder}{\trianglerighteq}

\newcommand{\myQ}[3]{\begin{array}{c}\Qcircuit@C=#1em@R=#2em{#3} \end{array}}
\newcommand{\scalemyQ}[4]{\scalebox{#1}{$\myQ{#2}{#3}{#4}$}}

\newcommand{\res}{\up{res}}

\begin{document}

\title{An Operational Environment for \newline Quantum Self-Testing}

\author{Matthias Christandl}
\affiliation{Department of Mathematical Sciences, University of Copenhagen, \newline Universitetsparken 5, 2200 Copenhagen, Denmark}
\email{christandl@math.ku.dk}

\author{Nicholas Gauguin Houghton-Larsen}
\affiliation{Department of Mathematical Sciences, University of Copenhagen, \newline Universitetsparken 5, 2200 Copenhagen, Denmark}
\email{nicholas.gauguin@gmail.com}

\author{Laura Man\v{c}inska}
\affiliation{Department of Mathematical Sciences, University of Copenhagen, \newline Universitetsparken 5, 2200 Copenhagen, Denmark}
\email{mancinska@math.ku.dk}

\maketitle

\begin{abstract} 
Observed quantum correlations are known to determine in certain cases the underlying quantum state and measurements. This phenomenon is known as \emph{(quantum) self-testing}. 

Self-testing constitutes a significant research area with practical and theoretical ramifications for quantum information theory. But since its conception two decades ago by Mayers and Yao, the common way to rigorously formulate self-testing has been in terms of operator-algebraic identities, and this formulation lacks an operational interpretation. In particular, it is unclear how to formulate self-testing in other physical theories, in formulations of quantum theory not referring to operator-algebra, or in scenarios causally different from the standard one. 

In this paper, we explain how to understand quantum self-testing operationally, in terms of causally structured dilations of the input-output channel encoding the correlations. These dilations model side-information which leaks to an environment according to a specific schedule, and we show how self-testing concerns the relative strength between such scheduled leaks of information. As such, the title of our paper has double meaning: we recast conventional quantum self-testing in terms of information-leaks to an environment --- and this realises quantum self-testing as a special case within the surroundings of a general operational framework. 

Our new approach to quantum self-testing not only supplies an operational understanding apt for various generalisations, but also resolves some unexplained aspects of the existing definition, naturally suggests a distance measure suitable for robust self-testing, and points towards self-testing as a modular concept in a larger, cryptographic perspective.
\end{abstract}

\newpage

\setcounter{tocdepth}{2}
	
\tableofcontents
	
\newpage

\section{Introduction}
\label{sec:Intro}

\subsection{History and Motivation}
\label{subsec:Hist}

More than half a century ago, J.~S.~Bell understood that a maximally entangled pair of qubits may be subjected to local measurements such that the statistics of measurement outcomes predicted by quantum theory cannot be explained in terms of local or shared classical randomness \cite{Bell64}. Almost three decades later, A.~Ekert observed that this remarkable fact amounts to some of the randomness in the measurement outcomes being fundamentally \myuline{unknowable} in advance of the experiment in Bell's scenario, and therefore in principle useful for cryptographic purposes \cite{Ek91}. 

These ideas were revisited when D.~Mayers and A.~Yao introduced the notion of \emph{quantum self-testing} \cite{MY98,MY04}, according to which a collection of measurement statistics may in certain cases imply that a specific configuration of quantum state and measurements was used in the experiment.\footnote{Granted that the state in this configuration is pure and entangled, self-testing strengthens Ekert's observation (because no environment can hold information about a pure state), as well as Bell's (because no separable state can reproduce the measurement statistics).} Today, quantum self-testing has grown to a sizeable sub-field of quantum information theory, with applications in randomness generation \cite{Colb06,ColbKent11,Coud14, Sasa16, MS16, BMP18}, quantum key distribution \cite{MS16,JMS20} and delegated quantum computation \cite{RUV13,McK16,NV17,CGJV19,JMS20}, and with theoretical significance for the understanding of quantum correlations \cite{Goh18, CS18,Cola20embezzle} and quantum complexity theory \cite{NW19,MIP20}.

Though the cardinal idea behind self-testing --- that certain measurement statistics arise from a unique configuration of quantum state and measurements --- is perhaps intuitively meaningful, it is not entirely trivial to formalise. This is because there are strictly speaking always several different state-measurement-configurations which realise the same measurement statistics. For example, starting from any initial configuration, a simultaneous unitary rotation of the measurements and state will leave the statistics invariant. So will adding an ancillary state which is discarded in measurements. 

This paper is concerned with the proper mathematical formalisation of self-testing. \\

The conventional formal definition of self-testing was introduced in Ref.~\cite{MY04}, standardised in Ref.~\cite{MYS12} and has recently been reviewed in Ref.~\cite{SB19}. This definition relates the different possible state-measurement-configurations with given measurement statistics to a fixed \emph{canonical} configuration, by accounting for the possibility of unitary rotations and unmeasured ancillary states in the traditional operator-algebraic language for quantum theory.\footnote{Specifically, this is the definition of `full' self-testing, in which one makes a statement simultaneously about the state and the measurements; one can instead consider relaxed self-testing statements e.g.~about the state alone \cite{SB19}.} As witnessed by the breadth of self-testing results, this definition has been used successfully for more than 15 years to progress our understanding of theoretical and practical aspects of self-testing. 

However, because the definition is phrased in terms of operator-algebraic objects rather than information channels, it suffers from a lack of \emph{operational meaning}.\footnote{It is unfortunate that the syntactically similar `operator' and `operational' are here in semantic opposition to each other.} This circumstance is not only theoretically unaesthetic, but also hinders the immediate generalisation of self-testing to scenarios causally different from Bell's \cite{Gis10, Gis12, TF12}, to theories beyond quantum theory \cite{AbCo04, Sel04, Baez06,Barr07, Barn11,Barn16,Chir10,Hardy10}, or even to formalisations of quantum theory which transcend the operator-algebraic formalism \cite{Hard01,Chir11}. It is conceivable that quantum theory, or aspects thereof, will eventually be best understood and studied from an operational perspective free of operator-algebra, and thus it is highly relevant to have a definition of self-testing compatible with that mode of abstraction. \\

In this paper, we propose a new formalisation of self-testing which resolves these issues. Our formalisation is purely operational, and thus immediately interpretable in theories beyond quantum theory and independent of a specific formalism for quantum theory. It moreover naturally hints at an operational notion of `closeness' between state-measurement-configurations apt for the definition of approximate (i.e.~\emph{robust)} self-testing \cite{MYS12}, a concept which for practical purposes of self-testing is very important. Finally, the new approach points to interesting relaxations of self-testing which eventually exhibit self-testing as the special case of a phenomenon definable in a general framework of causally structured networks of channels --- as such, it automatically generalises self-testing to arbitrary causal scenarios. As a benefit, our new definition additionally clarifies certain vague narratives about the conventional operator-algebraic definition (e.g.~a commonly employed restriction to projective measurements) and yields new, general and intuitive proofs of many known features related to self-testing (e.g.~the facts that any configuration which can be self-tested uses a pure state and must give rise to extremal measurement statistics).

All of these results are discussed in \cref{subsec:MainIdea} in more detail, along with the main strokes of our formalisation.

\subsection{Main Idea and Results} 
\label{subsec:MainIdea}

Consider two separated parties $\sfA$ and $\sfB$ in possession of Hilbert spaces $\cH_\sfA$ and $\cH_\sfB$, respectively, and a shared quantum state $\varrho$ on $\cH_\sfA \otimes \cH_\sfB$. (Readers familiar with the CHSH-game can freely think of the case $\cH_\sfA = \cH_\sfB = \C^2$.) The two parties can locally apply measurements indexed by finite sets $X_\sfA$ and $X_\sfB$, with outcomes in finite sets $Y_\sfA$ and $Y_\sfB$, respectively. (Again, for concreteness, the reader may consider the case where $X_\sfA = X_\sfA = \{0,1\}$ and $Y_\sfA = Y_\sfB = \{1,-1\}$, so that each party can choose between two measurements, each with two outcomes $\pm 1$.) 

For each index $x_\sfA \in X_\sfA$, the corresponding measurement for party $\sfA$ may be represented by a quantum channel $\Lambda^{x_\sfA}_\sfA$ from $\cH_\sfA$ to $\C^{Y_\sfA}$, whose outputs are diagonal in the computational basis (for any finite set $I$, we denote by $\C^I$ the Hilbert space with computational basis $(\ket{i})_{i \in I}$).\footnote{As discussed in \cref{subsec:ClasStruct}, the channel $\Lambda^{x_\sfA}_\sfA$ may be represented by a POVM on $\cH_\sfA$, namely a collection $E^{x_\sfA}_\sfA = \big(E^{x_\sfA}_\sfA(y_\sfA)\big)_{y_\sfA \in Y_\sfA}$ of positive operators on $\cH_\sfA$ such that $\sum_{y_\sfA \in Y_\sfA} E^{x_\sfA}_\sfA(y_\sfA) = \bone_{\cH_\sfA}$. Specifically, the measurement channel $\Lambda^{x_\sfA}_\sfA$ associated to the POVM $E^{x_\sfA}_\sfA$ is given by $\Lambda^{x_\sfA}_\sfA(T) = \sum_{y_\sfA \in Y_\sfA} \tr(E^{x_\sfA}_\sfA(y_\sfA) T)  \, \ketbra{y_\sfA}$ for $T \in \End{\cH_\sfA}$.} The measurement ensemble $(\Lambda^{x_\sfA}_\sfA)_{x_\sfA \in X_\sfA}$ may in turn be described by a single quantum channel $\Lambda_\sfA$ from the system $\C^{X_\sfA} \otimes \cH_\sfA$ to the system $\C^{Y_\sfA}$, which measures the system $\C^{X_\sfA}$ in the computational basis to read off the measurement choice $x_\sfA \in X_\sfA$ and then applies the corresponding measurement $\Lambda^{x_\sfA}_\sfA$. Likewise, the measurement ensemble for party $\sfB$ may be summarised by a single quantum channel $\Lambda_\sfB$ from $\C^{X_\sfB} \otimes \cH_\sfB$ to $\C^{Y_\sfB}$. Altogether, the quantum channels $\Lambda_\sfA$ and $\Lambda_B$ describe how classical information from the systems $\C^{X_\sfA}$ and $\C^{X_\sfB}$ along with quantum information from the systems $\cH_\sfA$ and $\cH_\sfB$ is transformed to classical information in the systems $\C^{Y_\sfA}$ and $\C^{Y_\sfB}$, respectively. 

The quadruple $\fB = (X_\sfA, X_\sfB, Y_\sfA, Y_\sfB)$ which defines the measurement indices and outcomes is called a \emph{Bell-scenario}, and the triple $S= (\varrho, \Lambda_\sfA, \Lambda_\sfB)$ is called a \emph{quantum strategy} for $\fB$. To any quantum strategy $S$ is an associated \emph{(input-output) behaviour}, namely the channel from $\C^{X_\sfA} \otimes \C^{X_\sfB}$ to $\C^{Y_\sfA} \otimes \C^{Y_\sfB}$ given by $(\Lambda_\sfA \otimes \Lambda_\sfB) \circ (\id_{\C^{X_\sfA}} \otimes \varrho \otimes \id_{\C^{X_\sfB}})$, which takes as input any pair of measurement indices and yields as output the joint distribution of measurement outcomes. We may pictorially represent the behaviour channel as
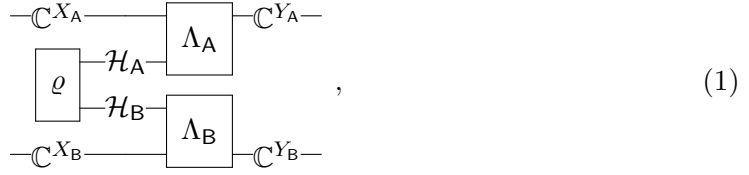
\begin{align} \label{eq:BHV}
	\myQ{0.7}{0.7}{
		& \push{\C^{X_\sfA}}  \qw 	& \qw  & \multigate{1}{\Lambda_\sfA} & \push{\C^{Y_\sfA}}  \qw & \qw   \\
		& \Nmultigate{1}{\varrho}  & \push{\cH_\sfA} \qw & \ghost{\Lambda_\sfA} &  \\
		& \Nghost{\varrho} & \push{\cH_\sfB} \qw  & \multigate{1}{\Lambda_\sfB}  \\
	&   \push{\C^{X_\sfB}} \qw  		& \qw  & \ghost{\Lambda_\sfB} & \push{\C^{Y_\sfB}}  \qw & \qw \\
	},
	\end{align}
with time flowing from left to right. This channel can be identified with a classical information channel from the set $X:= X_\sfA \times X_\sfB$ to the set $Y:= Y_\sfA \times Y_\sfB$ and as such the behaviour of a strategy is equivalent to a collection $P= (P^x)_{x \in X}$ of probability distributions on $Y$.\footnote{This collection is sometimes referred to as the `\emph{correlations}' produced by the strategy. The term `behaviour', which we shall use throughout, is due to Cirelson \cite{Cir93}.} (In \cref{subsec:Hist}, such collections were referred to as `measurement statistics' and the quantum strategies as `configurations'.)\\

Now, the conventional formalisation of self-testing is roughly as follows. One says that the behaviour $P$ \emph{self-tests} the strategy $\tilde{S}= (\tilde{\varrho}, \tilde{\Lambda}_\sfA, \tilde{\Lambda}_\sfB)$, if any other strategy $S=(\varrho, \Lambda_\sfA, \Lambda_\sfB)$ with behaviour $P$ is `reducible' to $\tilde{S}$, where \emph{reducibility} is a certain relation defined among quantum strategies in terms of the constituents $\varrho, \Lambda_\sfA, \Lambda_\sfB, \tilde{\varrho}, \tilde{\Lambda}_\sfA$ and $\tilde{\Lambda}_\sfB$. The details of this definition are reviewed in \cref{sec:Standard}, but for now it suffices to understand that the main issue addressed by our paper is that the definition is phrased using operator-algebraic objects which abstractly represent these constituents, rather than using operational features of the constituents themselves. For example, when the states $\varrho$ and $\tilde{\varrho}$ are pure, represented by unit vectors $\ket{\psi} \in \cH_\sfA \otimes \cH_\sfB$ and $\ket*{\tilde{\psi}} \in \tilde{\cH}_\sfA \otimes \tilde{\cH}_\sfB$, and when the measurements are all projective, represented by PVM-elements $\Pi^{x_\sfA}_\sfA(y_\sfA)$, $\Pi^{x_\sfB}_\sfB(y_\sfB)$,  $\tilde{\Pi}^{x_\sfA}_\sfA(y_\sfA)$ and $\tilde{\Pi}^{x_\sfB}_\sfB(y_\sfB)$, then the reducibility condition requires the existence of isometries $W_\sfA$, $W_\sfB$ and a unit vector $\ket*{\psi^\res}$ such that $[W_\sfA \otimes W_\sfB] [\Pi^{x_\sfA}_\sfA(y_\sfA) \otimes\Pi^{x_\sfB}_\sfB(y_\sfB) ] \ket{\psi} = [\tilde{\Pi}^{x_\sfA}_\sfA(y_\sfA) \otimes \tilde{\Pi}^{x_\sfB}_\sfB(y_\sfB) ] \ket*{\tilde{\psi}} \otimes \ket*{\psi^\res}$ for all $x_\sfA \in X_\sfA$, $x_\sfB \in X_\sfB$, $y_\sfA \in Y_\sfA$ and $y_\sfB \in Y_\sfB$.

What we propose to do instead is the following. To each strategy $S= (\varrho, \Lambda_\sfA, \Lambda_\sfB)$, we may associate the quantum channel
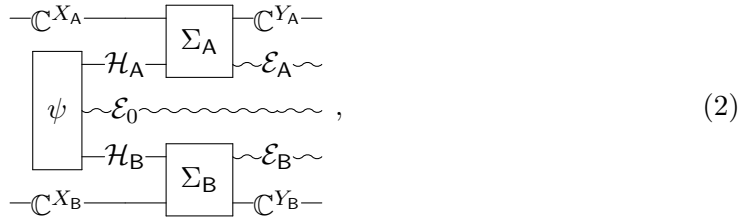
\begin{align}
	\label{eq:CS}
	    \myQ{0.7}{0.7}{
	& \push{\C^{X_\sfA}}  \qw  & \qw & \multigate{1}{ \Sigma_\sfA} & \push{\C^{Y_\sfA}}  \qw & \qw  \\
	& \Nmultigate{2}{\psi}  & \push{\cH_\sfA} \qw & \ghost{\Sigma_\sfA} & \push{\cE_\sfA} \ww & \ww \\
	& \Nghost{\psi}  & \push{\cE_0} \ww & \ww & \ww & \ww \\
	& \Nghost{\psi} & \push{\cH_\sfB} \qw & \multigate{1}{\Sigma_\sfB} & \push{\cE_\sfB} \ww & \ww  \\
	&   \push{\C^{X_\sfB}} \qw  & \qw&  \ghost{\Sigma_\sfB} & \push{\C^{Y_\sfB}}  \qw & \qw  \\
} ,
	\end{align}
algebraically given by $(\Sigma_\sfA \otimes \id_{\cE_0} \otimes  \Sigma_\sfB)\circ  (\id_{\what{X}_\sfA} \otimes \psi \otimes \id_{\what{X}_\sfB})$, from $\C^{X_\sfA} \otimes \C^{X_\sfB}$ to $\C^{Y_\sfA} \otimes \C^{Y_\sfB} \otimes \cE_\sfA \otimes \cE_\sfB  \otimes \cE_0$, where $\Sigma_\sfA$ and $\Sigma_\sfB$ are Stinespring dilations of $\Lambda_\sfA$ and $\Lambda_\sfB$, respectively, and where $\psi$ is a Stinespring dilation (i.e.~purification) of $\varrho$. (We review Stinespring dilations in \cref{subsec:Stinespring}; forming a Stinespring dilation is sometimes called `going to the Church of the Larger Hilbert Space'.) The channel \eqref{eq:CS} will be called the \emph{Stinespring implementation of $S$}.\footnote{Strictly speaking, the association of Stinespring implementations to strategies is not unique, since Stinespring dilations are not themselves unique; however, it is unique up to a triple of channels acting on the systems $\cE_\sfA$, $\cE_0$ and $\cE_\sfB$, cf.~\cref{thm:Stine}. (If we were to take each Stinespring dilation \myuline{minimal}, the uniqueness would be even up to a triple of \myuline{unitary} channels acting on $\cE_\sfA$, $\cE_0$ and $\cE_\sfB$.)} 

Now, the association of Stinespring implementations to quantum strategies does not look terribly operational, as Stinespring dilations are specific to quantum theory. However, for the purpose of redefining self-testing we shall only rely on the following operational feature of a Stinespring implementation: it represents \emph{maximal side-information} leaking to an environment in the course of employing the strategy $S$, and as such is strongest among more general `implementations' of the strategy (this will be made precise in \cref{def:Sim} and \cref{prop:Rechar}).

Specifically, the operational interpretation of a Stinespring implementation is that the system $\cE_0$ holds the maximal amount of information extractable about the strategy $S$  \myuline{before} knowing the inputs $x_\sfA$ and $x_\sfB$, whereas for $\sfP \in \{\sfA, \sfB\}$, the system $\cE_\sfP$ holds the maximal information extractable \myuline{after} receiving the input $x_\sfP$. For example, a copy of the input $x_\sfP$ will be derivable from the side-information in $\cE_\sfP$ (though the information content is generally larger, as is it a Stinespring dilation and thus conveys also quantum information), and if the state $\varrho$ is mixed then some classical non-trivial side-information will be available already in the system $\cE_0$ (but again, its information content is generally larger). \\

The Stinespring implementation \eqref{eq:CS} of $S$ is evidently a Stinespring dilation of the behaviour channel \eqref{eq:BHV} when considering the system $\cE_\sfA \otimes \cE_0 \otimes \cE_\sfB$ as one single environment; but the whole point will be that the concept of self-testing emerges by virtue of the difference between the three systems $\cE_\sfA$, $\cE_\sfB$ and $\cE_0$ which is due to their causal relationships to the inputs. 
Indeed, the key insight in our paper is that the Stinespring implementations of strategies can be used to recast the conventional operator-algebraic notion of reducibility among strategies, and thus of self-testing, in a way which is purely operational. 

In fact, we will introduce \myuline{three} different relations among strategies, all defined in terms of their Stinespring implementations (but definable in arbitrary theories), and all of which illuminate in new ways the conventional understanding of self-testing. We describe these three relations below, along with our results. \\

Let $S=(\varrho, \Lambda_\sfA, \Lambda_\sfB)$ and $S'=(\varrho', \Lambda'_\sfA, \Lambda'_\sfB)$ be quantum strategies, and let $\psi$, $\Sigma_\sfA$, $\Sigma_\sfB$, $\psi'$, $\Sigma'_\sfA$ and $\Sigma'_\sfB$ be Stinespring dilations of their respective constituents. Consider the condition 
\begin{align}   \label{eq:Relations}
\myQ{0.7}{0.7}{
	& \push{\C^{X_\sfA}}  \qw  & \multigate{1}{ \Sigma_\sfA} & \qw & \push{\C^{Y_\sfA}}  \qw & \qw & \qw \\
	& \Nmultigate{2}{\psi}  & \ghost{\Sigma_\sfA} &  \push{\cE_\sfA} \ww & \Nmultigate{2}{\Gamma}{\ww}  &  \push{\cE'_\sfA} \ww& \ww\\
	& \Nghost{\psi} & \ww &  \push{\cE_0} \ww &\Nghost{\Gamma}{\ww}  &  \push{\cE'_0} \ww& \ww \\
	& \Nghost{\psi} & \multigate{1}{\Sigma_\sfB} &   \push{\cE_\sfB} \ww &\Nghost{\Gamma}{\ww} &  \push{\cE'_\sfB} \ww & \ww \\
	&   \push{\C^{X_\sfB}} \qw  & \ghost{\Sigma_\sfB} & \qw & \push{\C^{Y_\sfB}}  \qw & \qw & \qw  \\
} 
\quad = \quad 
\myQ{0.7}{0.7}{
	& \push{\C^{X_\sfA}}  \qw  & \multigate{1}{ \Sigma'_\sfA} & \push{\C^{Y_\sfA}}  \qw & \qw  \\
	& \Nmultigate{2}{\psi'}  & \ghost{\Sigma'_\sfA} &  \push{\cE'_\sfA} \ww &  \ww\\
	& \Nghost{\psi'} & \ww  &  \push{\cE'_0} \ww& \ww \\
	& \Nghost{\psi'} & \multigate{1}{\Sigma'_\sfB} &  \push{\cE'_\sfB} \ww & \ww  \\
	&   \push{\C^{X_\sfB}} \qw  & \ghost{\Sigma'_\sfB} & \push{\C^{Y_\sfB}}  \qw & \qw  \\
} 
\end{align}
among the corresponding Stinespring implementations, for some channel $\Gamma$ from $\cE_\sfA \otimes \cE_0 \otimes \cE_\sfB$ to $\cE'_\sfA \otimes \cE'_0 \otimes \cE'_\sfB$. (Here, we have dropped the labelling of systems $\cH_\sfA$, $\cH_\sfB$ and $\cH'_\sfA$, $\cH'_\sfB$ corresponding to the internal wires on the left and right hand side, respectively.) A channel $\Gamma$ satisfying \cref{eq:Relations} exists if and only if $S$ and $S'$ have the same behaviour, indeed the identity expresses simply that two Stinespring dilations of the behaviour channel are equivalent (\cref{thm:Stine}). If however we restrict the \myuline{form} that $\Gamma$ may take in \cref{eq:Relations}, then we obtain finer, non-trivial relations among strategies. In particular, we will introduce the following three operational relations and derive a number of results illuminating self-testing in terms of them. 
\begin{itemize}
    \item $S$ \emph{locally simulates} $S'$ (we write $S \geq_{l.s.} S'$) if there exists $\Gamma$ of the form $\scalemyQ{.8}{0.7}{0.7}{    &  \push{\cE_\sfA} \ww &  \Ngate{\Gamma_\sfA}{\ww} & \push{\cE'_\sfA} \ww & \ww \\
	  & \push{\cE_0} \ww &   \Ngate{\Gamma_0}{\ww}  & \push{\cE'_0} \ww & \ww \\
	  &  \push{\cE_\sfB} \ww & \Ngate{\Gamma_\sfB}{\ww} & \push{\cE'_\sfB}\ww & \ww}$ satisfying \cref{eq:Relations}.
	  
	  As we show, any quantum strategy $S$ is equivalent under local simulation to a strategy $S_\up{proj.}$ with \myuline{projective} measurements (\cref{prop:Proj}). In this formal sense, measurements may always without loss of generality be assumed projective. Furthermore, under a simple rank-assumption on the state of $\tilde{S}$, the operator-algebraic notion of $S$ being reducible to $\tilde{S}$ is equivalent to the operational notion of $S$ locally simulating $\tilde{S}$ (\cref{thm:RvsS}). This is our key result, and the only one with a technical proof.
	  
	  Local simulation gives rise to a natural notion self-testing (\cref{def:OpSelf}):
   We say that \emph{$P$ self-tests $\tilde{S}$ according to local simulation} (for short, \emph{a.t.l.s.}) if any strategy $S$ with behaviour $P$ locally simulates $\tilde{S}$, i.e.~$S \geq_{l.s.} \tilde{S}$.
    This notion allows to recast the conventional operator-algebraic formalisation of self-testing in operational terms
    (\cref{thm:SelfRvsS}), giving a contribution to the understanding of self-testing.

     \item $S$ \emph{locally assisted simulates} $S'$ (written $S \geq_{l.a.s.} S'$) if there exists $\Gamma$ of the form $\scalemyQ{.8}{0.7}{0.7}{ 	   & \push{\cE_\sfA} \ww &  \ww &  \Nmultigate{1}{\Gamma_\sfA}{\ww} & \push{\cE'_\sfA} \ww & \ww   \\
	   &  & \Nmultigate{3}{ \alpha } & \Nghost{\Gamma_\sfA}{\ww}  \\  
	   & \push{\cE_0} \ww &   \ww &   \Nmultigate{1}{\Gamma_0}{\ww} & \push{\cE'_0} \ww & \ww  \\
	   &   & \Nghost{ \alpha } &   \Nghost{\Gamma_0}{\ww}   \\
	   &  & \Nghost{ \alpha } &   \Nmultigate{1}{\Gamma_\sfB}{\ww}   \\
	  & \push{\cE_\sfB} \ww  &  \ww & \Nghost{\Gamma_\sfB}{\ww} & \push{\cE'_\sfB} \ww & \ww    }
 $ satisfying \cref{eq:Relations};
     
     Local assisted simulation is coarser than local simulation and, somewhat surprisingly, \myuline{symmetric} and therefore an equivalence relation among quantum strategies (\cref{thm:AssSym}). In fact, local assisted simulation is the equivalence relation \emph{generated by} local simulation (\cref{rem:SelfEquiv}).
As a consequence of this, if a behaviour $P$ self-tests a strategy $\tilde{S}$ in the conventional sense, then all strategies with behaviour $P$ are \myuline{equivalent} in a formal sense, namely according to local assisted simulation. (For example, all optimal strategies for the CHSH-game are equivalent under local assisted simulation.)
     
       One may define a new notion of \emph{self-testing according to local assisted simulation} (for short, \emph{self-testing a.t.l.a.s.}), which requires $S \leq_{l.a.s.} \tilde{S}$ for every strategy $S$ with behaviour $P$. By the above results, the condition $S \leq_{l.a.s.} \tilde{S}$ weakens the relation $S \geq_{l.s.} \tilde{S}$ (since the latter implies $S \geq_{l.a.s.} \tilde{S}$ as well as $S \leq_{l.a.s.} \tilde{S}$). Therefore, this notion of self-testing relaxes the conventional notion of self-testing. It puts emphasis on \myuline{upper} bounding the abilities of unknown strategies in terms of the known strategy $\tilde{S}$ (in line with Ekert \cite{Ek91}) rather than \myuline{lower} bounding the abilities of unknown strategies (in line with Bell \cite{Bell64}).
       
      \item $S$ \emph{causally simulates} $S'$ (written $S \geq_{c.s.} S'$) if there exists $\Gamma$ of the form $\scalemyQ{.8}{0.7}{0.7}{  &  &  \push{\cE_\sfA} \ww &  \Nmultigate{1}{\Gamma_\sfA}{\ww} & \push{\cE'_\sfA} \ww & \ww \\
	    &  & \Nmultigate{2}{\Gamma_0} & \Nghost{\Gamma_\sfA}{\ww} &  \\  
	  & \push{\cE_0} \ww &   \Nghost{\Gamma_0}{\ww}   & \push{\cE'_0} \ww & \ww  \\
	&  & \Nghost{ \Gamma_0} &   \Nmultigate{1}{\Gamma_\sfB}{\ww}   \\
	 & &  \push{\cE_\sfB} \ww & \Nghost{\Gamma_\sfB}{\ww} & \push{\cE'_\sfB}\ww & \ww}$ satisfying \cref{eq:Relations}.
    
This notion of simulation is again coarser than the previous one. And like local assisted simulation, causal simulation gives rise to a self-testing notion too, by requiring $S \leq_{c.s.} \tilde{S}$ for every strategy $S$ with behaviour $P$. We call this \emph{self-testing according to causal simulation} (for short, \emph{self-testing a.t.c.s.}).\footnote{When the state of $\tilde{S}$ is pure, self-testing a.t.c.s.~coincides with self-testing a.t.l.a.s.~(since in that case $S \leq_{c.s.} \tilde{S}$ is equivalent to $S \leq_{l.a.s.} \tilde{S}$, by inspecting the structure of the channel $\Gamma$).} The significance of self-testing a.t.c.s.~is that it realises self-testing as part of a very general theory of \emph{causal channels} and \emph{causal dilations} introduced in the PhD thesis \cite{Hou21}. We discuss this connection in \cref{subsec:Causal}. 

\end{itemize}
%

We emphasize that all three relations turn out be operationally definable; this follows from the operational `maximal leak'-feature of Stinespring implementations, which we will make explicit in alternative characterisations of the relations (see \cref{def:Sim}, \cref{rem:AssOp} and \cref{rem:CausDefOp}). In particular, it will be clear how to interpret the three modes of simulation in any other theory of information channels which are serially and parallelly composable \cite{Hou21}.\\




To the best of our knowledge, the operational account of general self-testing presented in this paper is the first of its kind to appear in the literature. It should be noted, however, that Ref.~\cite{LOSR} has discussed operational formulations of self-testing of \underline{states}. Specifically, the authors introduce a notion of extractability in terms of local operations and shared classical randomness, and they render a state $\varrho $ self-tested by a behaviour $B$, if $B$ can be extracted from $\varrho$ and if $\varrho$ is extractible from any other state $\sigma$ from which $B$ can be extracted. Though also based on an operational idea of extraction, this approach differs from ours in two important respects. Firstly, it differs in complexity by disregarding the \myuline{measurements} in a self-testing scenario; in our approach, the appearance of measurements is manifested in the paramount causal difference between outputs, which we discussed above (the inclusion of measurements is also what makes the equivalence to standard self-testing non-trivial to establish, cf.~\cref{lem:technical}). Secondly, the two approaches differ in concept, since we consider \myuline{dilations} of the involved state (and measurements). As a consequence, certain mixed states are self-testable according to the definition in Ref.~\cite{LOSR}, whereas, according to ours, the inclusion of dilations ultimately turns out to demand purity of self-testable states (cf.~\cref{cor:SelfPure}).

After completion of the research in this work, we have been made aware of Ref.~\cite{RB21}, where notions of self-testing are discussed in classical causal network settings. It is indeed also possible to consider our operational definition of self-testing in non-quantum theories, for example the theory of classical information (we call this theory $\CIT$). We believe it would be interesting to conduct a study of self-testing in other theories, this was, however, out of scope for the work presented here. We refer the interested reader to Ref.~\cite{Hou21}.

\subsection{Outline of the Paper}

\begin{itemize}
   \item[$\triangleright$]
We shall start by making more precise what we mean by \emph{operational concepts} in a physical theory such as quantum theory. This will be the topic of \cref{sec:Operational}, which also serves as preliminary section to establish common notions, terminology and results from quantum information theory. We discuss in particular \emph{compositional structure} of a theory (\cref{subsec:CompStruct}), the interplay with \emph{classical notions} such as probabilities (\cref{subsec:ClasStruct}), and we recall Naimark's theorem for measurements (\cref{subsec:Naimark}) and Stinespring's dilation theorem (\cref{subsec:Stinespring}). Readers who are familiar with these results and already feel comfortable with the notion of operational structure may skip \cref{sec:Operational}, referring to it if need be.
\item[$\triangleright$] In \cref{sec:Standard}, we discuss the conventional perception of self-testing (\cref{subsec:Aim}) and recall in detail the standard operator-algebraic definition of reducibility and self-testing (\cref{subsec:Standard}). We then discuss the need for an alternative understanding of self-testing (\cref{subsec:SumCrit}), which is operational in the sense of \cref{sec:Operational}.
\item[$\triangleright$] Our new operational approach to self-testing is presented in \cref{sec:Rework}, which is the core section of the paper. We introduce here \emph{local simulation} and the corresponding reformulation of self-testing, and we establish the equivalence to reducibility and the conventional notion of self-testing (\cref{subsec:Sim}). We also introduce the relaxation to \emph{local assisted simulation} and discuss how it illuminates quantum self-testing in new ways (\cref{subsec:AssSim}). Finally, we explain how the further relaxation to \emph{causal simulation} allows us to realise quantum self-testing from an even larger perspective, namely within the framework of causal channels and dilations (\cref{subsec:Causal}). 
 \item[$\triangleright$] In \cref{sec:Utility}, we end by showing how several features of self-testing can be given new, operational proofs based on our reformulations of self-testing.\footnote{\cref{sec:Utility} may, apart for the definition of causal simulation (\cref{def:CausSim}), be read independently of \cref{subsec:Causal}.} In particular, we establish a relationship between local simulation and state extractability (\cref{subsec:StateExt}), between local assisted simulation and measurement extractability (\cref{subsec:MeasExt}), and between causal simulation and the structure of convex decompositions of the behaviour (\cref{subsec:SecExt}). We also observe an operational feature of self-testing which we believe has so far gone unnoticed, but which from our operational formulation is natural (\cref{subsec:Exhausted}). It suggests that self-testing is in certain cases characterised by the property of `exhausting' the parties $\sfA$ and $\sfB$ of information.  
\end{itemize}
Each of the above four sections begins with a more detailed section outline. In the concluding \cref{sec:Conclusions}, we summarise the paper and discuss possible directions for future investigations.

\section{The Operational Structure of Quantum Information Theory}
\label{sec:Operational}

This section has two purposes. Partly, it serves as a preliminary reminder of quantum information theory as cast in the mathematical language of completely positive trace-preserving maps (quantum channels). Partly, the style of presentation serves to exemplify what we mean by the \emph{operational structure} of a physical theory, in line with a `pragmatist' tradition of physics \cite{CFS16}.  We define a concept to be \emph{operational} if it can be formulated by referring only to operational structure. Our definition of this (\cref{def:Operational}) is not fully formalised, but we render it safe to adopt the viewpoint that it is evident when a concept has been cast as operational according to this definition (whereas it may \myuline{not} be evident when this is impossible to do). The section is divided into four subsections. 

Roughly speaking, operational structure of a theory is comprised by (A) a purely \emph{compositional structure} (as pioneered in categorical frameworks \cite{AbCo04,Sel04, Baez06}), and additionally often (B) a connection to \emph{classical notions}, e.g.~probabilities (as usually formalised using Generalised Probabilistic Theories \cite{Barr07, Barn11, Barn16}). These two pillars are discussed in \cref{subsec:CompStruct} and \cref{subsec:ClasStruct}, respectively, with a special focus on quantum information theory. The compositional structure of a theory already sustains a notion of \emph{compositional concepts} (informally captured in \cref{def:Comp}), and by adding classical structure we then obtain the full notion  of \emph{operational concepts} (\cref{def:Operational}). Readers who already feel intuitively comfortable with the notion of `operational concepts' may skip \cref{subsec:CompStruct} and \cref{subsec:ClasStruct}, and refer to these subsections later if necessary.

 In \cref{subsec:Naimark} we recall Naimark's theorem for quantum measurement, and, finally, in \cref{subsec:Stinespring} we make some observations about \emph{isometric} quantum channels (\cref{ex:Isom}), reviewing in particular Stinespring's dilation theorem \cite{Stine55} from an operational perspective.

\subsection{The Compositional Structure of \textbf{QIT}}
\label{subsec:CompStruct}

The compositional structure of a physical theory comprises a collection of \emph{systems} and \emph{transformations} and defines how these elements can be \emph{composed}. Systems are interpreted as the basic information-carrying entities in the theory, and any two systems can be composed to form a single, joint system. Transformations from one system to another are interpreted as `physical processes' or `information channels' which modify the information carried by the systems, and transformations can be composed either \emph{serially} (one performed after another) or \emph{parallelly} (one performed alongside with another). 

Mathematically, this structure is succinctly captured by a \emph{symmetric monoidal category} \cite{MacLane}, however it suffices for our purposes to intuitively understand this concept insofar as it concerns the theory of quantum information.  \\

In quantum information theory, a \myuline{system} is identified with a non-zero finite-dimensional Hilbert space.\footnote{One can define a version of quantum information theory in which infinite-dimensional systems are allowed, but the definition of quantum channels is then slightly more technical and we shall not need this.} We generically typeset systems with capital calligraphic letters $\cH, \cK, \cL, \ldots$. A \myuline{transformation} between systems is in quantum information theory identified with a so-called \emph{quantum channel}. For any Hilbert space $\cL$, we denote by $\End{\cL}$ the vector space of all linear operators on $\cL$; a quantum channel from system $\cH$ to system $\cK$ is then by definition a linear map $\Lambda: \End{\cH} \to \End{\cK}$ which is \emph{trace-preserving} (meaning that $\tr(\Lambda(A))= \tr(A)$ for any $A \in \End{\cH}$) and \emph{completely positive} (meaning that for any system $\cR$ and any positive operator $A \in \End{\cH \otimes \cR}$, the operator $(\Lambda \otimes \id_{\End{\cR}})(A) \in \End{\cK \otimes \cR}$ is positive\footnote{Note that $\Lambda \otimes \id_{\End{\cR}}$ can be thought of as a map $\End{\cH \otimes \cR} \to \End{\cK \otimes \cR}$ by observing the canonical isomorphisms $\End{\cH \otimes \cR}  \cong \End{\cH} \otimes \End{\cR}$ and  $\End{\cK \otimes \cR}  \cong \End{\cK} \otimes \End{\cR}$.}). We generically typeset quantum channels with Greek letters $\Lambda,  \Phi, \Sigma, \ldots$. 

Transformations $\Lambda: \End{\cH} \to \End{\cK}$ and $\Phi: \End{\cK} \to \End{\cL}$ can be composed \myuline{serially} by means of functional composition to yield the transformation $\Phi \circ \Lambda : \End{\cH} \to \End{\cL}$. The serial composition is associative (meaning that $\Xi \circ (\Phi \circ \Lambda) = (\Xi \circ \Phi) \circ \Lambda$), and there exists for every system $\cH$ a distinguished identity transformation which acts as unit for the serial composition, namely the identity channel $\id_{\End{\cH}} : \End{\cH} \to \End{\cH}$, $A \mapsto A$, which for brevity we denote $\id_\cH$. These features of systems and transformations so far gives quantum information theory the mathematical structure of a \emph{category} \cite{MacLane}, namely a collection of morphisms which compose serially subject to natural requirements.

\begin{Example}[Isometric Quantum Channels] \label{ex:Isom}
If $V: \cH \to \cK$ is an \emph{isometry}, i.e.~a linear operator with $V^* V = \bone_\cH$, then the linear map $\Sigma_V: \End{\cH} \to \End{\cK}$ given by
$\Sigma_V(A)= VAV^*$ is trace-preserving and completely positive, hence a quantum channel from $\cH$ to $\cK$. We call quantum channels of this form \emph{isometric}. The identity transformation $\id_\cL$ on a system $\cL$ is isometric, and the serial composition of isometric channels is again isometric. \end{Example}

Quantum information theory has more structure than that of a category. Indeed, systems $\cH_1$ and $\cH_2$ can be composed by means of the tensor product to form the joint system $\cH_1 \otimes \cH_2$, and transformations $\Lambda_1 : \End{\cH_1} \to \End{\cK_1}$ and $\Lambda_2 : \End{\cH_2} \to \End{\cK_2}$ can be composed \myuline{parallelly} by means of the tensor product to yield the transformation $\Lambda_1 \otimes \Lambda_2 : \End{\cH_1 \otimes \cH_2} \to \End{\cK_1 \otimes \cK_2}$. (Complete positivity ensures that the resulting map is a quantum channel.) The parallel composition of transformations is tied to the composition of system via the identities $\id_{\cH_1 \otimes \cH_2} = \id_{\cH_1} \otimes \id_{\cH_2}$. Parallel composition is commutative and associative among transformations as well and systems,\footnote{Or rather, commutative and associative `up to natural isomorphisms'; for systems we have $\cH_1 \otimes \cH_2 \cong \cH_2 \otimes \cH_1$ and $(\cH_1 \otimes \cH_2) \otimes \cH_3 \cong \cH_1 \otimes (\cH_2 \otimes \cH_3)$, and for transformations similar conditions.} and there is a distinguished \emph{trivial system}, namely the one-dimensional Hilbert space $\triv := \C$, which physically represents `nothing'. The trivial system acts as a unit for the composition of systems ($\cH \otimes \triv  \cong \cH$) and its identity transformation $\id_\triv$ acts as a unit for the parallel composition. All these additional features of quantum theory promotes its mathematical structure to that of a \emph{\myuline{symmetric} \myuline{monoidal} category} \cite{MacLane}, namely a category augmented by a notion of parallel composition (the word `symmetric' refers to the commutativity of parallel composition).\\

In our definition of compositional structure, there is one more feature of quantum information theory we wish to impose generally:\footnote{Strictly speaking, this feature is not an additional piece of structure but rather an additional axiom about the existing structure.} the trivial system $\triv$ is such that from any system $\cH$ there is a unique transformation to $\triv$, namely the trace $\tr_\cH : \End{\cH} \to \C (\cong \End{\triv} ) $. As the system $\triv$ represents `nothing', this physically corresponds to saying that there is precisely one way of discarding information in a system $\cH$. This is important to us because it facilities a unique notion of \emph{marginalisation} by way of discarding only part of a system (e.g.~from $\cH_1 \otimes \cH_2$ to $\cH_1$ we have the transformation $\id_{\cH_1} \otimes \tr_{\cH_2}$, called \emph{partial trace} in quantum information theory). We shall rely heavily on marginalisation later on. 

As such, we arrive at the following model for compositional structure:

\begin{Definition}[Compositional Theories]
    A \emph{compositional theory} $\Theory$ is a symmetric monoidal category in which every system $\cH$ admits precisely one transformation to the trivial system $\triv$, denoted $\tr_\cH$. We denote the serial composition in $\Theory$ by $\circ$ and the parallel composition and composition of systems in $\Theory$ by $\otimes$. 
\end{Definition}

\begin{Remark} Symmetric monoidal categories form a well-established model for parallel and serial composition in physical theories \cite{AbCo04,Sel04,Baez06}, and the uniqueness requirement for the trivial system can be seen as formalising the impossibility of signalling from the future to the past \cite{Chir10,CoLal13,Coecke14}. Symmetric monoidal categories in which every system admits a unique transformation to $\triv$ are sometimes called \emph{semicartesian} symmetric monoidal categories.  \end{Remark}

\begin{Example}
We will denote by $\QIT$ the compositional theory of quantum information, which we defined above.
\end{Example}

\begin{Example} \label{ex:CIT}
We will denote by $\CIT$ the compositional theory of \myuline{classical} information, where systems are finite non-empty sets composing under the cartesian product $\times$, and where transformations from $X$ to $Y$ are classical channels (Markov kernels) from $X$ to $Y$, i.e.~families $T= (t^x)_{x \in X}$ of probability distributions $t^x$ on $Y$. The serial and parallel composition of transformations, which we denote $\circ$ and $\times$, respectively, are defined in the obvious way (see e.g.~\cite{Hou21} Example 1.1.9). This theory is also known as $\mathbf{FinStoch}$, cf.~Ref.~\cite{Fritz20Synthetic}. \end{Example}

Though we defined $\QIT$ using the language of Hilbert spaces and operators, it might be definable (up to a suitable notion of isomorphism) in ways which do not reference an operator-algebraic formalism. Indeed, it has been shown \cite{Hard01,Chir11} that there does exist such alternative definitions of $\QIT$, and in this sense the operator-algebraic formalism for quantum theory should be considered a choice of `coordinate system' for the theory, rather than identified with the theory itself.   

This is true more generally for \myuline{any} compositional theory: Even when a priori presented using specific mathematical objects, the theory may admit alternative presentations which do not reference those objects. This circumstance prompts the idea that some concepts in a compositional theory might be definable \emph{intrinsically}, using the compositional structure alone. Such concepts will necessarily be meaningful across different representations of the same theory, and also necessarily interpretable in every compositional theory: 

\begin{Definition}[Compositional Concepts, Informally] \label{def:Comp}
In a compositional theory, a \emph{compositional concept} is one that can be defined by reference only to the composition of systems $\otimes$, serial and parallel composition of transformations $\circ$ and $\otimes$, the distinguished trivial system $\triv$, and the distinguished identity and discarding transformations $\id_\cH$ and $\tr_\cH$.
\end{Definition}

For readers versed in mathematical logic, it is possible to make the definition of `compositional concepts' more precise by introducing an appropriate first-order language apt for symmetric monoidal categories (in particular with binary symbols $\circ$, $\otimes$ and constant symbol $\triv$). A compositional concept is then formally a predicate in this language.\footnote{\label{foot:FormLang} Let us however point out that such a formal language should have a built-in way of distinguishing e.g.~in $\QIT$ between the systems $\C^2 \otimes \C^2$ and $\C^4$, which are isomorphic and thus a priori compositionally indistinguishable (see \cref{foot:Iso}). This can be achieved for example by giving systems unique names and only allowing systems with distinct names to be composed. Moreover, we might wish to include elements of sets theory into the formal language, if for instance we wish to define certain concepts as compositional which involve infinite families of transformations.} 

Rather than being explicit about such a formal language, we find it suitable to here provide instead a selection of examples which should indicate what it means to be a compositional concept.

\begin{Example}[Reversibles] \label{ex:Revers}
We call a transformation $\Lambda$ from $\cH$ to $\cK$ \emph{reversible} if it has a left-inverse, i.e.~if there exists a transformation $\Lambda^-$ from $\cK$ to $\cH$ such that $\Lambda^- \circ \Lambda = \id_\cH$. Being a reversible transformation is a compositional concept.\footnote{Formally, reversibility would be a predicate with one free variable, $\Lambda$, namely the predicate $Rev(\Lambda)$  given by
\begin{align*}
    \exists \, \Lambda^- : \big(\up{Dom}(\Lambda^-)= \up{Cod}(\Lambda)\big)  \land \big(\up{Cod}(\Lambda^-)= \up{Dom}(\Lambda)\big) \land \big(\Lambda^- \circ \Lambda = \id_{\up{Dom}(\Lambda)}\big)
    \end{align*}
    Here, `$\up{Dom}$' and `$\up{Cod}$' are unary function symbols supposed to assign domain and codomain, respectively, to a transformation.} \end{Example}

\begin{Example}[Isomorphisms and Automorphisms]
An \emph{isomorphism}, i.e.~a transformation $\Xi$ from $\cH$ to $\cK$ for which there exists a two-sided inverse, is a compositional concept. So is an \emph{automorphism}, i.e.~an isomorphism from a system $\cH$ to itself.
\end{Example}

\begin{Example}[Dimensionality] \label{ex:Dim}
The systems in $\QIT$ are vector spaces and thus admit a notion of one system having larger dimension than another. The dimension of a vector space is however defined in terms of maximal sets of linearly independent vectors, and this concept is not evidently compositional as it references mathematical structure internal to the systems. Nevertheless, the concept that $\dim \cH \leq \dim \cK$ \myuline{can} be compositionally defined: this condition holds precisely if there exists a reversible transformation from $\cH $ to $\cK$, and reversibility is compositional by \cref{ex:Revers}.\footnote{Formally, this concept is a predicate with two free variables, namely the predicate $Lar(\cH, \cK)$ given by 
\begin{align*}
    \exists \, \Lambda :  Rev(\Lambda) \land \big(\up{Dom}(\Lambda) = \cH \big) \land \big(\up{Cod}(\Lambda) = \cK \big) 
    \end{align*}
    where $Rev(\Lambda)$ is the predicate from before expressing reversibility of $\Lambda$.}
    
    It follows that also e.g.~the notion of a quantum system $\cH$ being $2$-dimensional can be compositionally defined, indeed it means that $\dim \triv \leq \dim \cH$, that $\dim \cH \nleq \dim \triv$, and that $\dim \cH \leq \dim \cH'$ for any other system $\cH'$ with those two properties. By induction, any concrete finite dimensionality can be defined as a compositional concept. (This is not to say that every compositional theory will actually \myuline{exhibit} systems of every concrete dimension.)  \end{Example}

\begin{Example}[States]
 The notion of a \emph{state} is compositional, since states are definable as transformations from the trivial system $\triv$. Specifically, in $\QIT$, a state on the system $\cH$ is usually defined as a \emph{density operator} $\varrho \in \End{\cH}$ (i.e.~a positive operator with unit trace), but density operators $\varrho$ can be identified with quantum channels from $\triv$ to $\cH$ since these are precisely the linear maps $ \C (\cong \End{\triv} ) \to \End{\cH}$ given by $z \mapsto z \varrho$. (By abuse of notation, we will always identify states in $\QIT$ with density operators.)  In the compositional theory $\CIT$, this notion of state also translates to the intended notion, indeed a transformation in $\CIT$ from the one-element set $\triv$ to a general set $X$ corresponds to a probability distribution on $X$. 
 
Provided we have a suitable naming of systems (see \cref{foot:FormLang}), the concept of \emph{product states} $\varrho_1 \otimes \varrho_2$ on a system $\cH_1 \otimes \cH_2$ is also compositional. In the theory $\CIT$, it translates to stochastic independence. 
 \end{Example}

A quantum state is called \emph{pure} if its density operator is a rank-one projection, i.e.~of the form $\ketbra{\psi}$ for some unit vector $\ket{\psi}$. Pure states are in fact isometric channels (since an isometry $V: \triv \to \cK$ is essentially a unit vector $\ket{\psi} \in \cK$). In \cref{subsec:Stinespring} we shall see that the concept of being an isometric channel (\cref{ex:Isom}) --- which is a priori defined using operator-algebra --- is in fact also a compositional concept, albeit in disguise (\cref{rem:IsoMeansPure}). In particular the concept of pure states is compositional.  \\

Now, the intuitive understanding of what it means for a concept to be \myuline{operational} subsumes all concepts which are compositional. There are however concepts which are operational in the conventional, intuitive sense and which seem to fall outside the scope of compositional structure. 

Two such examples are the concepts of \emph{probability} and \emph{measurement} (another example is that of \emph{separable} states, i.e.~convex combinations of product states). Formally, given a set $Y$, it is natural to define in $\QIT$ a `probability distribution on $Y$' as a quantum state on the system $\C^Y$ which is diagonal in the computational basis, i.e.~is of the form $\sum_{y \in Y} p(y) \ketbra{y}$ for some function $p:Y \to [0,1]$. Likewise, a `measurement on $\cH$ with outcomes in $Y$' may be defined as a quantum channel $M: \End{\cH}\to \End{\C^Y}$ such that for every state $\varrho$ on $\cH$, the state $M(\varrho)$ is a probability distribution on $Y$. (As such, being a probability is actually a special case of being a measurement.) The concepts of `probability distribution' and `measurement' are instrumental to the operational structure of quantum theory, but they are not compositional in our sense.\footnote{\label{foot:Iso}This can be shown rigorously by formalising the idea that every compositional concept is `invariant under conjugation by automorphisms', i.e.~if we pick an automorphism $\alpha_{\cH}$ of each system $\cH$, in a coherent way such that $\alpha_{\cH_1 \otimes \cH_2} = \alpha_{\cH_1} \otimes \alpha_{\cH_2}$, then for any predicate $C(\Lambda_1, \ldots, \Lambda_n)$ with transformations $\Lambda_1, \ldots, \Lambda_n$ as free variables, the sentence $\forall \Lambda_1, \ldots, \Lambda_n : \big( C(\Lambda_1, \ldots, \Lambda_n) \Leftrightarrow C(\alpha_{\up{Cod}(\Lambda_1)} \circ \Lambda_1 \circ \alpha^{-1}_{\up{Dom}(\Lambda_1)}, \ldots, \alpha_{\up{Cod}(\Lambda_n)} \circ \Lambda_n \circ \alpha^{-1}_{\up{Dom}(\Lambda_n)}) \big)$ is provable. On the other hand, the above defined concept of measurement is not invariant when an automorphism is applied to the codomain $\C^Y$.} Even disregarding this circumstance, it is unclear that the concept of being a \myuline{specific} probability distribution would be compositional, e.g. `the distribution on $\{0,1\}$ which assigns probability $1/3$ to $0$'. This concept is however also operational in the conventional sense, since we may test it (at least to arbitrary precision) by gathering statistics from a large number of independent experiments.

For these reasons, we need to add elements to a compositional theory to capture its full operational structure.

\subsection{Classical Notions in \textbf{QIT}}
\label{subsec:ClasStruct}

In order to capture operational concepts additional to the compositional ones, a `classical structure' should be added to a compositional theory. This structure should describe how to interpret certain information as classical (e.g.~measurement outcomes), and how to assign probabilities to such information. This has been formalised mathematically in the notion of Generalised Probabilistic Theories \cite{Barr07, Barn11, Barn16} and more explicitly merged with category-based compositional structure in Refs.~\cite{Chir10,Hardy10}. Below, we give an alternative and for our purposes simpler treatment which works by virtue of a natural \emph{embedding} of $\CIT$ into the target theory, in our case $\QIT$.\footnote{This idea was mentioned in Ref.~\cite{Hou21}, see Remarks 1.1.12 and 1.1.13 and Section 1.4 therein.} \\

Recall from \cref{ex:CIT} that $\CIT$ is the compositional theory of classical information, whose systems are non-empty finite sets and whose transformations from $X$ to $Y$ are classical channels, i.e.~families $T= (t^x)_{x \in X}$ of probability distributions $t^x$ on $Y$. We now describe how the theory $\CIT$ embeds into the theory $\QIT$. This embedding can then be used to define classical notions within $\QIT$.

Any system $X$ in $\CIT$ yields a system $\what{X} := \C^X$ in $\QIT$, which is the quantum representation of the classical system $X$. We will call a quantum system of the form  $\what{X}$ an \emph{embedded classical system}. (It is not the intention that embedded classical systems can only hold states which are diagonal in the computational basis; being embedded classical simply means being equipped with a notion of which states are considered classical.) Note that this embedding preserves the composition of systems, since $\what{X_1 \times X_2} \cong \what{X}_1 \otimes \what{X}_2$, and that the trivial system in $\CIT$ embeds as the trivial system in $\QIT$.

 Any transformation $T = (t^x)_{x \in X}$ from $X$ to $Y$ in $\CIT$ yields a transformation $\what{T}$ from $\what{X}$ to $\what{Y}$ in $\QIT$, namely the quantum channel given by $\what{T}(A) = \sum_{x \in X, y \in Y} \bra{x} A \ket{x} t^x(y) \ketbra{y}$. We will call a quantum channel of the form $\what{T}$ an \emph{embedded classical channel}. In particular, a classical state is a state of the form $\sum_{y \in Y} p(y) \ketbra{y}$ for a probability distribution $p: Y \to [0,1]$. Again, note that the embedding preserves the serial and parallel composition of transformations, since $\what{ S \circ T } = \what{S} \circ \what{T}$ and $\what{T_1 \times T_2} = \what{T}_1 \otimes \what{T}_2$.  
 
All in all, the above embedding of systems and transformations yields a faithful picture of the compositional theory $\CIT$ inside the compositional theory $\QIT$.  \\ 
 
 The embedding of $\CIT$ into $\QIT$ does \myuline{not} preserve identities, i.e.~the channel $\Delta_X := \what{\id_X}$ is distinct from the channel $\id_{\what{X}}$.\footnote{As observed in \cite{Hou21} Remark 1.1.12, this is not a peculiarity of the particular embedding chosen --- there simply does not exist any identity-preserving embedding (formally: injective strong monoidal functor) from $\CIT$ to $\QIT$.} Explicitly, the channel is given by $\Delta_X(A) = \sum_{x \in X} \bra{x} A \ket{x}  \ketbra{x}$ and we will call it the \emph{dephasing channel on $\what{X}$}. It is easy to verify that an arbitrary quantum channel $\Lambda$ from $\what{X}$ to $\what{Y}$ is embedded classical (i.e.~of the form $\what{T}$) if and only if $\Delta_Y \circ \Lambda \circ \Delta_X = \Lambda$. As such, the dephasing channels can be used to \myuline{identify} $\CIT$ as a sub-theory within $\QIT$. \\
 
 Altogether, the above ideas may be abstracted from the case of quantum information theory to yield a general notion of operational structure, just as we did for compositional structure. 
 
 Given compositional theories $\Theory_0$ and $\Theory$, an \emph{embedding} of $\Theory_0$ into $\Theory$ means a map $\what{\phantom{x}}: \Theory_0 \to \Theory$ which injectively takes systems to systems and transformations to transformations, and which preserves the trivial system, composition of systems, and serial and parallel composition of transformations. Let us moreover say that \emph{$\what{\phantom{x}}$ identifies $\Theory_0$ within $\Theory$} if for any embedded systems $\what{X}$ and $\what{Y}$ in $\Theory$, an arbitrary transformation $\Lambda$ from $\what{X}$ to $\what{Y}$ is embedded (i.e.~of the form $\what{T}$) if and only if it satisfies $\what{\id_Y} \circ \Lambda \circ \what{\id_X} = \Lambda$. Note that this requirement can be considered a natural relaxation of surjectivity: any embedded transformation $\Lambda = \what{T}$ satisfies $\what{\id_Y} \circ \what{T} \circ \what{\id_X} = \what{\id_Y \circ T \circ \id_X} = \what{T}$ by properties of the embedding, but the requirement guarantees that \myuline{every} transformation $\Lambda$ satisfying $\what{\id_Y} \circ \Lambda \circ \what{\id_X} = \Lambda$ (i.e.~every `classical' transformation) is in the image of $\what{\phantom{x}}$. 

 \begin{Definition}[Operational Theories]\label{def:OpTheory}
     An \emph{operational theory} is a pair $(\Theory, \what{\phantom{x}})$, where $\Theory$ is a compositional theory and where $\what{\phantom{x}}: \CIT \to \Theory$ is an embedding which identifies $\CIT$ within $\Theory$.\end{Definition}

Quantum information theory is an operational theory by virtue of the embedding $\what{\phantom{x}}: \CIT \to \QIT$ described above. Classical information theory is an operational theory too, by virtue of the identity embedding $\what{\phantom{x}}: \CIT \to \CIT$. 

In an operational theory, we may successfully capture the previously considered operational concepts which were not compositional. For example, the dephasing\footnote{In a general theory, there may be nothing `dephasing' about a dephasing channel, however we anyway choose to import the term from quantum information theory.} channels $\Delta_X = \what{\id_X}$ can be used to define all sorts of hybrid transformations:  \\

\textbf{Measurements.} If $\what{Y}$ is an embedded classical system and $M$ a transformation from $\cH$ to $\what{Y}$, we say that $M$ \emph{has classical outcomes} if $\Delta_Y \circ M= M$. Alternatively, we call such a channel a \emph{measurement on $\cH$ with outcomes in $Y$}. In the operational theory $\QIT$, it is easy to verify using the Kraus representation of the channel $M$, that if $M$ is a measurement then there exists a $Y$-indexed \emph{POVM (Positive Operator-Valued Measure)} on $\cH$, i.e.~a family $E= (E(y))_{y \in Y}$ of positive operators $E(y)$ on $\cH$ with $\sum_{y \in Y} E(y) = \bone_\cH$, such that
\begin{align} \label{eq:POVM}
M(A) = \sum_{y \in Y} \tr(E(y) A)  \ketbra{y} \quad \text{for all $A \in \End{\cH}$}.
\end{align}

The POVM is unique given the measurement $M$, and any POVM $E$ defines a measurement as above. Thus, we can in $\QIT$ identify measurements with POVMs. \\

\textbf{Ensembles of Channels.} If $\what{X}$ is an embedded classical system and if $\Lambda$ is a transformation from $\cH \otimes \what{X}$ to $\cK$, we say that $\Lambda$ \emph{has classical inputs on $\what{X}$} if $\Lambda \circ (\id_\cH \otimes \Delta_X) = \Lambda$. In $\QIT$, it is easy to verify that this holds precisely if there exists a family $(\Lambda^x)_{x \in X}$ of channels $\Lambda^x : \cH \to \cK$ such that 
\begin{align}
\Lambda (A \otimes B) = \sum_{x \in X} \Lambda^x(A) \bra{x} B \ket{x}   \quad \text{for all $A \in \End{\cH}$, $B \in \End{\what{X}}$} .
\end{align}

Thus, to specify such a channel is to specify an \myuline{ensemble} of channels $\Lambda^x : \cH \to \cK$, indexed by the set $X$. (In particular, a channel $\Lambda : \what{X} \to \cK$ with classical inputs corresponds to an ensemble of states on $\cK$.) We will often use the terminology that $\Lambda$ `measures' of `reads off' the classical value $x$ and applies the according channel $\Lambda^x$. \\

It should be obvious how to use the dephasing channels to define also quantum \emph{instruments} (that is, channels with classical outcomes on some factor in the output system), ensembles of instruments, etc. 

With a definition of operational structure in place, we can state more precisely what we mean by operational concepts:

\begin{Definition}[Operational Concepts, Informally] \label{def:Operational}
    In an operational theory, an \emph{operational concept} is one that can be defined in terms of the compositional structure, the embedding $\what{\phantom{x}}$ and any specific systems and transformations in $\CIT$.\end{Definition}

    This definition subsumes compositional concepts and concepts defined in terms of dephasing channels (in particular, measurements and ensembles as defined above), but is also meant to subsume reference to specific embedded systems and probabilities (e.g.~`the probability distribution on $\{0,1\}$ which assigns probability $1/3$ to $0$'). Again, readers versed in formal logic may entertain the exercise of constructing a more rigorous definition, but we shall here content ourselves with providing a few illustrative examples.

\begin{Example}[Biased Coins]
In $\QIT$, the state on $\C^2$ with density operator $\varrho = \frac{1}{3} \ketbra{0}+ \frac{2}{3}\ketbra{1}$ is an operational concept. Indeed, $\C^2$ can be defined as the embedded system $\what{\{0,1\}}$, and $\varrho$ as the embedded state $\what{s}$, where $s$ is the probability distribution in $\CIT$ on the system $\{0,1\}$ which assigns probability $\frac{1}{3}$ to outcome $0$. Similarly, any other state $p \ketbra{0}+ (1-p) \ketbra{1}$ on $\C^2$ is operational when $p \in [0,1]$ is a definable real number. (The state is in fact \myuline{compositional} when $p=1/2$, since it is then the unique state invariant under all automorphisms on $\C^2$.) 
\end{Example}

\begin{Example}[Maximally Entangled]
The notion of being a \emph{maximally entangled} state on $\C^2 \otimes \C^2$ is operationally definable. Indeed, $\C^2 \otimes \C^2$ is the embedded system $\what{\{0,1\}} \otimes \what{\{0,1\}}$, and a state $\psi$ on this system is maximally entangled precisely if it is pure and if for some local automorphism $\alpha_1 \otimes \alpha_2$ of $\C^2 \otimes \C^2$, the state  $ \psi' :=[\alpha_1 \otimes \alpha_2](\psi)$ satisfies  $(\Delta_{\{0,1\}} \otimes \Delta_{\{0,1\}})(\psi') = \frac{1}{2} \ketbra{00} + \frac{1}{2} \ketbra{11}$.  \end{Example}

As an exercise, we encourage the reader to consider how the \emph{trace distance} $\delta(\varrho, \sigma) := \frac{1}{2} \norm{\varrho-\sigma}_1$ between quantum states $\varrho$ and $\sigma$ may be realised as an operational concept, using that $\delta(\varrho, \sigma)= 2p^*(\varrho, \sigma)-1$ where $p^*(\varrho, \sigma)$ is the optimal probability of distinguishing in a measurement between the states $\varrho$ or $\sigma$ if provided uniformly at random.\\

This concludes our presentation of what we mean by operational concepts in quantum information theory (or indeed any operational theory). As mentioned in the beginning, we circumvent an excruciating formalisation of \cref{def:Operational} by adopting the viewpoint that it is intuitively clear when a concept has been successfully given an operational definition according to this definition.

\subsection{Naimark's Theorem}
\label{subsec:Naimark}

In this subsection, we recall a useful structure-theorem for quantum measurements by M.~Naimark. In general, a measurement which correspond to the POVM $E=(E(y))_{y \in Y}$ is called \emph{projective} if all of the positive operators $E(y)$ are projections, i.e.~satisfy $E(y)^2 = E(y)$. (We do not know whether the notion of being a projective measurement is operational or not, but this is not relevant for our purposes.)  Naimark's theorem \cite{Neum40} states that any measurement can be realised using a projective measurement on a larger system:

\begin{Thm}[Naimark] \label{thm:Naim} For any quantum measurement $M: \End{\cH} \to \End{\what{Y}}$, there exists a Hilbert space $\cK^\up{Nai}$, a \myuline{projective} quantum measurement $M^\up{Nai} : \End{\cH \otimes \cK^\up{Nai}} \to \End{\what{Y}}$ and a pure state $\phi^\up{Nai}$ on $\cK^\up{Nai}$ such that 
\begin{align} \label{eq:Naimark}
M = M^\up{Nai} \circ ( \id_{\cH} \otimes \phi^\up{Nai}).
\end{align}
\end{Thm}

\begin{Remark}
It is not uncommon in statements of Naimark's theorem to take $\cK^\up{Nai}=\C^m$ for some $m \in \N$ and fix $\phi^\up{Nai}$ to be the state $\ketbra{0}$ (this can always be achieved by a unitary rotation). However, this level of specification is irrelevant for our purposes. 
\end{Remark}

\begin{Remark}[Generalisation to Ensembles] \label{rem:NaiEnsemble}
Naimark's theorem can straightforwardly be generalised to measurement \myuline{ensembles}. Specifically, any measurement ensemble $\Lambda : \End{\what{X} \otimes \cH} \to \End{\what{Y}}$ is of the form $\Lambda^\up{Nai} \circ (\id_{\what{X}} \otimes \id_\cH \otimes \phi^\up{Nai})$ for some pure state $\phi^\up{Nai}$ on a system $\cK^\up{Nai}$ and some ensemble $\Lambda^\up{Nai}$ of \myuline{projective} measurements on $\cH \otimes \cK^\up{Nai}$. The proof is by induction on the cardinality $\abs{X}$, with the case $\abs{X}=1$ corresponding to Naimark's original theorem. 
\end{Remark}

\subsection{Isometric Channels and Stinespring's Dilation Theorem}

\label{subsec:Stinespring}

Recall from \cref{ex:Isom} that an \emph{isometric} quantum channel $\Sigma$ from $\cH$ to $\cK$ is one of the form $A \mapsto VAV^*$ for some isometry $V: \cH \to \cK$. Note that an isometric channel determines its underlying isometry up to a phase, that is, the isometries $V$ and $e^{i \theta} V$ give rise to the same isometric channel for any $\theta \in [0, 2 \pi)$, but this is the only ambiguity. In particular, pure states on $\cK$ (which are simply isometric channels from $\triv$ to $\cK$) determine their underlying unit vectors $\ket{\psi} \in \cK$ up to a phase. We will denote pure states by letters $\psi, \phi, \chi, \ldots$, and corresponding vector representatives by $\ket{\psi}, \ket{\phi}, \ket{\chi}, \ldots$ whenever we write out equations whose content is insensitive to the choice of phase.\\

We start by observing the following useful fact about isometric channels (recall from \cref{ex:Revers} the notion of reversibility):

\begin{Lem} \label{lem:IsoRev}
  Every isometric channel $\Sigma: \End{\cH} \to \End{\cK}$ is reversible.
\end{Lem}

\begin{proof}
Let $V$ be an isometry representing $\Sigma$. It is tempting to define a left-inverse $\Sigma^-: \End{\cK} \to \End{\cH}$ as $B \mapsto V^* B V$. However, even though $\Sigma^-$ reverses the action of $\Sigma$, it does not define a quantum channel as it is not trace-preserving (except when $V$ is unitary). Thus, we instead pick an arbitrary state $\tau$ on $\cH$ and define $\Sigma^-$ by $\Sigma^-(B)= V^* B V + \tau \tr(\sqrt{\bone_\cH - VV^*} B \sqrt{\bone_\cH - VV^*})$. This map is completely positive and trace-preserving, thus a quantum channel, and it satisfies $\Sigma^-(\Sigma(A))= A$ for every $A \in \End{\cH}$. 
\end{proof}

Isometric channels are famous for their appearance in W.~Stinespring's ubiquitous dilation theorem \cite{Stine55}, which we now recall in the context of quantum information theory.

Given any quantum channel $\Lambda : \End{\cH} \to \End{\cK}$, a \emph{(one-sided) dilation} \cite{Chir10,Hou21} of $\Lambda$ is a quantum channel $\Phi: \End{\cH} \to \End{\cK \otimes \cE}$ such that $(\id_\cK \otimes \tr_\cE) \circ \Phi = \Lambda$. We call $\cE$ the \emph{environment} of the dilation $\Phi$. The concept of dilation is compositional, as it is definable in terms of serial and parallel composition and the distinguished discarding transformations. In general, a dilation can be interpreted as describing side-information which escapes to a hidden system (namely the environment).

A dilation $\Sigma$ of $\Lambda$ is called a \emph{Stinespring dilation} if is it isometric.

\begin{Thm}[Stinespring] \label{thm:Stine}
Every quantum channel $\Lambda: \End{\cH} \to \End{\cK}$ has a Stinespring dilation. It is unique up to transformations on the environment, in the sense that if $\Sigma: \End{\cH} \to \End{\cK \otimes \cE}$ and $\Sigma': \End{\cH} \to \End{\cK \otimes \cE'}$ are both Stinespring dilations of $\Lambda$, then there exists a quantum channel $\Gamma : \End{\cE} \to \End{\cE'}$ such that $\Sigma' = (\Gamma \otimes \id_\cK) \circ \Sigma$.   
\end{Thm}

\begin{Remark}
    If $\dim \cE \leq \dim \cE'$, then the channel $\Gamma$ can be taken isometric --- in fact, this is how the uniqueness clause is usually stated. The statement in \cref{thm:Stine} follows from the usual uniqueness statement by \cref{lem:IsoRev}.
\end{Remark}
	
\begin{Remark}
In the case where $\cH = \triv$, the channel $\Lambda$ corresponds to a density operator $\varrho \in \End{\cK}$ and its Stinespring dilations are pure states whose corresponding density operators $\psi$ satisfy $(\id_\cH \otimes \tr_\cE)(\psi) = \varrho$. These Stinespring dilations are more commonly called \emph{purifications} of $\varrho$. \end{Remark}

Stinespring's dilation theorem implies the following property, which in particular lends itself to the interpretation that Stinespring dilations leak maximal information to the environment:

\begin{Thm}[Completeness of Stinespring Dilations] \label{thm:StineComp}
Let $\Lambda : \End{\cH} \to \End{\cK}$ be a quantum channel. A Stinespring dilation $\Sigma : \End{\cH} \to \End{\cK \otimes \cE}$ is a \emph{complete dilation} \cite{Hou21} of $\Lambda$, meaning that if $\Phi: \End{\cH} \to \End{\cK \otimes \cF}$ is any dilation of $\Lambda$, then there exists a quantum channel $\Gamma : \End{\cE} \to \End{\cF}$ such that $\Phi = (\Gamma \otimes \id_\cK) \circ \Sigma$.
\end{Thm}

\begin{proof}
The channel $\Phi: \End{\cH} \to \End{\cK} \otimes \End{\cF}$ has a Stinespring dilation $\Sigma' : \End{\cH} \to \End{ \cK \otimes \cF \otimes \cG}$. Now, $\Sigma'$ is also a Stinespring dilation of $\Lambda$ (with environment $\cF \otimes \cG$). Therefore, by \cref{thm:Stine}, there exists a channel $\tilde{\Gamma}: \End{\cE} \to \End{\cF \otimes \cG}$ such that $\Sigma' = (\tilde{\Gamma} \otimes \id_\cK) \circ \Sigma$. Tracing out $\cG$, we see that $\Phi = (\Gamma \otimes \id_\cK) \circ \Sigma$, with $\Gamma = (\id_\cF \otimes \tr_\cG) \circ \tilde{\Gamma}$.
\end{proof}

\begin{Remark}[General Completeness] In the theory $\CIT$, it is also true that every channel has a complete dilation; it is given by copying the input and output to the environment. See Ref.~\cite{Hou21} for further details and a systematic study of dilational completeness.
\end{Remark}

\begin{Cor}[Isometric implies Pure.] \label{Cor:StinePure}
A isometric quantum channel $\Sigma : \End{\cH} \to \End{\cK}$ is \emph{dilationally pure} \cite{Chir14pure, Hou21}, meaning that every dilation of $\Sigma$ is of the form $\Sigma \otimes \varrho$ for some state $\varrho$.    
\end{Cor}

\begin{proof}
To say that a channel is dilationally pure is to say that it is a complete dilation of itself. However if $\Sigma$ is isometric then it is evidently a \myuline{Stinespring} dilation of itself, and therefore by \cref{thm:StineComp} complete. 
\end{proof}

\begin{Remark}[Isometric means Pure] \label{rem:IsoMeansPure}
 It can be shown quite easily that the dilationally pure quantum channels are \myuline{precisely} those which are isometric (\cite{Hou21} Cor.~2.3.32) --- the idea behind the proof is that if $\Sigma$ requires more than one Kraus operator in its Kraus representation, then we may construct from this representation a dilation which is not of the form $\Sigma \otimes \varrho$. Dilational purity is evidently a compositional concept, and isometric channels may thus, perhaps surprisingly, be given a compositional definition. 
\end{Remark}

\section{The Standard Perception of Quantum Self-Testing}
\label{sec:Standard}

In this section, we recall the standard operator-algebraic perception of quantum self-testing and discuss the need for an alternative, operational perspective. The section is divided into three subsections.

In \cref{subsec:Aim} we describe in broad terms what quantum self-testing aims to say. We shall here discuss Bell-scenarios, behaviours, and (quantum) strategies of a general kind. This discussion can be cast in operational terms cf.~\cref{sec:Operational}, and the broad aim of self-testing may as such be interpreted in general operational theories beyond quantum information theory. In \cref{subsec:Standard} we give the standard definition of \myuline{quantum} self-testing \cite{MY04,MYS12,SB19}, which is to say we lay out how the aim of self-testing is conventionally formalised for the theory $\QIT$. The fundamental problem addressed by our paper is that this standard formalisation \myuline{leaves} the realm of operational structure; it is cast in operator-algebraic language that does not allow for natural generalisations to other theories, formalisms or settings beyond the Bell-scenario.

Throughout, we restrict attention to the case of self-testing in bipartite Bell-scenarios (i.e.~with two parties, $\sfA$ and $\sfB$), but the entire discussion easily generalises.

\subsection{Bell-Scenarios and the Aim of Self-Testing}
\label{subsec:Aim}

\makebox[\linewidth][s]{By a \textbf{\emph{(bipartite) Bell-scenario}} we mean a quadruple of finite non-empty sets}, $\fB = (X_\sfA, X_\sfB, Y_\sfA, Y_\sfB)$. The interpretation is that two parties, $\sfA$ and $\sfB$, are separated and that party ${\sfP} \in \{\sfA, \sfB\}$ receives an input $x_{\sfP}$ from the set $X_{\sfP}$ and is expected to return, based on that input and without communicating with the other party, an output $y_{\sfP}$ from the set $Y_{\sfP}$. We use throughout the notation $X := X_\sfA \times X_\sfB $ and $Y := Y_\sfA \times Y_\sfB$, and denote pairs $(x_\sfA, x_\sfB)$ and $(y_\sfA, y_\sfB)$ by $x$ and $y$, respectively.

However the parties $\sfA$ and $\sfB$ proceed, it is natural to summarise their input-output behaviour by specifying for each input-pair $x=(x_\sfA, x_\sfB)$ a probability distribution $P^x$ which assigns to each output-pair $y=(y_\sfA, y_\sfB)$ the probability $P^x(y)$ of output $y$ on input $x$. Any such collection $P= (P^x)_{x \in X}$ of probability distributions on $Y = Y_\sfA \times Y_\sfB$ indexed by $X = X_\sfA \times X_\sfB$ will be called a \textbf{\emph{behaviour}} for the Bell-scenario $\fB$.\footnote{The term `behaviour' is due to Cirelson \cite{Cir93}. The term `correlations' is used synonymously in the literature.} It is worth observing that, as such, a behaviour $P$ is nothing but a transformation in the theory $\CIT$ from the set $X$ to the set $Y$. In a general operational theory, we may draw it as
\begin{align}
\myQ{0.7}{0.7}{& \push{\what{X}_\sfA} \qw & \multigate{1}{P} & \push{\what{Y}_\sfA} \qw & \qw \\ & \push{\what{X}_\sfB} \qw & \ghost{P} & \push{\what{Y}_\sfB} \qw & \qw } \quad ,
\end{align}
recalling from \cref{subsec:ClasStruct} that $\what{X}_{\sfP}$ and $\what{Y}_{\sfP}$ are the embeddings of the classical systems $X_{\sfP}$ and $Y_{\sfP}$ into the theory (in $\QIT$, they are given by $\what{X}_\sfP = \C^{X_\sfP}$ and $\what{Y}_\sfP = \C^{Y_\sfP}$). Intuitively, any theory of physics will specify which behaviours can be realised by non-communicating parties in the Bell-scenario $\fB$. Indeed, a given theory will determine a viable class of `strategies' (or `protocols') for obtaining outputs based on inputs, and to each such strategy will be associated a behaviour. Roughly speaking, self-testing is concerned with deducing characteristics of a strategy from its behaviour alone. \\

The notions of `strategy' and `realisable behaviour' can be made precise by referring to the operational structure of the theory. Namely, a behaviour can be realised in a given theory precisely if, as a channel, it can be built using serial and parallel compositions of channels allowed by the theory, in such a way that no information flows from $\sfA$'s input to $\sfB$'s output or from $\sfB$'s input to $\sfA$'s output.

For example, a channel of the form 
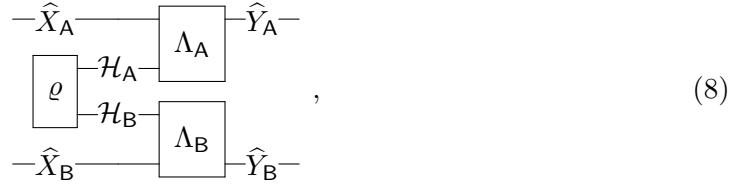
\begin{align} \label{eq:Bell}
	\myQ{0.7}{0.7}{
		& \push{\what{X}_\sfA}  \qw 	& \qw  & \multigate{1}{\Lambda_\sfA} & \push{\what{Y}_\sfA}  \qw & \qw   \\
		& \Nmultigate{1}{\varrho}  & \push{\cH_\sfA} \qw & \ghost{\Lambda_\sfA} &  \\
		& \Nghost{\varrho} & \push{\cH_\sfB} \qw  & \multigate{1}{\Lambda_\sfB}  \\
	&   \push{\what{X}_\sfB} \qw  		& \qw  & \ghost{\Lambda_\sfB} & \push{\what{Y}_\sfB}  \qw & \qw \\
	}, 
	\end{align}
where $\varrho$ is a state on a bipartite system $\cH_\sfA \otimes \cH_\sfB$ and where $\Lambda_\sfA$ and $\Lambda_\sfB$ are arbitrary channels, constitutes a realisable behaviour. The concrete circuit of channels, i.e.~the triple $(\varrho, \Lambda_\sfA, \Lambda_\sfB)$, constitutes a strategy which realises the behaviour. The circuit physically corresponds to the situation in which the parties $\sfA$ and $\sfB$ share a state which they process locally in an input-dependent way to obtain their outputs. 

On the other hand, a circuit of the form 
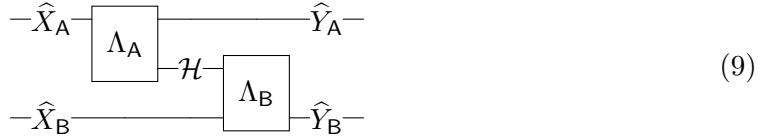
\begin{align}
    \myQ{0.7}{0.7}{
	& \push{\what{X}_\sfA}  \qw    &\multigate{1}{\Lambda_\sfA}   & \qw & \qw &   \push{\what{Y}_\sfA} \qw & \qw  \\
	&                   & \Nghost{\Lambda_\sfA}  & \push{\cH} \qw & \multigate{1}{\Lambda_\sfB} & \\
	&  \push{\what{X}_\sfB}  \qw & \qw  & \qw &  \ghost{\Lambda_\sfB}   &  \push{\what{Y}_\sfB} \qw  & \qw  \\
} 
\end{align}
does not a priori lead to a realisable behaviour, because information flows from $\sfA$'s input to $\sfB$'s output. 

A general strategy for the Bell-scenario $\fB=(X_\sfA, X_\sfB, Y_\sfA, Y_\sfB)$ could be defined as any circuit of transformations which has open inputs corresponding to the embeddings $\what{X}_{\sfA}$ and $\what{X}_\sfB$, open outputs corresponding to the embeddings $\what{Y}_{\sfA}$ and $\what{Y}_{\sfB}$, and such that in the directed acyclic graph describing the circuit structure there are no paths from $\what{X}_\sfA$ to $\what{Y}_\sfB$ or from $\what{X}_\sfB$ to $\what{Y}_\sfA$. It can however be shown (\cite{Hou21} Example 4.1.27) that every behaviour realisable by a general strategy is realisable by a triple-strategy $(\varrho, \Lambda_\sfA, \Lambda_\sfB)$ as depicted in \cref{eq:Bell}. Since the total input-output behaviour is a classical embedded channel, it may additionally be assumed that $\Lambda_{\sfP}$ has classical inputs on $\what{X}_{\sfP}$ and classical outputs on $\what{Y}_{\sfP}$ (cf.~\cref{subsec:ClasStruct}). 

Later, in \cref{subsec:Causal}, we shall make a point out of reconsidering general strategies described by arbitrary circuits, but in the standard description of self-testing one restricts attention to triple-strategies and the term `strategy' is used exclusively for those:

\begin{Definition}[Strategies] \label{def:OpStrat} 
 Consider (cf.~\cref{def:OpTheory}) an operational theory $(\Theory, \what{\phantom{x}})$, and let $\fB=(X_\sfA, X_\sfB, Y_\sfA, Y_\sfB)$ be a Bell-scenario. A \emph{strategy for $\fB$} is a triple $S= (\varrho, \Lambda_\sfA, \Lambda_\sfB)$, where $\varrho$ is a state on a bipartite system $\cH := \cH_\sfA \otimes \cH_\sfB$ in $\Theory$ and where, for ${\sfP} \in \{\sfA, \sfB\}$, $\Lambda_{\sfP}$ is a transformation in $\Theory$ from $\what{X}_{\sfP} \otimes \cH_{\sfP}$ to $\what{Y}_{\sfP}$ with classical inputs on $\what{X}_{\sfP}$ and with classical outputs. The \emph{behaviour associated to $S$} is the classical channel $P=(P^x)_{x \in X}$ defined by \cref{eq:Bell}.\end{Definition}

\begin{Example}[Strategies in $\CIT$: Local Hidden Variables and Bell's Theorem]
   In the theory of classical information, $\CIT$, a strategy for $\fB$ is a triple $(p, L_\sfA, L_\sfB)$ where $p$ is a probability distribution on some product set $H=H_\sfA \times H_\sfB$, and where $L_{\sfA} : X_{\sfA}  \times H_{\sfA}  \to Y_{\sfA}$ and $L_{\sfB} : X_{\sfB}  \times H_{\sfB} \to Y_{\sfB}$ are classical channels. Bell's theorem \cite{Bell64} is the statement that certain behaviours realisable by \myuline{quantum} strategies are not reproducible by such classical strategies. This result remains a landmark in our understanding of quantum theory, as it rules out an underlying description of certain experiments in terms of `hidden' random variables (encoded by the state $p$) and locally applied random functions (encoded by $L_\sfA$ and $L_\sfB$).\end{Example}

\begin{Example}[Strategies in $\QIT$: Operator-Algebraic Representations] \label{ex:QStrat} In quantum information theory, $\QIT$, a strategy for $\fB$ is a triple $(\varrho, \Lambda_\sfA, \Lambda_\sfB)$ where $\varrho$ is a quantum state (density operator) on a tensor-product Hilbert space $\cH = \cH_\sfA \otimes \cH_\sfB$, and where, for $\sfP \in \{\sfA, \sfB\}$, the channel $\Lambda_{\sfP} :\End{\what{X_{\sfP}} \otimes  \cH_{\sfP}}  \to \End{\what{Y_{\sfP}}}$ is an $X_{\sfP}$-indexed ensemble of measurements on $\cH_{\sfP}$ with outcomes in $Y_{\sfP}$. Each measurement in the ensemble can be identified with a POVM, cf. \cref{subsec:ClasStruct}; as such, the channel $\Lambda_{\sfP}$ may be identified with a collection $E_{\sfP} = (E^{x_{\sfP}}_{\sfP})_{x_{\sfP} \in X_{\sfP}}$ of POVMs $E^{x_{\sfP}}_{\sfP} = (E^{x_{\sfP}}_{\sfP}(y_{\sfP}))_{y_{\sfP} \in Y_{\sfP}}$ on $\cH_{\sfP}$. An equivalent definition of quantum strategies is thus as triples $(\varrho, E_\sfA, E_\sfB)$ with $\varrho$ a density operator and $E_\sfA$ and $E_\sfB$ ensembles of POVMs; this definition of quantum strategies is indeed more common in the literature on self-testing \cite{SB19}.

Given a quantum strategy $S = (\varrho, E_\sfA, E_\sfB)$ and an input-pair $x=(x_\sfA, x_\sfB)$, let us denote by $E^x$ the POVM on $\cH = \cH_\sfA \otimes \cH_\sfB$ defined by $E^x(y) = E^{x_\sfA}_\sfA(y_\sfA) \otimes E^{x_\sfB}_\sfB(y_\sfB)$ for output-pairs $y=(y_\sfA, y_\sfB)$. The POVM ensemble $(E^x)_{x \in X}$ then represents the channel $\Lambda_\sfA \otimes \Lambda_\sfB$, and the behaviour associated to the strategy $S$, i.e. the channel $P=(P^x)_{x \in X}$ defined by \cref{eq:Bell}, is given by 
\begin{align} \label{eq:bhv}
    P^x(y) = \tr(E^x(y) \varrho)  \quad \text{for $x \in X$, $y \in Y$.}
\end{align}
(If $\varrho =: \psi$ is a pure state, this expression reduces to $\bra{\psi} E^x(y) \ket{\psi}$.) \end{Example}

We can now describe what quantum self-testing aims to do. Every quantum strategy for the Bell-scenario $\fB$ gives rise to a behaviour, as expressed by \cref{eq:Bell}, or equivalently \cref{eq:bhv}. In its broadest possible definition, quantum self-testing is the enterprise of deriving statements about a strategy from knowing just its behaviour. Mathematically, if 
\begin{align} \label{eq:strabhv}
    B : \Strat{\fB} \to \Bhv{\fB}
    \end{align}
    denotes the map which associates behaviours to strategies, the question is what can be said about strategies in the pre-image $B^{-1}(\{P\}) \subseteq \Strat{\fB}$ given a quantum behaviour $P \in \Bhv{\fB}$. In a narrower sense, self-testing refers to the situation in which one can derive from the behaviour essentially uniquely what the strategy is. This is indeed the sense in which self-testing was introduced \cite{MY04} and standardised \cite{MYS12}, and it is what we mean throughout when referring to the standard, or conventional, notion of self-testing. Specifically, it is concerned with the situation in which there exists a fixed \emph{canonical} strategy $\tilde{S} = (\tilde{\varrho}, \tilde{\Lambda}_\sfA, \tilde{\Lambda}_\sfB)$ in $B^{-1}(\{P\})$ such that every strategy $S=(\varrho, \Lambda_\sfA, \Lambda_\sfB)$ in $B^{-1}(\{P\})$ is `reducible' to the canonical strategy $\tilde{S}$. One says in this case that the behaviour $P$ \emph{self-tests} the strategy $\tilde{S}$. As such, to define self-testing properly, one needs only define what is meant by one strategy being `reducible' to another. \\
    
Now, there is nothing whatsoever about the aim of self-testing which seems to rely on quantum theory: the broad question of determining the form of strategies from their behaviour is meaningful in every operational theory. Nevertheless, any currently known definition of what it means for a strategy $S$ to be \emph{reducible} to another strategy $\tilde{S}$ is cast in terms of their operator-algebraic constituents described in \cref{ex:QStrat}. We shall present the standard definition in \cref{subsec:Standard} below. It is not clear how to translate it to operational language, and this is the fundamental issue which we tackle in our paper and resolve by providing a new and purely operational definition.

\subsection{The Operator-Algebraic Formulation of Quantum Self-testing}
\label{subsec:Standard}

We now present the standard definition of reducibility among quantum strategies and quantum self-testing  \cite{MYS12,SB19}. Some authors \cite{Kan17} consider a slightly alternative reducibility definition (in fact in line with Ref.~\cite{MY98}), but it too relies on operator-algebra. We discuss in \cref{rem:SelfAlt} why we choose to make operational the standard definition rather than such alternative definitions. \\

Recall from \cref{ex:QStrat} that a quantum strategy for the Bell-scenario $\fB$ is formally a triple of state and channels, $(\varrho, \Lambda_\sfA, \Lambda_\sfB)$, or equivalently a triple of state and POVM-ensembles, $(\varrho, E_\sfA, E_\sfB)$. Let us employ the following terminology:

\begin{Definition}[Types of Quantum Strategies]
A quantum strategy $S=(\varrho, E_\sfA, E_\sfB)$ is called 
\begin{itemize}
    \item \emph{pure-state} if the state $\varrho$ is pure;
    \item \emph{projective} if all the POVM-elements $E^{x_\sfP}_\sfP(y_\sfP)$ are projections;
    \item \emph{full-rank} if the marginal states $\varrho_\sfA$ and $\varrho_\sfB$ have full ranks on $\cH_\sfA$ and $\cH_\sfB$, respectively.
\end{itemize}
\end{Definition}

It is worth observing that \myuline{every} behaviour realisable by a quantum strategy can be realised by some projective pure-state strategy, by picking any initial strategy, Naimark-extending all of its measurements (cf.~\cref{rem:NaiEnsemble}) and purifying its state.\footnote{However, it is fully possible that not every quantum realisable behaviour can be realised by a projective \myuline{full-rank} strategy.} \\

One would expect the reducibility relation to be a pre-order, $\geq_{red.}$, on the class of all strategies, $\Strat{\fB}$, with $S \geq_{red.} \tilde{S}$ signifying that $S$ is reducible to $\tilde{S}$. In reality, however, it is quite standard in expositions on self-testing to restrict attention to projective strategies \cite{SB19}, and sometimes even projective pure-state strategies \cite{MYS12}. The sufficiency of this restricted domain of attention has been explained by various arguments which we now briefly discuss. Note that restricting the form of the fixed canonical strategy $\tilde{S}$ is conceptually different from restricting the form of the variable strategy $S$.

As for the canonical strategy $\tilde{S}$, it has been suggested \cite{Goh18} that it is impossible to self-test a strategy which is \myuline{not} pure-state and projective. A formal proof of such a statement would however require the reducibility relation to be already defined for arbitrary strategies, and to the best of our knowledge such a definition has not been given explicitly before. Instead, we adopt the viewpoint that it just so happens that self-testing has only been studied in the case of a projective pure-state canonical strategy, and that consequently a reducibility definition for a general $\tilde{S}$ has not been explored.\footnote{After presenting our new operational version of reducibility, which applies to all strategies, we will prove formally that the corresponding version of self-testing is indeed only possible for pure-state strategies (\cref{cor:SelfPure}). One may take it furthermore projective without loss of generality (\cref{prop:Proj}).}

Restricting the form of the variable strategy $S$, on the other hand, cannot be justified in a similar way; indeed, the entire spirit of self-testing is to characterise arbitrary, unknown strategies with a given behaviour. Nevertheless, a restriction to projective strategies is usually justified by a reference to the possibility of Naimark-extending the measurements \cite{SB19}. Referring additionally to the possibility of purifying the state, $S$ is sometimes restricted to be also pure-state, in line with Ref.~\cite{MYS12}. Now, if one has a particular \myuline{application} of the self-testing phenomenon in mind, it may or may not be valid to consider only projective pure-state strategies, depending on what that application is. We believe, however, that restrictions on the unknown strategy $S$ are in general ill-motivated if one wants to characterise all strategies with a given behaviour.\footnote{Some works do obtain such a characterisation, see e.g.~Ref.~\cite{Kan17}.} In fact, assuming strategies to be always pure-state can lead to downright misleading conclusions. For example, one then readily arrives at the ridiculous statement that entanglement is necessary to produce any non-product behaviour (since every non-entangled pure state is product). Also, one concludes incorrectly that no adversarial environment can hold pre-existing information correlated with the measurement outcomes of the strategy. We do not know that the restriction to projective measurements should cause similar issues.\footnote{After presenting our operational version of reducibility, which applies to all strategies, we will \myuline{prove} that it suffices to consider projective strategies (\cref{prop:Proj}), thus giving substance to this perception.}

For now, we shall refrain from further discussing the quality of the motivations behind these restrictions, and rather state the reducibility and self-testing definitions which we will henceforth unambiguously refer to as the standard definitions (specifically, they appear as Def.~2 in the survey \cite{SB19}). Throughout, we will denote projective POVMs by the letter $\Pi$ rather than $E$.

\begin{Definition}[Operator-Algebraic Reducibility among Quantum Strategies] \label{def:Redattl}
	Let $S= (\varrho, \Pi_\sfA, \Pi_\sfB)$ be a projective strategy and let $\tilde{S}=(\tilde{\psi}, \tilde{\Pi}_\sfA, \tilde{\Pi}_\sfB)$ be a projective pure-state strategy. Let $\psi$ be a purification of $\varrho$, with purifying system $\cP$. We say that \emph{$S$ is reducible to $\tilde{S}$}, written $S \geq_{red.} \tilde{S}$, if there exist Hilbert spaces $\cH^\res_\sfA$, $\cH^\res_\sfB$, isometries $W_{\sfA}: \cH_{\sfA}  \to \tilde{\cH}_{\sfA} \otimes \tilde{\cH}^\res_{\sfA}$ and $W_{\sfB}: \cH_{\sfB}  \to \tilde{\cH}_{\sfB} \otimes \tilde{\cH}^\res_{\sfB}$, and a pure state $\psi^\res$ on $\cH^\res_\sfA \otimes \cH^\res_\sfB \otimes \cP$, such that
	\begin{align} \label{eq:RedAttlComp}
\forall x \in X, y \in Y : \; 	[	W \Pi^{x}(y)   \otimes \bone_\cP ] \ket{\psi} =  \tilde{\Pi}^{x}(y) \tilde{\ket{\psi}} \otimes \ket{\psi^\res},
	\end{align}

where $W= W_\sfA \otimes W_\sfB$, and where $\Pi^x(y)=\Pi^{x_\sfA}_\sfA(y_\sfA) \otimes \Pi^{x_\sfB}_\sfB(y_\sfB)$ and $\tilde{\Pi}^x(y)=\tilde{\Pi}^{x_\sfA}_\sfA(y_\sfA) \otimes \tilde{\Pi}^{x_\sfB}_\sfB(y_\sfB)$ for $x=(x_\sfA, x_\sfB)$ and $y = (y_\sfA, y_\sfB)$.
	\end{Definition}
	
\begin{Remark} The reducibility condition is independent of the choice of purification $\psi$ of $\varrho$, since the state $\psi^\res$ can be changed accordingly.
\end{Remark}
	
	\begin{Definition}[Self-Testing according to Reducibility] \label{def:SelftestAttl} 
		Suppose that $\tilde{S}=(\tilde{\psi}, \tilde{\Pi}_\sfA, \tilde{\Pi}_\sfB)$ is a projective pure-state strategy with behaviour $P$. We say that \emph{$P$ self-tests $\tilde{S}$} if every projective strategy $S$ with behaviour $P$ is reducible to $\tilde{S}$.
	\end{Definition}

The inner logic to this self-testing definition is that any strategy $S$ which \myuline{does} abide to \cref{eq:RedAttlComp} has the same behaviour as $\tilde{S}$; therefore we cannot hope to prove anything stronger about the form of an unknown strategy $S$ with the same behaviour as $\tilde{S}$. The reducibility condition is often put into words by saying that `up to local isometric embeddings, the measurements act on the state in $S$ as if $\tilde{S}$ is appended with an unmeasured residual state $\psi^\res$'.

		\begin{Remark}[Pure-State Strategies]
	  In the case where the state $\varrho$ is pure, with vector-representative $\ket{\psi}$, the reducibility condition $S \geq_{red.} \tilde{S}$ reads
   \begin{align} \label{eq:RedVers1}
       \forall x \in X, y \in Y : \; W \Pi^x(y) \ket{\psi} = \tilde{\Pi}^x(y) \ket*{\tilde{\psi}} \otimes \ket*{\psi^\res}.
   \end{align}
   This is a common way of stating the reducibility condition in treatments which assume pure-state strategies throughout.\end{Remark}

\begin{Remark}[Terminology] \label{rem:RedEquiv}
The term `reducible' is not standard in the literature, indeed the reducibility condition is rarely separated from the self-testing definition and explicitly named. A more common terminology is that `$S$ is \emph{equivalent} to $\tilde{S}$', however the relation $\geq_{red.}$ is a not symmetric and therefore no equivalence relation. We shall in \cref{subsec:AssSim} discuss how the relation $\geq_{red.}$ (or rather, our operational version of it) \myuline{generates} an equivalence relation. This generated equivalence relation is in a sense very close to the original relation, and thus provides some substance to the term `equivalent'. See \cref{thm:AssSym} and \cref{rem:SelfEquiv}. \end{Remark}

\begin{Remark}[Other Self-Testing Definitions] \label{rem:SelfAlt}
   Some authors consider an alternative self-testing definition in which statements are made separately about the state and measurements (see e.g.~\cite{Kan17} Proposition A.3 and \cite{SB19} Definition 3). In our opinion, it is better to consider statements which, like \cref{def:Redattl}, mixes state and measurements, since statements separately about measurements will not generalise adequately to an approximate setting (essentially because the marginal states $\varrho_\sfA$ and $\varrho_\sfB$ may be supported on subspaces corresponding to small eigenvalues). To our knowledge, \cref{def:Redattl} above is the only definition with this property which currently exists. Therefore, this is the definition we choose to modify into an operational definition. 
\end{Remark}

\subsection{Summary and Critique of the Standard Approach}
\label{subsec:SumCrit}

In summary, the conventional narrative of quantum self-testing in a Bell-scenario $\fB$ is as follows:  
\vspace{.3cm}
\begin{itemize}
\item An input-output behaviour for the Bell-scenario is identified with a collection of probability distributions $P=(P^x)_{x \in X}$.
\item Any imaginable strategy for producing an input-output behaviour without communication can be reduced to a standardised strategy (\cref{def:OpStrat}), in which each party makes an input-dependent measurement on a shared state and returns as output the measurement outcome. 
    \item A standardised strategy can, by the operator-algebraic formalism of quantum theory, be identified with a triple $S=(\varrho, E_\sfA, E_\sfB)$ of operator-algebraic objects as prescribed in \cref{ex:QStrat}. 
    \item Quantum self-testing refers to the situation in which every strategy $S$ with behaviour $P$ can be \emph{reduced} to a single `canonical' strategy $\tilde{S}$. Formally, reducibility is a relation between strategies stated in operator-algebraic language (\cref{def:Redattl}). The condition which defines the relation is motivated by the observation that any strategy $\tilde{S}$ can be subjected to a number of behaviour-preserving transformations, e.g.~tensoring with unmeasured states, local unitary rotations and local embeddings into larger spaces.
\end{itemize}

As discussed in \cref{subsec:Standard}, the reducibility relation $S \geq_{red.} \tilde{S}$ is usually only stated and enforced when $S$ is projective and $\tilde{S}$ is projective and pure-state. The significance of this is, at best, unexplained, but the message of our paper is of course that the operator-algebraic narrative of self-testing is deceiving in its entirety; it fails to give operational meaning to self-testing in line with \cref{sec:Operational}. \\

In the remainder of the paper, we set out to recast quantum self-testing in a purely operational language. The definition we propose (\cref{def:OpSelf}) has several benefits over the conventional definition.\\

It will provide a meaning to self-testing interpretable in any operational theory. It will also clarify the significance of restricting to projective strategies (\cref{prop:Proj}), yield a sense in which different strategies for a self-testing behaviour are genuinely equivalent (cf.~\cref{rem:RedEquiv}), and innately suggest an approximate relaxation suitable for robust self-testing (\cref{rem:Approx}). Moreover, by relaxing the definition and ultimately relating it to the framework of causal channels and dilations \cite{Hou21}, it will additionally clarify the significance of standardised strategies (cf.~the second item above) and point to a sensible definition of self-testing in experimental scenarios with other causal structures than Bell-scenarios. Most importantly, perhaps, this connection to a larger framework uncovers quantum self-testing as a special instance of a general and natural phenomenon (\cref{subsec:Causal}). \\

Our operational definition affects one more issue, which we now invite the reader to consider. 

The logic which drives the operator-algebraic approach to self-testing is that the behaviour-preserving transformations on $\Strat{\fB}$ -- e.g.~unitary conjugations and tensorings with unmeasured states -- should determine and justify the shape of the formal definition of reducibility (\cref{def:Redattl}). The mindset which motivates this definition is that each such behaviour-preserving transformation witnesses an ambiguity in the operator-algebraic description of the physical system, and that \emph{therefore} the definition of self-testing should be no stricter than to include this transformation as an allowed modification to the canonical strategy \cite{SB19}. 

Now, however, suppose that -- in addition to tensorings and unitary conjugations -- a third kind of behaviour-preserving transformation $\scrT: \Strat{\fB} \to \Strat{\fB}$ were to be discovered, such that for some behaviour $P$ and for any strategy $S \in \Strat{\fB}$ with behaviour $P$, the strategy $\scrT(S)$ is not reducible to $S$ in the sense of \cref{def:Redattl}. This situation is not at all hypothetical: The operation of \emph{transposing} (in some basis) the elements of a strategy preserves its behaviour, while it does not necessarily result in a strategy which is reducible in the sense of \cref{def:Redattl} to the initial one (see e.g.~Section 3.7.1 in Ref.~\cite{SB19}; in this example, the state is invariant under transposition). Likewise, Refs.~\cite{Horo15, Christ20} present examples of strategies related by \myuline{partial} transposition (on one of the systems $\cH_\sfP)$; these strategies have identical behaviours too, but their states have radically different amounts of entanglement\footnote{Since the states have a positive partial transpose (PPT), both the original state and its partial transpose are valid states.} and they are therefore likely not reducible to one another. 

One possible reaction to the discovery of such a new behaviour-preserving transformation, is to conclude that $P$ then constitutes a counterexample to self-testing. The alternative reaction is to conclude that the definition of self-testing itself should be weakened, so as to additionally incorporate the transformation $\scrT$ which was so far overlooked. 

Now, the conventional operator-algebraic approach to self-testing is really unable to decide which of the two reactions to choose. Both are equally valid from an operator-algebraic perspective. It seems that, generally, the reaction of altering the definition is favoured (see e.g.~Section 3.7.1 of Ref.~\cite{SB19}); this might be said to be also the natural choice from the operator-algebraic viewpoint, since it harmonises with the logic and mindset which motivated the reducibility definition in the first place.

However, our operational approach to self-testing shows that a completely different mindset is possible. The operational definition of self-testing presented in this work will turn out to be equivalent (cf.~\cref{thm:RvsS}) to the definition induced by \cref{def:Redattl} --- without transpositions, indeed without \emph{any} additional behaviour-preserving transformations. While we leave it as an interesting open question whether there exists \myuline{another} operational definition which corresponds to the inclusion of e.g.~transpositions, our work demonstrates that a criterion external to the operator-algebraic formalism \emph{can} be used to take a stand on the question of whether additional transformations should be absorbed into the self-testing definition or not. By our approach, the original reducibility definition, \cref{def:Redattl}, is justified and enforced not only \emph{algebraically} but indeed \emph{operationally}.

\section{Reworking Quantum Self-Testing}
\label{sec:Rework}

In this section, we present a purely operational definition of quantum self-testing and prove its equivalence to the standard notion of self-testing. We then elaborate on this operational formalisation to eventually explain how quantum self-testing can be realised from an even more general perspective. This section is the core section of our paper, and it is divided into three subsections.

Our new definition and the link to conventional self-testing is established in \cref{subsec:Sim}. The new definition works by replacing the operator-algebraic reducibility relation (\cref{def:Redattl}) with an operationally defined pre-order between strategies, which we will call \emph{local simulation} (\cref{def:Sim}). The corresponding notion of self-testing (\cref{def:OpSelf}) will be referred to as \emph{self-testing according to local simulation} (for short, \emph{self-testing a.t.l.s.}), whereas the conventional self-testing notion (\cref{def:SelftestAttl}) is for emphasis referred to as \emph{self-testing according to reducibility} (for short, \emph{self-testing a.t.r.}). We then prove the main theorems relating local simulation to reducibility (\cref{thm:RvsS}) and self-testing a.t.l.s.~to self-testing a.t.r.~(\cref{thm:SelfRvsS}). Strictly speaking, we can only prove the two notions equivalent under the technical assumption that the pure state of the canonical strategy $\tilde{S}$ has locally full rank.\footnote{This is a non-trivial condition, since $\tilde{S}$ is assumed projective.} We know of no examples of self-testing where this condition is not met.

 In \cref{subsec:AssSim}, in order to improve our understanding of local simulation, we introduce a slightly coarser relation, \emph{local assisted simulation} (\cref{def:LocAss}). Like local simulation, local assisted simulation is a pre-order on the class of all strategies, and it is interpretable in any operational theory. It turns out that local assisted simulation is in quantum theory an  \myuline{equivalence} relation (\cref{thm:AssSym}), in fact the equivalence relation \emph{generated by} the relation of local simulation (\cref{rem:SelfEquiv}). This circumstance reveals in particular that quantum self-testing can be thought of as a two-fold phenomenon and it leads us to define an alternative notion of self-testing based on local assisted simulation (see \cref{def:SelfTestAtlas} and the preceding discussion).
 
\cref{subsec:Causal} takes our operational approach to self-testing one step further. By coarsening local assisted simulation further to \emph{causal simulation} (\cref{def:CausSim}), we ultimately realise quantum self-testing as an instance of the general theory of causal channels and dilations \cite{Hou21}. This framework in particular explains how to perceive of quantum strategies and self-testing in a modular environment of general circuits, and how to generalise self-testing to scenarios with arbitrary causal structure. The presentation in \cref{subsec:Causal} is kept somewhat brief, and we refer the reader to the PhD thesis \cite{Hou21} for further details (the general version of self-testing is called `rigidity' therein), and for comparisons to other frameworks involving general circuits of information channels \cite{Chir09combs,Peri17, Kiss17}.

\subsection{A Purely Operational Formulation of Self-Testing }
\label{subsec:Sim}

Let $\fB=(X_\sfA, X_\sfB, Y_\sfA, Y_\sfB)$ be a Bell-scenario. Recall from \cref{def:OpStrat} the notion of a strategy for $\fB$, namely a triple $(\varrho, \Lambda_\sfA, \Lambda_\sfB)$ consisting of a state and channels representing measurement ensembles. Recall also that this notion of strategy is purely operational, as it refers only to the channels of the theory and to the operational notion of being a measurement ensemble.

Finally, recall from \cref{subsec:Stinespring} that a \emph{dilation} of a channel $\Lambda$ from $\cH$ to $\cK$ is a channel $\Phi$ from $\cH$ to $\cK \otimes \cE$ whose marginal, when discarding the  \emph{environment} $\cE$, is $\Lambda$.

\begin{Definition}[Component-wise Dilations and Implementations]
Let $S=(\varrho, \Lambda_\sfA, \Lambda_\sfB)$ be a strategy. A \emph{component-wise dilation of $S$} is a triple $(\xi, \Phi_\sfA, \Phi_\sfB)$ where $\xi$ is a dilation of $\varrho$, $\Phi_\sfA$ a dilation of $\Lambda_\sfA$ and $\Phi_\sfB$ a dilation of $\Lambda_\sfB$. An \emph{implementation of $S$} is a channel of the form 
 \begin{align}
	\label{eq:BhvImp}
	    \myQ{0.7}{0.7}{
	& \push{\what{X}_\sfA}  \qw  & \qw & \multigate{1}{ \Phi_\sfA} & \push{\what{Y}_\sfA}  \qw & \qw  \\
	& \Nmultigate{2}{\xi}  & \push{\cH_\sfA} \qw & \ghost{\Phi_\sfA} & \push{\cE_\sfA} \ww & \ww \\
	& \Nghost{\xi}  & \push{\cE_0} \ww & \ww & \ww & \ww \\
	& \Nghost{\xi} & \push{\cH_\sfB} \qw & \multigate{1}{\Phi_\sfB} & \push{\cE_\sfB} \ww & \ww  \\
	&   \push{\what{X}_\sfB} \qw  & \qw&  \ghost{\Phi_\sfB} & \push{\what{Y}_\sfB}  \qw & \qw  \\
} ,
	\end{align}
	where $(\xi, \Phi_\sfA, \Phi_\sfB)$ is a component-wise dilation of $S$.
\end{Definition}

The \emph{trivial implementation} of $S= (\varrho, \Lambda_\sfA, \Lambda_\sfB)$ is the implementation corresponding to the component-wise dilation $(\varrho, \Lambda_\sfA, \Lambda_\sfB)$ of $S$. Note that this is simply the behaviour of the strategy. In the theory $\QIT$, a \emph{Stinespring implementation} of $S$ is an implementation corresponding to a \emph{component-wise Stinespring dilation}, i.e.~a component-wise dilation $(\psi, \Sigma_\sfA, \Sigma_\sfB)$ where $\Sigma_{\sfA}$ and $\Sigma_\sfB$ are Stinespring dilations of $\Lambda_{\sfA}$ and $\Lambda_\sfB$, respectively, and where $\psi$ is a Stinespring dilation (purification) of $\varrho$. \\

The concepts of component-wise dilation and implementation of a strategy are operational, in fact compositional, cf.~\cref{sec:Operational}.

Dilations of a channel represent side-information escaping to an environment during the execution of the channel. Thus, given a strategy $(\varrho, \Lambda_\sfA, \Lambda_\sfB)$, a viable interpretation of the component-wise dilation $(\xi, \Phi_\sfA, \Phi_\sfB)$ is that $\xi$ represents pre-existing side-information about the state $\varrho$, and that, for $\sfP \in \{\sfA, \sfB\}$, $\Phi_{\sfP}$ represents side-information formed locally about the input $x_{\sfP} \in X_{\sfP}$ and the subsystem $\cH_{\sfP}$ upon applying the channel $\Lambda_{\sfP}$. We think of the implementation given by \cref{eq:BhvImp} as representing the total information processing which may occur when taking the environment into account: the side-information residing in system $\cE_0$ is released to the environment in advance of seeing the inputs $x_\sfA$ and $x_\sfB$, whereas the side-information residing in system $\cE_{\sfP}$ is leaked upon seeing input $x_{\sfP}$. Note that when $\varrho$ is a pure state, then (by \cref{Cor:StinePure}) any dilation $\xi$ factors as  $\varrho \otimes \sigma$ for some state $\sigma$, so any implementation factors too --- this signifies that pre-existing side-information is independent from information generated upon the inputs $x_\sfA$ and $x_\sfB$. \\

The implementation \eqref{eq:BhvImp} is always a dilation of the behaviour channel $P$, with environment $\cE_\sfA \otimes \cE_0 \otimes \cE_\sfB$. It is a Stinespring dilation of $P$ when the implementation is a Stinespring implementation. As we shall see, however, the distinction between these three parts of the environment $\cE_\sfA$, $\cE_\sfB$ and  $\cE_0$ is instrumental to the significance of implementations. 

\begin{Definition}[Local Simulation] \label{def:Sim}
    Let $S$ and $S'$ be any two strategies for $\fB$. We say that $S$ \emph{locally simulates} $S'$, written $S \geq_{l.s.} S'$, if every implementation of $S'$ is also an implementation of $S$. 
    
    Explicitly, $S \geq_{l.s.} S'$ if for any component-wise dilation $(\xi', \Phi'_\sfA, \Phi'_\sfB)$ of $S'$ there exists a component-wise dilation $(\xi, \Phi_\sfA, \Phi_\sfB)$ of $S$ such that the respective environments match up and their corresponding implementations coincide, i.e. 
    \begin{align}
	    \myQ{0.7}{0.7}{
	& \push{\what{X}_\sfA}  \qw  & \qw & \multigate{1}{ \Phi_\sfA} & \push{\what{Y}_\sfA}  \qw & \qw  \\
	& \Nmultigate{2}{\xi}  & \push{\cH_\sfA} \qw & \ghost{\Phi_\sfA} & \push{\cE'_\sfA} \ww & \ww \\
	& \Nghost{\xi}  & \push{\cE'_0} \ww & \ww & \ww & \ww \\
	& \Nghost{\xi} & \push{\cH_\sfB} \qw & \multigate{1}{\Phi_\sfB} & \push{\cE'_\sfB} \ww & \ww  \\
	&   \push{\what{X}_\sfB} \qw  & \qw&  \ghost{\Phi_\sfB} & \push{\what{Y}_\sfB}  \qw & \qw  \\
} \quad = \quad   \myQ{0.7}{0.7}{
& \push{\what{X}_\sfA}  \qw  & \qw & \multigate{1}{ \Phi'_\sfA} & \push{\what{Y}_\sfA}  \qw & \qw  \\
& \Nmultigate{2}{\xi'}  & \push{\cH'_\sfA} \qw & \ghost{\Phi'_\sfA} & \push{\cE'_\sfA} \ww & \ww \\
& \Nghost{\xi'}  & \push{\cE'_0} \ww & \ww & \ww & \ww \\
& \Nghost{\xi'} & \push{\cH'_\sfB} \qw & \multigate{1}{\Phi'_\sfB} & \push{\cE'_\sfB} \ww & \ww  \\
&   \push{\what{X}_\sfB} \qw  & \qw&  \ghost{\Phi'_\sfB} & \push{\what{Y}_\sfB}  \qw & \qw  \\
} .
	\end{align}
\end{Definition}

The relation $S \geq_{l.s.} S'$ seeks to express that any side-information which can be generated during the execution of $S'$ may also be generated during the execution of $S$. The reader may wonder why in \cref{def:Sim} we have chosen the name `\myuline{local} simulation' (and what it has to do with the description in \cref{sec:Intro}), but this will become clear in \cref{prop:Rechar}. The formulation in \cref{def:Sim} has the advantage that it is obviously interpretable in every compositional theory.

\begin{Remark}
If $S \geq_{l.s.} S'$, then the trivial implementation of $S'$ is also the trivial implementation of $S$. In other words, $S$ and $S'$ must have the same behaviour. 
\end{Remark}

\begin{Prop}
Local simulation is a pre-order on the class of strategies for $\fB$, $\Strat{\fB}$.    
\end{Prop}

\begin{proof}
   Reflexivity and transitivity are both obvious. 
\end{proof}

\begin{Example}[Full-Rank Strategies] \label{ex:fullrank}
Any strategy $S=(\varrho, \Lambda_\sfA, \Lambda_\sfB)$ is equivalent under local simulation to the strategy $\bar{S} = (\bar{\varrho}, \bar{\Lambda}_\sfA, \bar{\Lambda}_\sfB)$ which arises from $S$ by cutting down the systems $\cH_\sfA$ and $\cH_\sfB$ to the local supports of the state $\varrho$ (i.e.~$S \geq_{l.s.} \bar{S}$ and $\bar{S} \geq_{l.s.} S$). In this sense, a strategy may without loss of generality be assumed to be full-rank. 
\end{Example}

\begin{Example}[Augmentation by Ancillary States]  \label{ex:Augmented}
Let $S=(\varrho, \Lambda_\sfA, \Lambda_\sfB)$ be a strategy and let $\gamma$ be a state on a bipartite system $\cK_\sfA \otimes \cK_\sfB$. Define the \emph{$\gamma$-augmented strategy $S[\gamma]$} as the one which arises from $S$ by appending $\gamma$ to the state $\varrho$ and locally discarding the system $\cK_\sfP$ when measuring. Explicitly, the new state is $\varrho \otimes \gamma$ on the bipartite system $(\cH_\sfA \otimes \cK_\sfA) \otimes (\cH_\sfB \otimes \cK_\sfB)$, and the new measurement channels are $\Lambda_\sfA \otimes \tr_{\cK_\sfA}$ and $\Lambda_\sfB \otimes \tr_{\cK_\sfB}$. It is easy to check that $S[\gamma] \geq_{l.s.} S$. (The converse is generally false, as will be clear from \cref{prop:StateExtract}.)
\end{Example}

\pagebreak[2]

We now state our new, operational self-testing definition: \begin{Definition}[Self-Testing according to Local Simulation] \label{def:OpSelf}
    Let $\tilde{S}$ be any strategy with behaviour $P$. We say that \emph{$P$ self-tests $\tilde{S}$ according to local simulation} (for short, \emph{a.t.l.s.}) if any strategy $S$ with behaviour $P$ locally simulates $\tilde{S}$, i.e.~$S \geq_{l.s.} \tilde{S}$.
\end{Definition}

Contrary to the usual notion of self-testing according to reducibility (a.t.r.), our notion of self-testing a.t.l.s.~unproblematically quantifies universally over \myuline{all} strategies $S$, projective or not, and moreover does not a priori assume the canonical strategy $\tilde{S}$ to be pure-state and projective. \\

The main results of this subsection are the following:

\begin{Thm}[Reducibility versus Local Simulation] \label{thm:RvsS}
    Let $\tilde{S}$ be a projective pure-state strategy. Assume moreover that $\tilde{S}$ is full-rank. Then, for any projective strategy $S$, it holds that $S \geq_{red.} \tilde{S}$ if and only if $S \geq_{l.s.} \tilde{S}$. 
\end{Thm}

\begin{Thm}[Self-Testing a.t.r.~versus Self-Testing a.t.l.s.] \label{thm:SelfRvsS}
Let  $\tilde{S}$ be a projective pure-state and full-rank strategy, and let $P$ denote its behaviour. Then, $P$ self-tests $\tilde{S}$ a.t.r. if and only if $P$ self-tests $\tilde{S}$ a.t.l.s.
\end{Thm}

\begin{Remark}
In connection with \cref{thm:RvsS}, it is worth remarking that $S \geq_{red.} \tilde{S}$ implies $S \geq_{l.s.} \tilde{S}$ regardless of the rank-assumption on the state.  \end{Remark}

The remainder of this section is devoted to proving \cref{thm:RvsS} and \cref{thm:SelfRvsS}. Note that the `if'-direction in \cref{thm:SelfRvsS} follows immediately from \cref{thm:RvsS}. The `only if'-direction is not immediate, since self-testing a.t.r.~quantifies only over projective strategies; it will however follow if we can prove that every strategy is $\geq_{l.s.}$-equivalent to a projective strategy. Thus, we prove this (\cref{prop:Proj} below) and we prove \cref{thm:RvsS}.  \\

First, we need to better understand the notion of local simulation among strategies. A key step is the observation that local simulation may in $\QIT$ be characterised in terms of Stinespring implementations:\footnote{As the proof of \cref{prop:Rechar} shows, local simulation may more generally be characterised in terms of implementations corresponding to \emph{complete} dilations, cf.~the discussion in \cref{subsec:Stinespring}.} 

\begin{Lem}[Stinespring Characterisation of Local Simulation] \label{prop:Rechar}
     Let $S$ and $S'$ be strategies, and let $(\psi, \Sigma_\sfA, \Sigma_\sfB)$ and $(\psi', \Sigma'_\sfA, \Sigma'_\sfB)$ be component-wise Stinespring dilations of $S$ and $S'$, respectively. Then, $S \geq_{l.s.} S'$ if and only if there exist channels $\Gamma_0$, $\Gamma_\sfA$ and $\Gamma_\sfB$ such that 
  \begin{align}  \label{eq:SimStine}
\myQ{0.7}{0.7}{
	& \push{\what{X}_\sfA}  \qw  & \multigate{1}{ \Sigma_\sfA} & \qw &  \push{\what{Y}_\sfA}  \qw & \qw  \\
	& \Nmultigate{2}{\psi}  & \ghost{\Sigma_\sfA} & \push{\cE_\sfA} \ww &  \Ngate{\Gamma_\sfA}{\ww} &\push{\cE'_A} \ww& \ww\\
	& \Nghost{\psi} & \ww & \push{\cE_0} \ww &\Ngate{\Gamma_0}{\ww}& \push{\cE'_0} \ww& \ww \\
	& \Nghost{\psi} & \multigate{1}{\Sigma_\sfB} & \push{\cE_\sfB} \ww &  \Ngate{\Gamma_\sfB}{\ww} &\push{\cE'_B} \ww& \ww \\
	&   \push{\what{X}_\sfB} \qw  & \ghost{\Sigma_\sfB} & \qw &  \push{\what{Y}_\sfB}  \qw & \qw  \\
} 
\quad = \quad 
\myQ{0.7}{0.7}{
	& \push{\what{X}_\sfA}  \qw  & \multigate{1}{ \Sigma'_\sfA} & \qw &  \push{\what{Y}_\sfA}  \qw & \qw  \\
	& \Nmultigate{2}{\psi'}  & \ghost{\Sigma'_\sfA} & \push{\cE'_A} \ww& \ww\\
	& \Nghost{\psi'} & \ww & \push{\cE'_0} \ww&\ww  \\
	& \Nghost{\psi'} & \multigate{1}{\Sigma'_\sfB} & \push{\cE'_\sfB} \ww&\ww  \\
	&   \push{\what{X}_\sfB} \qw  & \ghost{\Sigma'_\sfB} & \qw &  \push{\what{Y}_\sfB}  \qw & \qw  \\
} .
\end{align}
\end{Lem}

\begin{proof}
    By \cref{thm:StineComp}, any component-wise dilation $(\xi, \Phi_\sfA, \Phi_\sfB)$ of $S$ can be obtained by applying channels $\Gamma_\sfA$, $\Gamma_\sfB$ and $\Gamma_0$ to the elements of the component-wise Stinespring dilation $(\psi, \Sigma_\sfA, \Sigma_\sfB)$. Hence, any implementation of the strategy $S$ is of the form displayed on the left hand side in \cref{eq:SimStine}, for some choice of $\Gamma_\sfA$, $\Gamma_\sfB$ and $\Gamma_0$. Similarly, any implementation of $S'$ is obtained by applying some channels $\Gamma'_\sfA$, $\Gamma'_\sfB$ and $\Gamma'_0$ to the Stinespring implementation on the right hand side. As such, local simulation amounts to finding for each choice of $\Gamma'_\sfA$, $\Gamma'_\sfB$ and $\Gamma'_0$ a matching choice of $\Gamma_\sfA$, $\Gamma_\sfB$ and $\Gamma_0$. However, it suffices to do this when $\Gamma'_\sfA$, $\Gamma'_\sfB$ and $\Gamma'_0$ are all identity channels, since a general choice of channels $\Gamma_i$ is the matched by the compositions $\Gamma'_i \circ \Gamma_i$, for $i \in \{\sfA, \sfB, 0\}$. \end{proof}

\begin{Remark}When strategies $S$ and $S'$ have the same behaviour, there is always by \cref{thm:StineComp} \myuline{some} channel $\Gamma$ such that 
 \begin{align}  
\myQ{0.7}{0.7}{
	& \push{\what{X}_\sfA}  \qw  & \multigate{1}{ \Sigma_\sfA} & \qw & \push{\what{Y}_\sfA}  \qw & \qw  \\
	& \Nmultigate{2}{\psi}  & \ghost{\Sigma_\sfA} & \push{\cE_\sfA} \ww  &  \Nmultigate{2}{\Gamma}{\ww}& \push{\cE'_\sfA} \ww  & \ww\\
	& \Nghost{\psi} & \ww& \push{\cE_0} \ww  &\Nghost{\Gamma}{\ww} & \push{\cE'_0} \ww & \ww \\
	& \Nghost{\psi} & \multigate{1}{\Sigma_\sfB} & \push{\cE_\sfB} \ww & \Nghost{\Gamma}{\ww} & \push{\cE'_\sfB} \ww & \ww \\
	&   \push{\what{X}_\sfB} \qw  & \ghost{\Sigma_\sfB} & \qw &  \push{\what{Y}_\sfB}  \qw & \qw  \\
} 
\quad = \quad 
\myQ{0.7}{0.7}{
	& \push{\what{X}_\sfA}  \qw  & \multigate{1}{ \Sigma'_\sfA} & \push{\what{Y}_\sfA}  \qw & \qw  \\
	& \Nmultigate{2}{\psi'}  & \ghost{\Sigma'_\sfA}  & \push{\cE'_\sfA} \ww & \ww\\
	& \Nghost{\psi'} & \ww & \push{\cE'_0} \ww & \ww \\
	& \Nghost{\psi'} & \multigate{1}{\Sigma'_\sfB} & \push{\cE'_\sfB} \ww &\ww  \\
	&   \push{\what{X}_\sfB} \qw  & \ghost{\Sigma'_\sfB} & \push{\what{Y}_\sfB}  \qw & \qw  \\
} ,
\end{align}
 since a Stinespring implementation of a strategy is in particular a Stinespring dilation of its behaviour. Local simulation is a non-trivial relation between Stinespring implementations because it requires a channel $\Gamma$ which \myuline{factors} as $\Gamma_\sfA \otimes \Gamma_0 \otimes \Gamma_\sfB$.
\end{Remark}

With \cref{prop:Rechar} at our disposal, local simulation has become a more tractable relation. (The reader may find it interesting to consult \cref{ex:fullrank} and \cref{ex:Augmented} in this light.) In particular, we will use it now to prove that any strategy is equivalent under local simulation to a projective strategy --- once that result is established, \cref{thm:SelfRvsS} will follow automatically from \cref{thm:RvsS}. The result is however also interesting by itself, since it gives, for the first time, mathematical substance to the common perception that one may `without loss of generality' assume quantum strategies to be projective:

\begin{Prop} \label{prop:Proj}
    Any strategy $S$ is equivalent under local simulation to a projective strategy $S'$.
\end{Prop}

\begin{proof}
Let $S=(\varrho, \Lambda_\sfA, \Lambda_\sfB)$. For each $\sfP \in \{\sfA, \sfB\}$, it follows from Naimark's theorem (see \cref{rem:NaiEnsemble}) that the measurement ensemble $\Lambda_{\sfP}$ can be written as $\Lambda^\up{Nai}_{\sfP} \circ (  \id_{\what{X}_{\sfP}} \otimes \id_{\cH_{\sfP}} \otimes \phi^\up{Nai}_{\sfP} )$ with $\phi^{\up{Nai}}_{\sfP}$ a pure state on some system $\cK^{\up{Nai}}_{\sfP}$, and $\Lambda^{\up{Nai}}_{\sfP}$ an ensemble of \myuline{projective} measurements on $\cH_{\sfP} \otimes \cK^{\up{Nai}}_{\sfP}$. We show that $S$ is $\geq_{l.s.}$-equivalent to the projective strategy $S':= (\varrho \otimes \phi^{\up{Nai}}_{\sfA} \otimes \phi^{\up{Nai}}_{\sfB}, \Lambda^{\up{Nai}}_\sfA, \Lambda^{\up{Nai}}_\sfB)$. To this end, observe that if $\psi$ is a purification of $\varrho$ and $\Sigma^{\up{Nai}}_{\sfP}$ a Stinespring dilation of $\Lambda^{\up{Nai}}_{\sfP}$, then $(\psi \otimes  \phi^{\up{Nai}}_{\sfA} \otimes \phi^{\up{Nai}}_{\sfB}, \Sigma^{\up{Nai}}_\sfA, \Sigma^{\up{Nai}}_\sfB)$ is a component-wise Stinespring dilation of $S'$ and $(\psi, \Sigma^{\up{Nai}}_\sfA \circ (  \id_{\what{X_\sfA}} \otimes\id_{\cH_\sfA} \otimes \phi^\up{Nai}_\sfA ), \Sigma^{\up{Nai}}_\sfB \circ (  \id_{\what{X_\sfB}} \otimes \id_{\cH_\sfB} \otimes \phi^\up{Nai}_\sfB ))$ is a component-wise Stinespring dilation of $S$. These corresponding Stinespring implementations are moreover identical. Thus, it follows immediately from \cref{prop:Rechar} that $S \geq_{l.s.} S'$ and $S' \geq_{l.s.} S$.\end{proof}

It remains now only to prove \cref{thm:RvsS}. To this end, we first translate local simulation to operator-algebraic language in the case of projective strategies:\footnote{We additionally assume that the state of the simulated strategy is pure, because the result in this case takes on a particularly simple form, and because we shall not need the result in the case of a general mixed state.}

\begin{Prop}[Operator-Algebraic Characterisation of Local Simulation] \label{prop:OpSim}
Let $S = (\varrho, \Pi_\sfA, \Pi_\sfB)$ and $\tilde{S}=(\tilde{\psi}, \tilde{\Pi}_\sfA, \tilde{\Pi}_\sfB)$ be projective strategies, with $\tilde{\psi}$ a pure state. Let $\psi$ be a purification of $\varrho$ with purifying space $\cP$. Then, $S \geq_{l.s.} \tilde{S}$ if and only if there exist Hilbert spaces $\cH^\res_\sfA$, $\cH^\res_\sfB$,  a pure state $\psi^\res$ on $\cH^\res_\sfA \otimes \cH^\res_\sfB \otimes \cP$, and isometries $V_\sfA: \what{X}_\sfA \otimes \cH_{\sfA}  \to \what{X}_{\sfA} \otimes \tilde{\cH}_{\sfA} \otimes \cH^\res_{\sfA}$ and $V_\sfB: \what{X}_\sfB \otimes \cH_{\sfB}  \to \what{X}_{\sfB} \otimes \tilde{\cH}_{\sfB} \otimes \cH^\res_{\sfB}$ such that, with $V= V_\sfA \otimes V_\sfB$, we have
\begin{align} \label{eq:SimOp}
 \;[ V ( \ket{x}  \otimes \Pi^x(y)) \otimes \bone_\cP] \ket{\psi} = \ket{x} \otimes \tilde{\Pi}^x(y) \tilde{\ket{\psi}} \otimes \ket{\psi^\res} .
	\end{align} 
\end{Prop}

\begin{Remark}
Note that, contrary to the conventional notion of reducibility (in \cref{def:Redattl}), the local isometries $V_\sfP$ in \cref{eq:SimOp} are allowed to `look at $x_\sfP$' on the left hand side, and must as well `return $x_\sfP$' on the right hand side.  This distinction is conceptually interesting, and we suspect it to be important for approximate versions of local simulation (\cref{rem:Approx}); in the exact case, however, it turns out to vanish (\cref{lem:technical}).
\end{Remark}

\begin{proof}[Proof (of \cref{prop:OpSim})]
  There are two key ideas in this proof. By \cref{prop:Rechar}, the condition $S \geq_{l.s.} \tilde{S}$ is equivalent, for any choice of Stinespring dilations $\Sigma_{\sfA}$, $\Sigma_\sfB$, $\tilde{\Sigma}_\sfA, \tilde{\Sigma}_\sfB$ of $\Lambda_{\sfA}, \Lambda_\sfB, \tilde{\Lambda}_\sfA, \tilde{\Lambda}_\sfB$, to the existence of $\Gamma_\sfA$ and $\Gamma_\sfB$ such that 
 \begin{align} \label{eq:SimStart}
  \myQ{0.7}{0.7}{
	& \push{\what{X}_\sfA}  \qw  & \multigate{1}{ \Sigma_\sfA} & \qw &  \push{\what{Y}_\sfA}  \qw & \qw  \\
	& \Nmultigate{2}{\psi}  & \ghost{\Sigma_\sfA} & \push{\cE_\sfA} \ww &   \Ngate{\Gamma_\sfA}{\ww} &  \push{\tilde{\cE}_\sfA} \ww & \ww\\
		& \Nghost{\psi}  & \push{\cP} \ww & \Ngate{\tr}{\ww} \\
	& \Nghost{\psi} & \multigate{1}{\Sigma_\sfB}  & \push{\cE_\sfB} \ww  & \Ngate{\Gamma_\sfB}{\ww} &  \push{\tilde{\cE}_\sfB} \ww & \ww \\
	&   \push{\what{X}_\sfB} \qw  & \ghost{\Sigma_\sfB} & \qw & \push{\what{Y}_\sfB}  \qw & \qw  
} 	\quad =  \quad 
\myQ{0.7}{0.7}{
	& \push{\what{X}_\sfA}  \qw  & \multigate{1}{ \tilde{\Sigma}_\sfA} & \push{\what{Y}_\sfA}  \qw & \qw  \\
	& \Nmultigate{2}{\tilde{\psi}}  & \ghost{\tilde{\Sigma}_\sfA} &  \push{\tilde{\cE}_\sfA} \ww &  \ww\\
		& \Nghost{\tilde{\psi}}  \\
	& \Nghost{\tilde{\psi}} & \multigate{1}{\tilde{\Sigma}_\sfB} &  \push{\tilde{\cE}_\sfB} \ww & \ww  \\
	&   \push{\what{X}_\sfB} \qw  & \ghost{\tilde{\Sigma}_\sfB} & \push{\what{Y}_\sfB}  \qw & \qw  
} 
\end{align}
 (the trace $\tr_\cP$ is the only possible choice for $\Gamma_0$). The first idea is to realise that this condition is equivalent to the existence of \myuline{isometric} channels $\breve{\Gamma}_\sfA$, $\breve{\Gamma}_\sfB$ and a pure state $\psi^\res$ such that 
\begin{align} \label{eq:SimPure}
      	\myQ{0.7}{0.7}{
		& \push{\what{X}_\sfA}  \qw  & \multigate{1}{ \Sigma_\sfA} & \qw &  \push{\what{Y}_\sfA}  \qw & \qw  \\
		& \Nmultigate{4}{\psi}  & \ghost{\Sigma_\sfA}  & \push{\cE_\sfA} \ww &  \Nmultigate{1}{\breve{\Gamma}_\sfA}{\ww} & \push{\tilde{\cE}_\sfA}  \ww & \ww \\
		& \Nghost{\psi}& & & \Nghost{\breve{\Gamma}_\sfA} & \push{\cH^\res_\sfA} \ww & \ww\\
		& \Nghost{\psi}& \ww & \ww & \ww & \push{{\cP}} \ww & \ww\\
		&\Nghost{\psi} & & & \Nmultigate{1}{\breve{\Gamma}_\sfB} &   \push{\cH^\res_\sfB} \ww & \ww\\
		& \Nghost{\psi} & \multigate{1}{\Sigma_\sfB} & \push{\cE_\sfB} \ww &  \Nghost{\breve{\Gamma}_\sfB}{\ww} & \push{\tilde{\cE}_\sfB}  \ww & \ww  \\
		&   \push{\what{X}_\sfB} \qw  & \ghost{\Sigma_\sfB} & \qw & \push{\what{Y}_\sfB}  \qw & \qw  \\
	} 
\quad 	 =  \quad 
	\myQ{0.7}{0.7}{
		& \push{\what{X}_\sfA}  \qw  & \multigate{1}{ \tilde{\Sigma}_\sfA} & \push{\what{Y}_\sfA}  \qw & \qw  \\
		& \Nmultigate{4}{\tilde{\psi}}  & \ghost{\tilde{\Sigma}_\sfA} & \push{\tilde{\cE}_\sfA}  \ww & \ww \\
		& \Nghost{\tilde{\psi}} &  \Nmultigate{2}{\psi^\textup{res}} &   \push{\cH^\res_\sfA} \ww & \ww \\
		&	\Nghost{\psi'} & \Nghost{\psi^\textup{res}} & \push{{\cP}}\ww & \ww  \\ 
		& \Nghost{\psi'} &  \Nghost{\psi^\textup{res}} &   \push{\cH^\res_\sfB} \ww & \ww \\
		& \Nghost{\tilde{\psi}} & \multigate{1}{\tilde{\Sigma}_\sfB}  & \push{\tilde{\cE}_\sfB}  \ww & \ww \\
		&   \push{\what{X}_\sfB} \qw  & \ghost{\tilde{\Sigma}_\sfB} & \push{\what{Y}_\sfB}  \qw & \qw  \\
	} .
  \end{align}

Indeed, \cref{eq:SimStart} follows trivially from \cref{eq:SimPure} by tracing out $\cH^\res_\sfA$ and $\cH^\res_\sfB$. More interestingly, \cref{eq:SimPure} follows from \cref{eq:SimStart} by observing that if on the left hand of \cref{eq:SimStart} we replace the trace by an identity and, for each $\sfP \in \{\sfA, \sfB\}$, the channels $\Gamma_{\sfP}$ by Stinespring dilations $\breve{\Gamma}_{\sfP}$, then the resulting channel is a dilation of the right hand side, which is isometric and therefore by \cref{Cor:StinePure} dilationally pure; consequently, that channel must take the form of tensoring with some state $\psi^\res$, and this state must be pure since the total channel is now isometric. All in all \cref{eq:SimPure} thus follows. It is worth highlighting that the condition \eqref{eq:SimPure} may be thought of as a `purified' recharacterisation of local simulation, which might have been worthy of a separate lemma, given unbounded space.

Given the purified characterisation \eqref{eq:SimPure}, the second idea in the proof is to choose the Stinespring dilations $\Sigma_{\sfA}, \Sigma_\sfB$, $\tilde{\Sigma}_\sfA$ and $\tilde{\Sigma}_{\sfB}$ conveniently. For $\sfP \in \{\sfA, \sfB\}$, the projective measurement ensemble $\Lambda_{\sfP} : \what{X}_{\sfP} \otimes \cH_{\sfP} \to \what{Y}_{\sfP}$ corresponding to the PVM $\Pi_{\sfP}$ is given by 
\begin{align}
    \Lambda_{\sfP}(F \otimes G) = \sum_{x_{\sfP} \in X_{\sfP}, y_{\sfP} \in Y_{\sfP}} \bra{x_{\sfP}} F \ket{x_{\sfP}} \tr(\Pi^{x_{\sfP}}_{\sfP}(y_{\sfP}) G) \ketbra{y_{\sfP}} \; \text{for $F \in \End{\what{X}_{\sfP}}$, $G \in \End{\cH_{\sfP}}$}.
\end{align}
It can be easily checked that the operator $R_{\sfP} := \sum_{x_{\sfP} \in X_{\sfP}, y_{\sfP} \in Y_{\sfP}} \ket{x_{\sfP}} \otimes \Pi^{x_{\sfP}}_{\sfP}(y_{\sfP}) \otimes \ketbra{y_{\sfP}}{x_{\sfP}}$ is an isometry from $ \what{X}_{\sfP} \otimes \cH_{\sfP}$ to $ \what{X}_{\sfP} \otimes \cH_{\sfP} \otimes \what{Y_{\sfP}}$ (exploiting that each $\Pi^{x_{\sfP}}_{\sfP}$ is a PVM). Furthermore, the isometric channel $\Sigma_{\sfP}(\cdot) := R_{\sfP} (\cdot) R^*_{\sfP}$ is a dilation of $\Lambda_{\sfP}$ with environment $\cE_{\sfP} := \what{X}_{\sfP} \otimes \cH_{\sfP}$. Hence, $\Sigma_{\sfP}$ is a \myuline{Stinespring} dilation of $\Lambda_{\sfP}$. By the same argument,  $\tilde{R}_{\sfP} := \sum_{x_{\sfP} \in X_{\sfP}, y_{\sfP} \in Y_{\sfP}} \ket{x_{\sfP}} \otimes \tilde{\Pi}^{x_{\sfP}}_{\sfP}(y_{\sfP}) \otimes \ketbra{y_{\sfP}}{x_{\sfP}}$ is an isometry giving rise to a Stinespring dilation $\tilde{\Sigma}_{\sfP}$ of $\tilde{\Lambda}_{\sfP}$, the measurement ensemble corresponding to the PVM $\tilde{\Pi}_{\sfP}$. 

With the above choices for $\Sigma_{\sfP}$ and $\tilde{\Sigma}_{\sfP}$, the two Stinespring implementations of $S$ and $\tilde{S}$ are the isometric channels represented by the isometries
\begin{align} \label{eq:ConvStine}
\sum_{x \in X, y \in Y} \ket{x} \otimes [\Pi^x(y) \otimes \bone_\cP]\ket{\psi} \otimes \ketbra{y}{x}, \quad \text{respectively} \quad \sum_{x \in X, y \in Y} \ket{x} \otimes \tilde{\Pi}^x(y) \ket*{\tilde{\psi}} \otimes \ketbra{y}{x}.
\end{align}
As such, the condition \eqref{eq:SimPure} expresses the existence of isometries $V_{\sfP} : \what{X}_{\sfP} \otimes \cH_{\sfP} \to \what{X}_{\sfP} \otimes \tilde{\cH}_{\sfP} \otimes \cH^\res_{\sfP}$ (representing $\breve{\Gamma}_{\sfP}$) and a unit vector $\ket*{\psi^\res} \in \cH^\res_\sfA \otimes \cH^\res_\sfB \otimes \cP$ (representing $\psi^\res$) such that
\begin{align} \label{eq:OpId}
\sum_{x \in X, y \in Y} [ V( \ket{x} \otimes \Pi^x(y)) \otimes \bone_\cP ] \ket{\psi} \otimes \ketbra{y}{x} = \sum_{x \in X, y \in Y} \ket{x} \otimes \tilde{\Pi}^x(y) \ket*{\tilde{\psi}}\otimes \ket*{\psi^\res}  \otimes \ketbra{y}{x} ,
\end{align}
with $V= V_\sfA \otimes V_\sfB$. Since $\big(\ketbra{y}{x}\big)_{x \in X, y \in Y}$ forms a linearly independent system of operators, \cref{eq:OpId} is equivalent to the condition \eqref{eq:SimOp}, as desired.\end{proof}

\begin{Remark}[Simulation for Non-Projective Strategies]
It is possible to state a version of \cref{prop:OpSim} for non-projective strategies. The only step in the proof where we relied on the strategies being projective was when we chose the Stinespring dilations $\Sigma_{\sfP}$ and $\tilde{\Sigma}_{\sfP}$ of the measurement ensembles $\Lambda_{\sfP}$ and $\tilde{\Lambda}_{\sfP}$. Suppose now that the measurement ensembles $\tilde{\Lambda}_{\sfA}$ and $\tilde{\Lambda}_\sfB$ are still projective, but that, for $\sfP \in \{\sfA, \sfB\}$, the ensemble $\Lambda_{\sfP}$ corresponds to general POVMs $(E^{x_{\sfP}}_{\sfP})_{x_{\sfP} \in X_{\sfP}}$. For the Stinespring dilation $\Sigma_{\sfP}$, we are then forced to include in its environment a copy of $\what{Y}_{\sfP}$, choosing as Stinespring isometry from $\what{X}_{\sfP} \otimes \cH_{\sfP}$ to $\what{X}_{\sfP} \otimes \what{Y}_{\sfP} \otimes \cH_{\sfP} \otimes  \what{Y}_{\sfP}$ the operator
\begin{align}
    R_{\sfP} := \sum_{x_{\sfP} \in X_{\sfP}, y_{\sfP} \in Y_{\sfP}} \ket{x_{\sfP}} \otimes \ket{y_{\sfP}} \otimes \sqrt{E^{x_{\sfP}}_{\sfP}(y_{\sfP})} \otimes \ketbra{y_{\sfP}}{x_{\sfP}}.
\end{align}
As such, the Stinespring implementation corresponding to $(\psi, \Sigma_\sfA, \Sigma_\sfB)$ is represented by the isometry $\sum_{x \in X, y \in Y} \ket{x} \otimes \ket{y} \otimes [\sqrt{E^{x}(y)} \otimes \bone_\cP]\ket{\psi} \otimes \ketbra{y}{x}$, and the simulation condition translates as the existence of a unit vector $\ket{\psi^\res} \in \cH^\res_\sfA \otimes \cH^\res_\sfB  \otimes \cP$ and isometries $V_\sfP :\what{X}_\sfP \otimes \what{Y}_\sfP \otimes \cH_{\sfP} \to  \what{X}_{\sfP} \otimes \what{Y}_{\sfP} \otimes \tilde{\cH}_{\sfP} \otimes \cH^\res_{\sfP}$ such that, with $V=V_\sfA \otimes V_\sfB$, we have
\begin{align}
\forall x \in X, y \in Y : \;  [ V (\ket{x} \otimes \ket{y} \otimes \sqrt{E^{x}(y)}) \otimes \bone_\cP] \ket{\psi} = \ket{x} \otimes \ket{y} \otimes \tilde{\Pi}^x(y) \ket*{\tilde{\psi}} \otimes \ket{\psi^\res} .
\end{align}
\end{Remark}

\cref{prop:OpSim} is the link between our operational notion of local simulation and the conventional operator-algebraic notion of reducibility --- indeed, the condition \eqref{eq:SimOp} looks very similar to the reducibility condition \eqref{eq:RedAttlComp}. Ignoring for a moment the difference that \eqref{eq:SimOp} allows the local isometry to observe $x_\sfP$, the proof of \cref{prop:OpSim} explains in an entirely new way the significance of the residual state, the embedding isometries and the projected states $\Pi^{x}(y)\ket{\psi}$ and $\tilde{\Pi}^x(y)\ket*{\tilde{\psi}}$: the projected states signify side-information generated in the course of implementing the strategies, and the embedding connects side-information in one strategy to side-information in the other strategy.\\

We shall now prove that the difference between reducibility and condition \eqref{eq:SimOp} is (at least under a rank-assumption) only one of appearance. This result marks the technical highlight of the subsection, by finally establishing \cref{thm:RvsS}. We state in the result three equivalent conditions, of which the first is \cref{eq:SimOp} and the third is the conventional reducibility condition \eqref{eq:RedAttlComp}. The second is an intermediate condition which mostly serves to smoothen the proof. 

\begin{Lem} (Technical Result for \cref{thm:RvsS}). \label{lem:technical} \\
	Let $S= (\varrho, \Pi_\sfA, \Pi_\sfB)$ and $\tilde{S}= (\tilde{\psi}, \tilde{\Pi}_\sfA, \tilde{\Pi}_\sfB)$ be projective strategies, for which $\tilde{\psi}$ is a pure state with locally full rank. Let $\psi$ be a purification of $\varrho$ with purifying space $\cP$. Finally, let $\psi^\res$ be a pure state on a system of the form $\cH^\res_\sfA \otimes \cH^\res_\sfB \otimes \cP$. The following are equivalent:
	\begin{itemize}
	    \item[(1)] There exist isometries $V_\sfP: \what{X}_\sfP \otimes \cH_\sfP \to \what{X}_\sfP \otimes \tilde{\cH}_\sfP \otimes \cH^\res_\sfP$ such that, with $V= V_\sfA \otimes V_\sfB$, we have  
	    	\begin{align} \label{eq:qdependent}
\forall x \in X, y \in Y : \; [V	(\ket{x} \otimes \Pi^x(y) ) \otimes \bone_\cP] \ket{\psi} = \ket{x} \otimes \tilde{\Pi}^x(y) \tilde{\ket{\psi}} \otimes \ket{\psi^\res}  .
	\end{align} 
	\item[(2)] There exist isometries $W^{x_\sfP}_\sfP : \cH_\sfP \to  \tilde{\cH}_\sfP \otimes \cH^\res_\sfP$ for $x_\sfP \in X_\sfP$, such that, with $W^x=W^{x_\sfA}_\sfA \otimes W^{x_\sfB}_\sfB$ for $x=(x_\sfA, x_\sfB)$, we have
	\begin{align} \label{eq:vers2}
       \forall x \in X, y \in Y : \; [W^x \Pi^x(y) \otimes \bone_\cP] \ket{\psi} = \tilde{\Pi}^x(y) \ket*{\tilde{\psi}} \otimes \ket*{\psi^\res}.
   \end{align}
	\item[(3)] There exist isometries $W_\sfP : \cH_\sfP \to  \tilde{\cH}_\sfP \otimes \cH^\res_\sfP$ such that, with $W=W_\sfA \otimes W_\sfB$, we have 
	\begin{align} \label{eq:vers3}
       \forall x \in X, y \in Y : \; [W \Pi^x(y) \otimes \bone_\cP] \ket{\psi} = \tilde{\Pi}^x(y) \ket*{\tilde{\psi}} \otimes \ket*{\psi^\res}.
   \end{align}
	\end{itemize}
\end{Lem}

\begin{proof}
It is clear that (3) implies (2), by letting $W^{x_\sfP}_\sfP = W_\sfP$ independently of $x_\sfP$. It is also clear that (2) implies (1), by letting $V_\sfP= \sum_{x_\sfP \in X_\sfP} \ketbra{x_\sfP} \otimes W^{x_\sfP}_\sfP$. We show that (1) implies (2), and that (2) implies (3). We begin, however, by showing that if (1) holds then for each party $\sfP \in \{\sfA, \sfB\}$, the local support $\supp{\varrho_\sfP}$ of the $\sfP$-marginal of $\varrho$ is invariant under all of the operators $\Pi^{x_\sfP}_\sfP(y_\sfP)$, i.e.~$\Pi^{x_\sfP}_\sfP(y_\sfP)[\supp{\varrho_\sfP} ] \subseteq \supp{\varrho_\sfP}$ for $x_\sfP \in X_\sfP$ and $y_\sfP \in Y_\sfP$. 

Let us assume without loss of generality that $\sfP=\sfA$. If in \cref{eq:qdependent} we sum over $y_\sfB \in Y_\sfB$, then we obtain %
\begin{align} \label{eq:sumB}
 \forall x \in X, y_\sfA \in Y_\sfA : \;  [V	(\ket{x} \otimes \Pi^{x_\sfA}_\sfA(y_\sfA) \otimes \bone_{\cH_\sfB} ) \otimes \bone_\cP] \ket{\psi} = \ket{x} \otimes [\tilde{\Pi}^{x_\sfA}_\sfA(y_\sfA) \otimes \bone_{\tilde{\cH}_\sfB} ]\tilde{\ket{\psi}} \otimes \ket{\psi^\res} .
\end{align}
 By marginalising to the $\sfA$-systems, \cref{eq:sumB} implies that
 \begin{align} \label{eq:cond0}
	V_\sfA \big( \ketbra{x_\sfA}{x'_\sfA} \otimes \Pi^{x_{\sfA}}_{\sfA}(y_{\sfA}) \varrho_{\sfA} \Pi^{x'_{\sfA}}_{\sfA}(y'_{\sfA})\big) V_\sfA^*  =  \ketbra{x_{\sfA}}{x'_\sfA} \otimes  \tilde{\Pi}^{x_{\sfA}}_{\sfA}(y_{\sfA})\tilde{\varrho}_{\sfA}\tilde{\Pi}^{x'_{\sfA}}_{\sfA}(y'_{\sfA})  \otimes \varrho^\res_{\sfA}
	\end{align} 
	for all $x_\sfA, x'_\sfA \in X_\sfA$ and all $y_\sfA, y'_\sfA \in Y_\sfA$. Here, $\tilde{\varrho}_\sfA$ denotes the $\sfA$-marginal of $\tilde{\psi}$ and $\varrho^\res_\sfA$ the $\sfA$-marginal of $\psi^\res$. If in \cref{eq:cond0} we let $x'_\sfA=x_\sfA$ and sum over $y'_\sfA \in Y_\sfA$, then we obtain 
	\begin{align} \label{eq:cond1}
	V_\sfA \big( \ketbra{x_\sfA} \otimes \Pi^{x_{\sfA}}_{\sfA}(y_{\sfA}) \varrho_{\sfA} \big) V_\sfA^*  =  \ketbra{x_{\sfA}} \otimes  \tilde{\Pi}^{x_{\sfA}}_{\sfA}(y_{\sfA})\tilde{\varrho}_{\sfA}\otimes \varrho^\res_{\sfA}
	\end{align}
for all $x_\sfA \in X, y_\sfA \in Y_\sfA$. If additionally we sum over $y_\sfA \in Y_\sfA$, we obtain 
\begin{align} \label{eq:cond2}
	V_\sfA \big( \ketbra{x_\sfA} \otimes \varrho_{\sfA} \big) V_\sfA^*  =  \ketbra{x_{\sfA}} \otimes  \tilde{\varrho}_{\sfA}  \otimes \varrho^\res_{\sfA}
	\end{align}
for all $x_\sfA \in X_\sfA$. Now fix, $x_\sfA \in X_\sfA$ and $y_\sfA \in Y_\sfA$. What we wish to show is that $\Pi^{x_\sfA}(y_\sfA)_\sfP[\supp{\varrho_\sfA} ] \subseteq \supp{\varrho_\sfA}$, and this amounts to showing that $\Im \Pi^{x_\sfA}_\sfA(y_\sfA) \varrho_\sfA \subseteq \Im \varrho_\sfA$ where `$\Im$' denotes the image (range) of an operator. By virtue of Eqs.~\eqref{eq:cond1} and \eqref{eq:cond2}, however, this follows if we can show that $\Im \tilde{\Pi}^{x_\sfA}_\sfA(y_\sfA)\tilde{\varrho}_\sfA \subseteq \Im \tilde{\varrho}_\sfA$. But that follows directly from the assumption that $\tilde{\varrho}_\sfA$ has full rank. (This is the only step in the proof where we need the rank-assumption on $\tilde{\varrho}$.)\footnote{Moreover, the conclusion relies only on the weaker circumstance that $\up{supp}(\tilde{\varrho}_{\sfA})$ is \myuline{invariant} under the operators $\tilde{\Pi}^{x_{\sfA}}_{\sfA}(y_{\sfA})$ for $x_{\sfA} \in X_{\sfA}, y_{\sfA} \in Y_{\sfA}$, and similarly for $\sfB$. Thus, the lemma could alternatively have been stated  under this weaker assumption.}

With this conclusion in place, we are ready to prove that (1) implies (2), and that (2) implies (3). 

First, assume that (1) holds, with isometries $V_\sfA$ and $V_\sfB$. For $\sfP \in \{\sfA, \sfB\}$ and $x_\sfP \in X_\sfP$, let 
\begin{align}
    W^{x_\sfP}_\sfP := (\bone_{\tilde{\cH}_\sfP \otimes \cH^\res_\sfP} \otimes \bra{x_\sfP}) V_\sfP (\bone_{\cH_\sfP } \otimes \ket{x_\sfP}).
\end{align}
It follows directly from \cref{eq:qdependent} that \cref{eq:vers2} holds, however it is not clear that the operators $W^{x_\sfP}_\sfP$ are isometries. Fix $x_\sfA \in X_\sfA$ and $x_\sfB \in X_\sfB$. It is enough to show that $W^{x_\sfP}_\sfP$ acts isometrically on the subspace $\cH^0_\sfP \subseteq \cH_\sfP$ generated\footnote{The subspace  \emph{generated} from a subspace $\cK$ by a family of operators $(F_i)_{i \in I}$ is by definition the smallest subspace containing $\cK$ and invariant under $F_i$ for every $i \in I$.}  from $\supp{\varrho_\sfP}$ by the operators $\big(\Pi^{x_\sfP}_\sfP(y_\sfP)\big)_{y_\sfP \in Y_\sfP}$, since we may always modify the action of $W^{x_\sfP}_\sfP$ outside this subspace. However, by our introductory observation the subspace $\cH^0_\sfP$ is simply $\supp{\varrho_\sfP}$ itself. To show that $W^{x_\sfP}_\sfP$ acts isometrically on $\supp{\varrho_\sfP}$, observe that 
	\begin{align}  \label{eq:sumovera}
	[	W^{x_\sfA}_\sfA \otimes  W^{x_\sfB}_\sfB \otimes \bone_{\cP}] \ket{\psi} =  \tilde{\ket{\psi}} \otimes \ket{\psi^\res} ,
	\end{align}
 by summing over $y \in Y$ in \cref{eq:vers2}. Consequently, 
\begin{align} \label{eq:Wiso}
\bra{\psi}	[\abs{W^{x_\sfA}_\sfA}^2 \otimes \abs{W^{x_\sfB}_\sfB}^2 \otimes \bone_{\cP}] \ket{\psi} =  1,
\end{align}
where $\abs{W^{x_\sfP}_\sfP}^2 = (	W^{x_\sfP}_\sfP)^*	W^{x_\sfP}_\sfP$. The operators $\abs{W^{x_\sfP}_\sfP}^2$ are positive, and by definition of $W^{x_\sfP}_\sfP$ satisfy $\abs{W^{x_\sfP}_\sfP}^2  \leq \bone_{\cH_\sfP}$. Therefore, \cref{eq:Wiso} entails that $\abs{W^{x_\sfP}_\sfP}^2$ acts as $\bone_{\cH_\sfP}$ on the subspace $\supp{\varrho_\sfP}$, which is precisely to say that $W^{x_\sfP}_\sfP$ acts isometrically on this subspace. Altogether, (2) holds as desired.

	Next, assume that (2) holds; we show (3). Since (2) implies (1), our introductory observation still applies: the subspace $\cH^0_\sfP$ generated from $\supp{\varrho_\sfP}$ by the operators $\Pi^{x_\sfP}_\sfP(y_\sfP)$ coincides with $\supp{\varrho_\sfP}$ itself. Therefore, in order to show the existence of $x_\sfP$-independent isometries $W_\sfP$ satisfying \cref{eq:vers3}, it suffices to show that for any two $x_\sfP, x'_\sfP \in X_\sfP$, the isometries $W^{x_\sfP}_\sfP$ and $W^{x'_\sfP}_\sfP$ from \cref{eq:vers2} act identically on the subspace $\supp{\varrho_\sfP}$. Let us assume without loss of generality that $\sfP = \sfA$. Let $\sum_{j=1}^r \sqrt{p(j)} \ket{\psi_{\sfA}(j)} \otimes \ket{\psi_{\neg \sfA}(j)}$ be a Schmidt decomposition of $\ket{\psi}$ with $p(1), \ldots, p(r)> 0$, relative to the factorisation $\cH_{\sfA} \otimes \cH_{\neg {\sfA}}$, where $\cH_{\neg \sfA} = \cH_\sfB \otimes \cP$. The condition (2) entails \cref{eq:sumovera}, which now reads 
\begin{align}	\label{eq:qi0collapsed}
  \sum_{j=1}^r \sqrt{p(j)} W^{x_\sfA} \ket{\psi_{\sfA}(j)} \otimes [W^{x_{\sfB}}_{\sfB}  \otimes \bone_\cP] \ket{\psi_{\neg \sfA}(j)}= \ket*{\tilde{\psi}} \otimes \ket*{\psi^\res} 
  \end{align}
for any $x_\sfA \in X_\sfA$ and $x_\sfB \in X_\sfB$. Fix some $x_\sfB \in X_\sfB$. By \cref{eq:qi0collapsed}, the vector $W^{x_\sfA} \ket{\psi_{\sfA}(j)}$ must be independent of $x_\sfA$, for each $j=1, \ldots, r$; this follows by applying the operator $\bone_{\tilde{\cH}_\sfA \otimes \cH^\res_\sfA} \otimes \bra{\psi_{\lnot \sfA}(j)} [(W^{x_\sfB}_\sfB)^* \otimes \bone_\cP]$ to both sides. Observing that $\up{span}\{\ket{\psi_{\sfA}(j)} \mid j=1, \ldots, r  \} = \up{supp}(\varrho_{\sfA})$, we thus conclude that the operator $W^{x_\sfA}_\sfA$ restricted to $\up{supp}(\varrho_\sfA)$ is independent of $x_{\sfA}$. This shows the desired, and altogether finishes the proof. \end{proof}

 \cref{lem:technical} concludes the presentation of our new operational self-testing definition and its relation to conventional self-testing. Let us end by briefly considering how our operational definition might lend itself to the setting of robust self-testing.

\begin{Remark}[Approximate Simulation and Robust Self-Testing] \label{rem:Approx} The aim of \emph{robust self-testing} \cite{MYS12,SB19} is to establish that if a strategy $S$ has a behaviour which is `close' to the behaviour $P$, then $S$ is `close' to being reducible to some canonical strategy $\tilde{S}$. This notion calls for a definition of \myuline{approximate} reducibility, and this is conventionally achieved by considering the Hilbert space norm of the difference between the left and right hand sides in the reducibility condition \eqref{eq:RedAttlComp}. Of course, the operational interpretation of this may be even more questionable than that of the exact case, and since \eqref{eq:RedAttlComp} comprises several identities (one for each $x \in X$ and $y \in Y$) it is not obvious how to form a single numerical quantity which indicates closeness. 

Our operational definition of local simulation, on the other hand, quite naturally lends itself to approximate generalisations. If $D$ is a metric on quantum channels, then \cref{def:Sim} suggests the following notion of \emph{approximate local simulation (w.r.t.~$D$)}: For $\varepsilon \geq 0$, we say that \emph{$S$ locally $\varepsilon$-simulates $\tilde{S}$ (w.r.t.~$D$)}, if for any implementation of $\tilde{S}$ there exists an implementation of $S$ which is $\varepsilon$-close to it, as measured by $D$. 

 If the metric $D$ is sufficiently well-behaved, then the chain of arguments which proved the equivalence between \cref{def:Sim} and the `purified' condition \eqref{eq:SimPure} will carry over to the approximate setting. Specifically (assuming for simplicity that $\tilde{S}$ is pure-state), it will hold that $S$ locally $\varepsilon$-simulates $\tilde{S}$ if and only if there exist isometric channels $\breve{\Gamma}_\sfA$, $\breve{\Gamma}_\sfB$ and a pure state $\psi^\up{res}$ such that 
\begin{align} \label{eq:SS}
      	\myQ{0.7}{0.7}{
		& \push{\what{X}_\sfA}  \qw  & \multigate{1}{ \Sigma_\sfA} & \qw &  \push{\what{Y}_\sfA}  \qw & \qw  \\
		& \Nmultigate{4}{\psi}  & \ghost{\Sigma_\sfA}  & \push{\cE_\sfA} \ww &  \Nmultigate{1}{\breve{\Gamma}_\sfA}{\ww} & \push{\tilde{\cE}_\sfA}  \ww & \ww \\
		& \Nghost{\psi}& & & \Nghost{\breve{\Gamma}_\sfA} & \push{\cH^\res_\sfA} \ww & \ww\\
		& \Nghost{\psi}& \ww & \ww & \ww & \push{{\cP}} \ww & \ww\\
		&\Nghost{\psi} & & & \Nmultigate{1}{\breve{\Gamma}_\sfB} &   \push{\cH^\res_\sfB} \ww & \ww\\
		& \Nghost{\psi} & \multigate{1}{\Sigma_\sfB} & \push{\cE_\sfB} \ww &  \Nghost{\breve{\Gamma}_\sfB}{\ww} & \push{\tilde{\cE}_\sfB}  \ww & \ww  \\
		&   \push{\what{X}_\sfB} \qw  & \ghost{\Sigma_\sfB} & \qw & \push{\what{Y}_\sfB}  \qw & \qw  \\
	} 
\quad 	 \approx^D_\varepsilon \quad 
	\myQ{0.7}{0.7}{
		& \push{\what{X}_\sfA}  \qw  & \multigate{1}{ \tilde{\Sigma}_\sfA} & \push{\what{Y}_\sfA}  \qw & \qw  \\
		& \Nmultigate{4}{\tilde{\psi}}  & \ghost{\tilde{\Sigma}_\sfA} & \push{\tilde{\cE}_\sfA}  \ww & \ww \\
		& \Nghost{\tilde{\psi}} &  \Nmultigate{2}{\psi^\textup{res}} &   \push{\cH^\res_\sfA} \ww & \ww \\
		&	\Nghost{\psi'} & \Nghost{\psi^\textup{res}} & \push{{\cP}}\ww & \ww  \\ 
		& \Nghost{\psi'} &  \Nghost{\psi^\textup{res}} &   \push{\cH^\res_\sfB} \ww & \ww \\
		& \Nghost{\tilde{\psi}} & \multigate{1}{\tilde{\Sigma}_\sfB}  & \push{\tilde{\cE}_\sfB}  \ww & \ww \\
		&   \push{\what{X}_\sfB} \qw  & \ghost{\tilde{\Sigma}_\sfB} & \push{\what{Y}_\sfB}  \qw & \qw  \\
	} ,\end{align}
where `$\approx^D_\varepsilon$' signifies that the two channels are $\varepsilon$-close as measured by $D$. Examples of a metric which is well-behaved in this sense are the \emph{Bures distance} \cite{KSW08math,KSW08phys} and the \emph{purified diamond-distance} \cite{Hou21} on quantum channels. (For a systematic discussion of metrics and their properties, see Chapter 3 in Ref.~\cite{Hou21}.\footnote{In fact, it follows from the monotonicity and dilational properties considered therein that $\varepsilon$-approximate version of \cref{def:Sim} is equivalent to \cref{eq:SS} for some isometric channels $\breve{\Gamma}_\sfA$, $\breve{\Gamma}_\sfB$ and some (not necessarily pure) state $\psi^\up{res}$; this is a straightforward modification of the exact argument. One can then argue that the state $\psi^\up{res}$ may be assumed pure without increasing the distance.}) 

Now, in a usual self-testing scenario, the inputs provided to the classical embedded systems $\what{X}_\sfA$ and $\what{X}_\sfB$ are all classical, so it is natural to consider instead of \cref{eq:SS} the weaker condition $ \max_{x \in X} D(\phi_x, \tilde{\phi}_x) \leq \varepsilon$, where $\phi_x$ and $\tilde{\phi}_x$ are the pure states which result from inserting input $\ketbra{x} = \ketbra{x_\sfA} \otimes \ketbra{x_\sfB}$ into the channels on the left and right hand sides of \cref{eq:SS}, respectively. In other words, if $V_\sfP: \cE_\sfA \to \tilde{\cE}_\sfP \otimes \cH^\res_\sfP$ are isometries representing $\breve{\Gamma}_\sfP$, for $\sfP \in \{\sfA, \sfB\}$,  and if $V= V_\sfA \otimes V_\sfB$, then the condition 
\begin{align} \label{eq:ApproxClas}
    \inf_{V_\sfA, V_\sfB, \ket*{\psi^\res}} \max_{x \in X} D(\phi_x, \tilde{\phi}_x) \leq \varepsilon
\end{align}
  is indicative of local $\varepsilon$-simulation of $\tilde{S}$ by $S$, where 
 \begin{align}
\resizebox{.92\hsize}{!}{$\displaystyle \ket{\phi_x} = \sum_{ y \in Y} [ V( \ket{x} \otimes \Pi^x(y)) \otimes \bone_\cP ] \ket{\psi} \otimes \ket{y} \quad \text{and} \quad \ket*{\tilde{\phi}_x} =\sum_{ y \in Y} \ket{x} \otimes \tilde{\Pi}^x(y) \ket*{\tilde{\psi}}\otimes \ket*{\psi^\res}  \otimes \ket{y}$} .
\end{align}
(We here assume the convenient choice of Stinespring dilations from the proof of \cref{prop:OpSim}, cf.~Eqs.~\eqref{eq:ConvStine} and \eqref{eq:OpId}. In particular, $\cE_\sfP = \what{X}_\sfP \otimes \cH_\sfP$ and $\tilde{\cE}_\sfP = \what{X}_\sfP \otimes \tilde{\cH}_\sfP \otimes \cH^\res_\sfP$.) 

The details of the condition \eqref{eq:ApproxClas} now of course depend on the metric $D$, but a general phenomenon is that the local isometry $V_\sfP$ is allowed to observe the local input $x_\sfP$. More formally, suppose we restrict the infimum in \cref{eq:ApproxClas} to isometries of the form $V_\sfP = \sum_{x_\sfP \in X_\sfP} \ketbra{x_\sfP} \otimes W^{x_\sfP}_\sfP$ with $W^{x_\sfP}_\sfP : \cH_\sfP \to \tilde{\cH}_\sfP \otimes \cH^\res_\sfP$ an isometry for each $x_\sfP \in X_\sfP$. We then have $\ket{\phi_x}= \sum_{ y \in Y} \ket{x} \otimes [ W^x \Pi^x(y) \otimes \bone_\cP ] \ket{\psi} \otimes \ket{y}$, so $\ket{x}$ factors out from both $\ket{\phi_x}$ and $\ket*{\tilde{\phi_x}}$. Natural metrics $D$ will be invariant under a fixed tensoring with $\ket{x}$, so the condition \eqref{eq:ApproxClas} holds if merely
\begin{align} \label{eq:Dquant}
\resizebox{.92\hsize}{!}{$\displaystyle \inf_{\substack{(W^{x_\sfA}_\sfA)_{x_\sfA \in X_\sfA},  \\ (W^{x_\sfB}_\sfB)_{x_\sfB \in X_\sfB}, \ket{\psi^\res}} } \max_{x \in X} D\left(\sum_{ y \in Y} [ W^x\Pi^x(y)) \otimes \bone_\cP ] \ket{\psi} \otimes \ket{y}, \tilde{\Pi}^x(y) \ket*{\tilde{\psi}}\otimes \ket*{\psi^\res}  \otimes \ket{y} \right) \leq \varepsilon$}.
\end{align}

Now, the infimum in \cref{eq:Dquant} might be notably smaller than the infimum over $x_\sfP$-\myuline{independent} isometries $W_\sfP$, indeed the argument used in the proof of \cref{lem:technical} to eliminate $x_\sfP$-dependence in the exact case ($\varepsilon = 0$) does not immediately carry over to the general approximate case ($\varepsilon > 0$). Thus, our notion of approximate simulation might ultimately yield robust self-testing results genuinely different from ones based on the conventional reducibility definition with input-independent isometries.

To finish off with an idea of what the quantity \eqref{eq:Dquant} might be like, let us take $D$ as the purified diamond distance $P_\diamond$ \cite{Hou21}, so that the $D$-distances appearing in \cref{eq:Dquant} are purified distances between states \cite{Toma10,Toma12}. As the involved states are already pure, the purified distances coincide with trace distances, and the trace distance between pure states is always upper bounded by the Hilbert space distance between vector representatives.\footnote{ In fact, we have for unit vectors $\ket{\chi}$ and $\ket{\chi'}$ the identity $\delta_1(\ket{\chi}, \ket{\chi'})=\frac{1+\abs{\braket{\chi}{\chi'}}}{2} \inf_{\theta \in [0, 2 \pi)} \norm{\ket{\chi}- e^{i \theta} \ket{\chi'}}_2$, so we have $\delta_1(\ket{\chi}, \ket{\chi'}) \approx \inf_{\theta \in [0, 2 \pi)} \norm{\ket{\chi}- e^{i \theta} \ket{\chi'}}_2$ when $\ket{\chi}$ and $\ket{\chi'}$ have large overlap.} Hence,
\begin{align}
& D\left(\sum_{ y \in Y} [ W^x\Pi^x(y)) \otimes \bone_\cP ] \ket{\psi} \otimes \ket{y}, \tilde{\Pi}^x(y) \ket*{\tilde{\psi}}\otimes \ket*{\psi^\res}  \otimes \ket{y} \right) \\ & \leq 
\norm{ \sum_{ y \in Y} [ W^x\Pi^x(y)) \otimes \bone_\cP ] \ket{\psi} \otimes \ket{y}- \tilde{\Pi}^x(y) \ket*{\tilde{\psi}}\otimes \ket*{\psi^\res}  \otimes \ket{y}}_2 \\ & = \sqrt{\sum_{ y \in Y} \norm{[ W^x\Pi^x(y)) \otimes \bone_\cP ] \ket{\psi}  - \tilde{\Pi}^x(y) \ket*{\tilde{\psi}}\otimes \ket*{\psi^\res}  }^2_2},
\end{align}
and \eqref{eq:Dquant} therefore follows if
\begin{align} 
    \inf_{\substack{(W^{x_\sfA}_\sfA)_{x_\sfA \in X_\sfA},  \\ (W^{x_\sfB}_\sfB)_{x_\sfB \in X_\sfB}, \ket{\psi^\res}} } \max_{x \in X}  \sqrt{\sum_{ y \in Y} \norm{[ W^x\Pi^x(y)) \otimes \bone_\cP ] \ket{\psi}  - \tilde{\Pi}^x(y) \ket*{\tilde{\psi}}\otimes \ket*{\psi^\res}  }^2_2} \leq \varepsilon.
\end{align}
\end{Remark}

For other considerations about the metric aspects of simulation, in particular as regards the development to come in \cref{subsec:Causal}, the reader is referred to the introductory remarks in Chapters 4 and 5 of Ref.~\cite{Hou21}. In particular, it is discussed there that an additional potentially desirable property of a metric is monotonicity under \emph{contraction} of certain input wires with output wires, corresponding to a non-increasing distance in \cref{eq:SS} if e.g.~we perform the procedure of processing and feeding the output from $\what{Y}_\sfA$ as input to $\what{X}_\sfB$.

\subsection{Symmetrising Simulation and Turning the Tables...}

\label{subsec:AssSim}

In \cref{subsec:Sim}, we introduced the operational notions of local simulation and of self-testing a.t.l.s., and we proved a formal relationship to the conventional operator-algebraic notions of reducibility and self-testing. From now on, we shall no longer look back to the operator-algebraic notions but rather work entirely with the operational counterparts. \\

This subsection serves to introduce a relaxation of local simulation, which we call \emph{local assisted simulation} (\cref{def:LocAss}). Like local simulation, it is an operationally defined pre-order on strategies. By features of Stinespring dilations, however, local assisted simulation is in quantum theory in fact an equivalence relation (\cref{thm:AssSym}).

As such, the relation of local assisted simulation, which as far as we know has not been studied from an operator-algebraic viewpoint, uncovers quantum self-testing as an explicitly two-fold phenomenon --- see in particular \cref{rem:SelfEquiv}, \cref{def:SelfTestAtlas} and the surrounding discussion.

\begin{Definition}[Local Assisted Simulation] \label{def:LocAss}
    Let $S$ and $S'$ be strategies for $\fB$, with respective component-wise Stinespring dilations $(\psi, \Sigma_\sfA, \Sigma_\sfB)$ and $(\psi', \Sigma'_\sfA, \Sigma'_\sfB)$. We say that \emph{$S$ locally assisted simulates $S'$}, written $S \geq_{l.a.s.} S'$, if there exist channels $\Gamma_0$, $\Gamma_\sfA$ and $\Gamma_\sfB$ and some state $\alpha$ such that 
    \begin{align} \label{eq:asssim}
       \myQ{0.7}{0.7}{
 &  \push{\what{X}_\sfA}   \qw  & \multigate{1}{\Sigma_\sfA}  &   \qw &  \push{\what{Y}_\sfA} \qw &  \qw  & \qw \\
	& \Nmultigate{5}{\psi}  & \ghost{\Sigma_\sfA} &  \push{\cE_\sfA} \ww &  \Nmultigate{1}{\Gamma_\sfA}{\ww} &\push{\cE'_\sfA} \ww  & \ww\\
	& \Nghost{\psi}   &  & \Nmultigate{3}{ \alpha } & \Nghost{\Gamma_\sfA}{\ww} &  \\  
	& \Nghost{\psi}  &\push{\cE_0}  \ww &   \ww &   \Nmultigate{1}{\Gamma_0}{\ww} &\push{\cE'_0} \ww  & \ww  \\
	& \Nghost{\psi}  &   & \Nghost{ \alpha } &   \Nghost{\Gamma_0}{\ww}   \\
	& \Nghost{\psi}  &  & \Nghost{ \alpha } &   \Nmultigate{1}{\Gamma_\sfB}{\ww}   \\
	& \Nghost{\psi}  & \multigate{1}{\Sigma_\sfB} &  \push{\cE_\sfB} \ww & \Nghost{\Gamma_\sfB}{\ww} &\push{\cE'_\sfB} \ww  & \ww  \\ 
 &  \push{\what{X}_\sfB}  \qw  & \ghost{\Sigma_\sfB}   & \qw &  \push{\what{Y}_\sfB} \qw &   \qw & \qw }
 \quad = \quad  
 \myQ{0.7}{0.7}{
	& \push{\what{X}_\sfA}  \qw  & \multigate{1}{ \Sigma'_\sfA} & \push{\what{Y}_\sfA}  \qw & \qw  \\
	& \Nmultigate{2}{\psi'}  & \ghost{\Sigma'_\sfA} &\push{\cE'_\sfA} \ww  & \ww\\
	& \Nghost{\psi'} & \ww &\push{\cE'_0} \ww  & \ww \\
	& \Nghost{\psi'} & \multigate{1}{\Sigma'_\sfB} &\push{\cE'_\sfB} \ww  & \ww \\
	&   \push{\what{X}_\sfB} \qw  & \ghost{\Sigma'_\sfB} & \push{\what{Y}_\sfB}  \qw & \qw \\
} .
    \end{align}
   (The reader is to imagine that the middle wiggly wire on the right hand side passes above the state $\alpha$, so that the pictorial representation is really three-dimensional.)
\end{Definition}

The interpretation of local assisted simulation is that, using a pre-distributed state $\alpha$ as `assistant', side-information of the strategy $S'$ can be generated from side-information of the strategy $S$.

\begin{Remark}[On the Definition of Local Assisted Simulation] \label{rem:AssOp} Local assisted simulation can be given an alternative definition in the style of \cref{def:Sim}, quantifying over general implementations: In fact, $S \geq_{l.a.s.} S'$ precisely if there exists a state $\alpha$ such that any general implementation $\scalemyQ{.7}{0.7}{0.7}{
	& \push{\what{X}_\sfA}  \qw  & \multigate{1}{ \Phi'_\sfA} & \push{\what{Y}_\sfA}  \qw & \qw  \\
	& \Nmultigate{2}{\xi'}  & \ghost{\Phi'_\sfA} &  \push{\cE'_\sfA}\ww & \ww\\
	& \Nghost{\xi'} & \ww  & \push{\cE'_0} \ww & \ww\\
	& \Nghost{\xi'} & \multigate{1}{\Phi'_\sfB} & \push{\cE'_\sfB }\ww & \ww  \\
	&   \push{\what{X}_\sfB} \qw  & \ghost{\Phi'_\sfB} & \push{\what{Y}_\sfB}  \qw & \qw \\
}$ of $S'$ is matched by $ \scalemyQ{.7}{0.7}{0.7}{
 &  \push{\what{X}_\sfA}   \qw  & \multigate{1}{\Phi_\sfA}  &   \qw &  \push{\what{Y}_\sfA} \qw &  \qw  & \qw \\
	& \Nmultigate{5}{\xi}  & \ghost{\Phi_\sfA} & \push{\cE_\sfA}  \ww &  \Nmultigate{1}{\Gamma_\sfA}{\ww} & \push{\cE'_\sfA} \ww & \ww  \\
	& \Nghost{\xi}   &  & \Nmultigate{3}{ \alpha } & \Nghost{\Gamma_\sfA}{\ww} &  \\  
	& \Nghost{\xi}  & \push{\cE_0} \ww &   \ww &   \Nmultigate{1}{\Gamma_0}{\ww} & \push{\cE'_0} \ww & \ww  \\
	& \Nghost{\xi}  &   & \Nghost{ \alpha } &   \Nghost{\Gamma_0}{\ww}   \\
	& \Nghost{\xi}  &  & \Nghost{ \alpha } &   \Nmultigate{1}{\Gamma_\sfB}{\ww}   \\
	& \Nghost{\xi}  & \multigate{1}{\Phi_\sfB} &  \push{\cE_\sfB} \ww & \Nghost{\Gamma_\sfB}{\ww} & \push{\cE'_\sfB} \ww & \ww  \\ 
 &  \push{\what{X}_\sfB}  \qw  & \ghost{\Phi_\sfB}   & \qw &  \push{\what{Y}_\sfB} \qw &   \qw & \qw }$ for some implementation of $S$ and some choice of $\Gamma_\sfA$, $\Gamma_\sfB$ and $\Gamma_0$. This means in particular that local assisted simulation can be defined in general operational theories (it also means that \cref{def:LocAss} does not depend on the choice of Stinespring dilations). We chose the formulation in \cref{def:LocAss} so as to avoid having to prove an equivalent of \cref{prop:Rechar} afterwards.   \end{Remark}
 
 \begin{Prop}
Local assisted simulation is a pre-order on the class of strategies for $\fB$, $\Strat{\fB}$. It coarsens the pre-order of local simulation, i.e.~the relation $S \geq_{l.s.} S'$ implies the relation $S \geq_{l.a.s.} S'$.   
\end{Prop}

\begin{proof}
   Reflexivity and transitivity are obvious. The second statement follows by taking $\alpha$ product. 
\end{proof}

 \begin{Remark}[Local Assisted Simulation in Terms of Local Simulation] \label{rem:LocAssLoc}
 Recall from \cref{ex:Augmented} the concept of \emph{augmenting} a strategy $S$ by a state $\gamma$ to form the new strategy $S[\gamma]$. It is worth observing that $S$ locally assisted simulates $S'$ with assistant $\alpha$ if and only if $S[\alpha_{\sfA \sfB}]$ locally simulates $S'$ (without assistance), where $\alpha_{\sfA \sfB}$ is the marginal state arising from $\alpha$ by discarding the middle system. When $\alpha$ is pure, this statement is visually evident from \cref{eq:asssim}, by grouping $\psi$ and $\alpha$ together as a single state which is then a purification of $\varrho \otimes \alpha_{\sfA \sfB}$. It follows from this observation that $S \geq_{l.a.s.} S'$ if and only if there exists some state $\gamma$ such that $S[\gamma] \geq_{l.s.} S'$.
 \end{Remark}

The following result relies on the features of Stinespring dilations exposed in \cref{subsec:Stinespring}:  

\begin{Thm} \label{thm:AssSym}
In quantum theory, $\geq_{l.a.s.}$ is symmetric and thus an equivalence relation on $\Strat{\fB}$.
\end{Thm}

\begin{proof}
    Suppose that $S \geq_{l.a.s.} S'$. As in the proof of \cref{prop:OpSim}, we may on the left hand side of \cref{eq:asssim} replace $\Gamma_0$, $\Gamma_\sfA$ and $\Gamma_\sfB$, by respective Stinespring dilations $\breve{\Gamma}_0$, $\breve{\Gamma}_\sfA$, and $\breve{\Gamma}_\sfB$; the resulting channel is then a dilation of the right hand side, which is isometric and hence dilationally pure, and it must therefore equal the channel obtained by tensoring the right hand side with a (pure) state, say $\beta$. Now, since $\breve{\Gamma}_0$, $\breve{\Gamma}_\sfA$, and $\breve{\Gamma}_\sfB$ are isometric, they are \myuline{reversible} by \cref{lem:IsoRev}, and can thus be cancelled by left-inverses $\breve{\Gamma}^-_0$, $\breve{\Gamma}^-_\sfA$, and $\breve{\Gamma}^-_\sfB$. If we do this and then trace out the state $\alpha$, we exhibit the relation $S' \geq_{l.a.s.} S$ with $\beta$ as assistant. This proves symmetry, as desired.   
\end{proof}

By virtue of \cref{thm:AssSym}, we will sometimes use the notation $\sim_{l.a.s.}$ rather than $\geq_{l.a.s.}$. Specifically, we use the notation $\sim_{l.a.s.}$ when discussing the relation in quantum theory, whereas we resort to the original notation $\geq_{l.a.s.}$ when discussing the relation generally, where the order might matter.

\begin{Remark} \label{rem:SelfEquiv}
 \cref{thm:AssSym} can be strengthened to assert that $\sim_{l.a.s.}$ is the equivalence relation \emph{generated by} the pre-order $\geq_{l.s.}$, that is, it is the finest equivalence relation which coarsens $\geq_{l.s.}$. Indeed, let $\cong$ be any equivalence relation which coarsens $\geq_{l.s.}$. We must show that $S \sim_{l.a.s.} S'$ implies $S \cong S'$. By \cref{rem:LocAssLoc},  $S[\gamma] \geq_{l.s.} S'$ for some state $\gamma$. By \cref{ex:Augmented},  $S[\gamma] \geq_{l.s.} S$. As $\cong$ coarsens $\geq_{l.s.}$, these two conditions imply $S' \cong S[\gamma]$ and $S[\gamma] \cong S$, respectively. Consequently, $S \cong S'$.\end{Remark}

An important consequence of \cref{thm:AssSym} is that local simulation implies equivalence under local assisted simulation. This fact substantiates the intuition of authors who use the term `equivalence of strategies' when referring to reducibility (cf.~\cref{rem:RedEquiv}): though reducibility is not itself a symmetric relation, quantum strategies related by reducibility really \myuline{are} mathematically equivalent, because they are related by local simulation and thus equivalent under local assisted simulation. For example, consider the CHSH-game. Since the canonical strategy $\tilde{S}$ can be self-tested according to local simulation, our result implies that all strategies with the maximal violation are \myuline{equivalent} under local assisted simulation. According to \cref{rem:LocAssLoc}, every optimal strategy can therefore be locally simulated by the canonical one after tensoring an appropriate assisting state to the canonical (singlet) state.

From a different perspective, this observation reveals that quantum self-testing of a canonical strategy $\tilde{S}$ may be considered a two-fold phenomenon: it is not just the phenomenon that any strategy $S$ with behaviour $P$ is `at least as strong as $\tilde{S}$' (namely, $S \geq_{l.s.} \tilde{S}$), it is \myuline{also} the phenomenon that any $S$ with behaviour $P$ is `no stronger than $\tilde{S}$' (namely, $\tilde{S} \geq_{l.a.s.} S$). As such, self-testing on the one hand lower bounds the capabilities of unknown strategies and on the other hand upper bounds these same capabilities. In quantum theory, it just so happens that the lower bound implies an upper bound automatically.

To isolate the effect of the upper bound, it is beneficial to introduce the following notion of self-testing:

\begin{Definition}[Self-Testing according to Local Assisted Simulation.] \label{def:SelfTestAtlas}
 Let $\tilde{S}$ be any strategy with behaviour $P$. We say that \emph{$P$ self-tests $\tilde{S}$ according to local assisted simulation} (for short, \emph{a.t.l.a.s.}) if any strategy $S$ with behaviour $P$ is locally assisted simulated by $\tilde{S}$, i.e.~$S \leq_{l.a.s.} \tilde{S}$. \end{Definition}
 
 \begin{Remark}
 In quantum theory, self-testing a.t.l.a.s.~is equally well defined by the requirements $S \leq_{l.a.s.} \tilde{S}$ or $S \sim_{l.a.s.} \tilde{S}$ for strategies $S$ with behaviour $P$. Indeed, these two requirements are identical, by \cref{thm:AssSym}. In other theories, however, this may not be so and we intend the definition of self-testing a.t.l.a.s.~to mean  `$S \leq_{l.a.s.} \tilde{S}$' rather than equivalence under local assisted simulation. 
 \end{Remark}

If $P$ self-tests a strategy $\tilde{S}$ a.t.l.s., then $P$ self-tests $\tilde{S}$ a.t.l.a.s. In fact, $P$ in this case self-tests \myuline{every} strategy $\tilde{S}'$ with behaviour $P$ a.t.l.a.s., since all such strategies are equivalent. We shall see in \cref{sec:Utility} that certain features of self-testing which have traditionally been attributed to the conventional reducibility relation can actually be derived from our weaker notion of self-testing a.t.l.a.s.~(see in particular \cref{prop:MeasExtract} and \cref{prop:SelfSec}).

The following we leave as an open problem:

\begin{OP} \label{op:SimpleRep}
 Suppose that $P$ self-tests a strategy $\tilde{S}$ a.t.l.a.s. Does there necessarily exist a strategy $\tilde{S}_0$ such that $P$ self-tests $\tilde{S}_0$ a.t.l.s.?
\end{OP}


If \cref{op:SimpleRep} has a positive answer, then quantum self-testing a.t.l.a.s.~may be considered as fundamental as quantum self-testing a.t.l.s.(and a.t.r., by \cref{thm:SelfRvsS}). The only difference between self-testing a.t.l.a.s.~and a.t.l.s.~is then that the latter highlights a certain canonical strategy as especially simple. We emphasize that the canonical CHSH-strategy (and indeed any other conventionally self-tested strategy) do not constitute a counterexample to the open problem, since these strategies already satisfy the strongest notion of self-testing according to local simulation.

\subsection{... all the Way: Quantum Self-Testing in a Larger Perspective}
\label{subsec:Causal}

We have now seen how to recast the conventional definition of quantum self-testing in an operational language (\cref{subsec:Sim}), and how this language uncovers quantum self-testing as a two-fold phenomenon (\cref{subsec:AssSim}).  The operational reformulation in \cref{subsec:Sim} suggests an interpretation in terms of side-information which escapes to an environment according to a particular causal structure. Our considerations in \cref{subsec:AssSim}, meanwhile, led to the idea of a self-testing notion which aims exclusively to \myuline{upper} bound the power of an unknown strategy and is unconcerned with \myuline{lower} bounding its capabilities, cf.~\cref{def:SelfTestAtlas}.  In this final subsection, we explain how to take these ideas even further by realising quantum self-testing in a general framework of causally structured circuits of channels and dilations.

We start in \cref{subsubsec:Caussim} by introducing a relaxation of local assisted simulation, which we call \emph{causal simulation} (\cref{def:CausSim}), and consider to be the ultimate mode of simulation among strategies. It gives rise to a notion of \emph{self-testing according to causal simulation} (for short, \emph{a.t.c.s.}), cf.~\cref{def:SelfTestCaus}. When the canonical strategy $\tilde{S}$ is pure-state causal simulation coincides with local assisted simulation (\cref{rem:CausPure}), however in general it is less restrictive as illustrated by \cref{ex:UnconSelf}. The point is now that self-testing a.t.c.s.~is an instance of a phenomenon which belongs to a larger theory of \emph{causal channels} and \emph{derivability} among \emph{causal dilations}, which we explain in \cref{subsubsec:Causchan} and \cref{subsubsec:Causdil}. We show in \cref{subsubsec:SpecialBell} how to recover quantum self-testing according to causal simulation as a special case within this framework. 

The presentation of the framework of causal channels and causal dilations and of its connection to quantum self-testing is brief and at points somewhat informal. For further details, and comparisons to other frameworks for networks of channels \cite{Chir09combs,Peri17,Kiss17}, the reader is referred to Chapters 4 and 5 in Ref.~\cite{Hou21}.

\subsubsection{Causal Simulation}

\label{subsubsec:Caussim}

Let us introduce the following relaxation of local assisted simulation:

\begin{Definition}[Causal Simulation] \label{def:CausSim}

Let $S$ and $S'$ be strategies for the Bell-scenario $\fB$, with respective component-wise Stinespring dilations $(\psi, \Sigma_\sfA, \Sigma_\sfB)$ and $(\psi', \Sigma'_\sfA, \Sigma'_\sfB)$. We say that \emph{$S$ causally simulates $S'$}, written $S \geq_{c.s.} S'$, if there exist channels $\Gamma_0$, $\Gamma_\sfA$ and $\Gamma_\sfB$ such that 
\begin{align} \label{eq:CausSim}
\scalemyQ{1}{0.7}{0.5}{
	&   \push{\what{X}_\sfA} \qw  & \multigate{1}{\Sigma_\sfA}  &   \qw &  \push{\what{Y}_\sfA}  \qw &  \qw  & \qw \\
	& \Nmultigate{4}{\psi}  & \ghost{\Sigma_\sfA} &  \push{\cE_\sfA} \ww &  \Nmultigate{1}{\Gamma_\sfA}{\ww} & \push{\cE'_\sfA} \ww & \ww \\
	& \Nghost{\psi}   &  & \Nmultigate{2}{\Gamma_0} & \Nghost{\Gamma_\sfA}{\ww} &  \\  
	& \Nghost{\psi}  & \push{\cE_0} \ww &   \Nghost{\Gamma_0}{\ww}  & \ww & \push{\cE'_0} \ww & \ww \\
	& \Nghost{\psi}  &  & \Nghost{ \Gamma_0} &   \Nmultigate{1}{\Gamma_\sfB}{\ww}   \\
	& \Nghost{\psi}   & \multigate{1}{\Sigma_\sfB} &  \push{\cE_\sfB} \ww & \Nghost{\Gamma_\sfB}{\ww} & \push{\cE'_\sfB}\ww & \ww \\ 
	&   \push{\what{X}_\sfB} \qw  & \ghost{\Sigma_\sfB}   & \qw &  \push{\what{Y}_\sfB} \qw &   \qw & \qw} =\scalemyQ{1}{0.7}{0.5}{
	& \push{\what{X}_\sfA}\qw  & \multigate{1}{ \Sigma'_\sfA} &  \push{\what{Y}_\sfB}  \qw & \qw  \\
	& \Nmultigate{2}{\psi'}  & \ghost{\Sigma'_\sfA} &  \push{\cE'_\sfA} \ww & \ww\\
	& \Nghost{\psi'} & \ww & \push{\cE'_0} \ww & \ww\\
	& \Nghost{\psi'} & \multigate{1}{\Sigma'_\sfB} & \push{\cE'_\sfB} \ww & \ww  \\
	&  \push{\what{X}_\sfA} \qw  & \ghost{\Sigma'_\sfB} &  \push{\what{Y}_\sfB} \qw & \qw  \\
}  \quad. 
\end{align}
\end{Definition}

\begin{Remark} \label{rem:CausDefOp}
As usual, the Stinespring implementations can be considered stand-ins for general implementations, and the definition of causal simulation is therefore meaningful in all operational theories (and independent of the choice of Stinespring dilations).
\end{Remark}

The interpretation of causal simulation is that the side-information generated from $S'$ can be derived from the side-information generated by $S$ in a way which preserves causality. On the left hand side, the structure of the channel comprised by $\Gamma_0$, $\Gamma_\sfA$ and $\Gamma_\sfB$ is such that every output requires the exact same inputs as on the right hand side --- therefore, the information can be leaked to the environment according to the same schedule on both sides.

\begin{Definition}[Self-Testing according to Causal Simulation.] \label{def:SelfTestCaus}
 Let $\tilde{S}$ be any strategy with behaviour $P$. We say that \emph{$P$ self-tests $\tilde{S}$ according to causal simulation} (for short, \emph{a.t.c.s.}) if  any strategy $S$ with behaviour $P$ is causally simulated by $\tilde{S}$, i.e.~$S \leq_{c.s.} \tilde{S}$. \end{Definition} 

\begin{Remark}[Causal Simulation by Pure-State Strategies] \label{rem:CausPure}
 Note that when a $\tilde{S}$ is a \myuline{pure-state} strategy, the condition $\tilde{S} \geq_{c.s.} S$ is equivalent to the condition $\tilde{S} \geq_{l.a.s.} S$, for any strategy $S$. This is because the purifying system $\tilde{\cE}_0$ of $\tilde{\varrho}$ may be taken trivial in a Stinespring implementation of $\tilde{S}$, and as such the channel $\Gamma_0$ witnessing causal simulation has trivial input and must in fact be a state. For this reason, it also holds that $P$ self-tests $\tilde{S}$ a.t.l.c.~if and only if $P$ self-tests $\tilde{S}$ a.t.l.a.s. 
\end{Remark}

The following example serves to show that causal simulation is generally distinct from local assisted simulation. In fact, it constitutes an example of a strategy $\tilde{S}$ which is self-tested a.t.c.s., but which cannot possible be self-tested a.t.l.a.s. (in particular cannot be self-tested in the conventional operator-algebraic sense), since by \cref{cor:extremal} to be established later its behaviour would then be extremal.

\begin{Example}[Unconventional Self-Testing]\label{ex:UnconSelf}
Consider the Bell-scenario $\fB$ for which $X_\sfA = X_\sfB = \{0\}$ and $Y_\sfA = Y_\sfB = \{0,1\}$. It corresponds to an experimental set-up in which two separated parties are required to respond with bits $y_\sfA \in \{0,1\}$ and $y_\sfB \in \{0,1\}$. There is only one possible input for each party, and we interpret the local reception of this input as a query for the bit. A behaviour for the Bell-scenario $\fB$ is simply a probability distribution $P$ on the set $Y_\sfA \times Y_\sfB = \{0,1\} \times \{0,1\}$. Let us consider the behaviour corresponding to perfectly correlated uniformly random outputs, i.e.~the behaviour which as a classical embedded state on $\what{Y}_\sfA \otimes \what{Y}_\sfB$ is given by $P = \kappa :=  \frac{1}{2} \ketbra{00} + \frac{1}{2} \ketbra{11}$. Consider also the strategy $\tilde{S}$ whose state is $\tilde{\varrho} = \kappa$ on $\C^2 \otimes \C^2$ and whose measurement channels are $\Lambda_\sfA = \Lambda_\sfB =  \Delta_{\{0,1\}}$, which measure in the computational basis of $\C^2$. We claim that $P$ self-tests $\tilde{S}$ according to causal simulation.

Let $S$ be a strategy with behaviour $P$. We must show that $S \leq_{c.s.} \tilde{S}$. Without loss of generality, we may assume that the local quantum systems $\cH_\sfA$ and $\cH_\sfB$ in the strategy $S$ are of the form $\C^2 \otimes \C^{m_\sfA}$ and $\C^2 \otimes \C^{m_\sfB}$ for some $m_\sfA, m_\sfB \in \N$, and that the measurement channels correspond to the projective measurements given by $\Pi_\sfA(y_\sfA) = \ketbra{y_{\sfA}} \otimes \bone_{\C^{m_{\sfA}}}$ and $\Pi_\sfB(y_\sfB) = \ketbra{y_{\sfB}} \otimes \bone_{\C^{m_{\sfB}}}$. This is by \cref{prop:Proj} and the fact that any projective strategy is equivalent under local simulation to one of the said form (by local isometric embedding and a local unitary change of basis). Now, observe that we may even without loss of generality assume that $m_\sfA = m_\sfB=1$, i.e.~that $S$ in fact consists in measuring a pair of qubits in the computational bases. This is because a strategy $S'$ with general values of $m'_\sfA, m'_\sfB \in \N$ is causally simulated by one with $m_\sfA = m_\sfB=1$. Indeed, the measurement channels $\Lambda'_\sfP$ are given by $\Delta_{\{0,1\}} \otimes \tr_{\C^{m'_\sfP}}$, so if $\breve{\Delta}_{\{0,1\}}: \End{\C^2} \to \End{\C^2 \otimes \cE} $ is a Stinespring dilation of $\Delta_{\{0,1\}}$ then $\breve{\Delta}_{\{0,1\}} \otimes \id_{\C^{m'_\sfP}}$ is a Stinespring dilation of $\Lambda'_\sfP$. Hence, $S'$ has Stinespring implementation 
\begin{align}
\scalemyQ{1}{0.7}{0.5}{
	& \push{\triv}\qw  & \qw & \multigate{1}{ \breve{\Delta}_{\{0,1\}} \otimes \id_{\C^{m'_\sfA}}}  & \qw &  \push{\C^2}  \qw & \qw  \\
	& \Nmultigate{2}{\psi'} & \push{\C^2 \otimes \C^{m'_\sfA}} \qw & \ghost{\breve{\Delta}_{\{0,1\}} \otimes \id_{\C^{m'_\sfA}}} &  \push{\cE \otimes \C^{m'_\sfA}} \ww & \ww\\
	& \Nghost{\psi'} & \ww  & \push{\cE'_0}\ww &  \ww \\
	& \Nghost{\psi'} & \push{\C^2 \otimes \C^{m'_\sfB}} \qw &\multigate{1}{\breve{\Delta}_{\{0,1\}} \otimes \id_{\C^{m'_\sfB}}} & \push{\cE \otimes \C^{m'_\sfB}} \ww & \ww  \\
	&  \push{\triv} \qw  & \qw & \ghost{\breve{\Delta}_{\{0,1\}} \otimes \id_{\C^{m'_\sfB}}}  &  \qw & \push{\C^2} \qw & \qw  \\
}
\end{align}
which can trivially be rewritten as 
\begin{align}
\scalemyQ{1}{0.7}{0.5}{
	& \push{\triv}\qw  & \qw & \multigate{1}{ \breve{\Delta}_{\{0,1\}} } & \qw &  \push{\C^2}  \qw & \qw  \\
	& \Nmultigate{4}{\psi'} & \push{\C^2} \qw & \ghost{\breve{\Delta}_{\{0,1\}}} &  \push{\cE} \ww & \Nmultigate{1}{\id}{\ww} & \push{\cE \otimes \C^{m'_\sfA}} \ww & \ww \\
	& \Nghost{\psi'} & & \Nmultigate{2}{\id} & \push{\C^{m'_\sfA}} \ww & \Nghost{\id}{\ww} \\
	& \Nghost{\psi'} & \push{\cE'_0 \otimes \C^{m'_\sfA} \otimes \C^{m'_\sfB}}\ww &  \Nghost{\id}{\ww} & \ww & \push{\cE'_0} \ww & \ww \\
	& \Nghost{\psi'} & & \Nghost{\id}& \push{\C^{m'_\sfB}} \ww & \Nmultigate{1}{\id}{\ww} \\
	& \Nghost{\psi'} & \push{\C^2} \qw &\multigate{1}{\breve{\Delta}_{\{0,1\}} } & \push{\cE} \ww & \Nghost{\id}{\ww} & \push{\cE \otimes \C^{m'_\sfB}} \ww & \ww  \\
	&  \push{\triv} \qw  & \qw & \ghost{\breve{\Delta}_{\{0,1\}}} & \qw &  \push{\C^2} \qw & \qw  \\
},
\end{align}
thus showing the asserted causal simulation by a qubit strategy $S$. (Operationally, information which in $S'$ was not leaked before the measurements has now been moved to the register of information which exists in the environment before receiving the bit queries.) 

To show the desired, it therefore altogether suffices to show that $\tilde{S}$ causally simulates any qubit strategy $S= (\varrho, \Delta_{\{0,1\}}, \Delta_{\{0,1\}})$ with behaviour $P$ whose measurements are in the computational basis. Let us characterise the purifications $\psi$ of states $\varrho$ in such strategies. First, write $\varrho = \sum_{k=1}^n p_k \psi_k$ as a convex combination of pure states. Then, observe that for the strategy $S$ to have behaviour $P$ means that $\kappa=\frac{1}{2}\ketbra{00}+\frac{1}{2} \ketbra{11} =(\Delta_{\{0,1\}}\otimes \Delta_{\{0,1\}} )(\varrho) = \sum_{k=1}^n p_k (\Delta_{\{0,1\}} \otimes \Delta_{\{0,1\}})(\psi_k)$. This can only happen if for $k=1, \ldots, n$ we have
\begin{align}
    (\Delta_{\{0,1\}}\otimes\Delta_{\{0,1\}} )(\psi_k) = q_k \ketbra{00}+(1-q_k) \ketbra{11} 
\end{align} 

for some $q_1, \ldots, q_n \in [0,1]$ with $\sum_{k=1}^n p_k q_k = \frac{1}{2}$. From this we conclude that each $\psi_k$ must have a vector representative of the form $\sqrt{q_k} \ket{00} + \sqrt{1-q_k} e^{i \theta_k} \ket{11}$ with $\theta_k \in [0, 2\pi)$. As such, a purification of $\varrho$, with purifying space $\C^n$, is given by 
\begin{align}
  \ket{\psi} &:=  \sum_{k=1}^m \sqrt{p_k} \ket{\psi_k} \otimes \ket{k} = \sum_{k=1}^m \sqrt{p_k} \big(\sqrt{q_k} \ket{00} + \sqrt{1-q_k} e^{i \theta_k} \ket{11}\big) \otimes \ket{k}  \\ &= \frac{1}{\sqrt{2}} \ket{00} \otimes \ket{\chi_0} + \frac{1}{\sqrt{2}} \ket{11} \otimes \ket{\chi_1}, 
\end{align}
with $\ket{\chi_0}:= \sum_{k=1}^n \sqrt{2p_k q_k} \ket{k}$ and $\ket{\chi_1}:= \sum_{k=1}^n \sqrt{2 p_k (1-q_k)}e^{i \theta_k} \ket{k}$. Since $\sum_{k=1}^n p_k q_k = 1/2$, both $\ket{\chi_0}$ and $\ket{\chi_1}$ are unit vectors. In conclusion, for any such strategy with the correct behaviour, a purification of the state $\varrho$ is given by $\ket{\psi}=\frac{1}{\sqrt{2}} \ket{00}\otimes  \ket{\chi_0} + \frac{1}{\sqrt{2}} \ket{11} \otimes \ket{\chi_1}$ with $\ket{\chi_0}, \ket{\chi_1} \in \C^n$  unit vectors; conversely, whenever $\varrho$ has a purification $\psi$ of this form it actually has the correct behaviour. (Our canonical strategy $\tilde{S}$ incidentally belongs to this class, with $\ket{\chi_j}= \ket{j} \in \C^2$.) 

To finish the proof, we must therefore show that for any such pure state $\psi$ there exist channels $\Gamma_0, \Gamma_\sfA, \Gamma_\sfB$ such that 
\begin{align} \label{eq:csim}
    	\scalemyQ{1}{0.7}{0.5}{
		& \push{\triv} \qw & \multigate{1}{\breve{\Delta}_{\{0,1\}}}  &   \qw &  \push{\C^2}  \qw &  \qw  & \qw \\
		& \Nmultigate{4}{\breve{\kappa}}  & \ghost{\breve{\Delta}_{\{0,1\}}} &  \push{\C^2} \ww &  \Nmultigate{1}{\Gamma_\sfA}{\ww} & \push{\C^2} \ww & \ww \\
		& \Nghost{\breve{\kappa}}   &  & \Nmultigate{2}{\Gamma_0} & \Nghost{\Gamma_\sfA}{\ww} &  \\  
		& \Nghost{\breve{\kappa}}  & \push{\C^2} \ww &   \Nghost{\Gamma_0}{\ww}  & \ww & \push{\C^n} \ww & \ww \\
		& \Nghost{\breve{\kappa}}  &  & \Nghost{ \Gamma_0} &   \Nmultigate{1}{\Gamma_\sfB}{\ww}   \\
		& \Nghost{\breve{\kappa}}   & \multigate{1}{\breve{\Delta}_{\{0,1\}}} &  \push{\C^2} \ww & \Nghost{\Gamma_\sfB}{\ww} & \push{\C^2}\ww & \ww \\ 
		&  \push{\triv} \qw & \ghost{\breve{\Delta}_{\{0,1\}}}   & \qw &  \push{\C^2} \qw &   \qw & \qw} =\scalemyQ{1}{0.7}{0.5}{
		& \push{\triv} \qw & \multigate{1}{ \breve{\Delta}_{\{0,1\}}} &  \push{\C^2}  \qw & \qw  \\
		& \Nmultigate{2}{\psi}  & \ghost{\breve{\Delta}_{\{0,1\}}} &  \push{\C^2} \ww & \ww\\
		& \Nghost{\psi} & \ww & \push{\C^n} \ww & \ww\\
		& \Nghost{\psi} & \multigate{1}{\breve{\Delta}_{\{0,1\}}} & \push{\C^2} \ww & \ww  \\
		& \push{\triv} \qw & \ghost{\breve{\Delta}_{\{0,1\}}} &  \push{\C^2} \qw & \qw  \\
	} ,
\end{align}
where $\breve{\kappa}$ denotes the purification of $\kappa$ with vector representative $\ket{\breve{\kappa}} = \frac{1}{\sqrt{2}} \ket{00} \otimes \ket{0}+\frac{1}{\sqrt{2}} \ket{11} \otimes \ket{1} \in (\C^2 \otimes \C^2) \otimes \C^2$. Indeed, this condition precisely expressed that $\tilde{S}$ causally simulates $S$. It suffices for this to find \myuline{isometric} channels $\Gamma_0$, $\Xi_\sfA$, $\Xi_\sfB$ such that 
\begin{align} \label{eq:easy}
    	\scalemyQ{1}{0.7}{0.5}{
		&  \push{\triv} \qw & \multigate{1}{\breve{\Delta}_{\{0,1\}}}  &   \qw &  \push{\C^2}  \qw &  \qw  & \qw \\
		& \Nmultigate{4}{\breve{\kappa}}  & \ghost{\breve{\Delta}_{\{0,1\}}} &  \push{\C^2} \ww &  \ww & \ww \\
		& \Nghost{\breve{\kappa}}   &  & \Nmultigate{2}{\Gamma_0}  & \ww & \ww  \\  
		& \Nghost{\breve{\kappa}}  & \push{\C^2} \ww &   \Nghost{\Gamma_0}{\ww} & \push{\C^n} \ww & \ww \\
		& \Nghost{\breve{\kappa}}  &  & \Nghost{ \Gamma_0} &  \ww  & \ww   \\
		& \Nghost{\breve{\kappa}}   & \multigate{1}{\breve{\Delta}_{\{0,1\}}} &  \push{\C^2} \ww & \ww & \ww \\ 
		&  \push{\triv} \qw & \ghost{\breve{\Delta}_{\{0,1\}}}   & \qw &  \push{\C^2} \qw &   \qw & \qw} =\scalemyQ{1}{0.7}{0.5}{
		& \push{\triv} \qw & \multigate{1}{ \breve{\Delta}_{\{0,1\}}} &  \push{\C^2}  \qw & \qw  \\
		& \Nmultigate{4}{\psi}  & \ghost{\breve{\Delta}_{\{0,1\}}} &  \push{\C^2} \ww & \Nmultigate{1}{\Xi_\sfA}{\ww} & \push{\C^2} \ww & \ww \\
		& \Nghost{\psi}& & & \Nghost{\Xi_\sfA} & \ww & \ww \\ 
		& \Nghost{\psi} & \ww & \ww & \ww & \push{\C^n} \ww & \ww\\
		& \Nghost{\psi} & & & \Nmultigate{1}{\Xi_\sfB} & \ww & \ww \\
		& \Nghost{\psi} & \multigate{1}{\breve{\Delta}_{\{0,1\}}} & \push{\C^2} \ww & \Nghost{\Xi_\sfB}{\ww} & \push{\C^2} \ww & \ww  \\
		& \push{\triv} \qw& \ghost{\breve{\Delta}_{\{0,1\}}} &  \push{\C^2} \qw & \qw  \\
	} ,
\end{align}
since by reversibility of isometric channels (\cref{lem:IsoRev}), $\Xi_\sfA$ and $\Xi_\sfB$ may then be cancelled to yield \cref{eq:csim}. Satisfying \cref{eq:easy} is easy, however, as we may simply choose $\Gamma_0$ corresponding to the isometry $\ket{j} \mapsto \ket{j} \otimes  \ket{\chi_j}\otimes  \ket{j}$ and $\Xi_{\sfA}$ and $\Xi_\sfB$ both corresponding to the isometry $\ket{j} \mapsto \ket{j} \otimes \ket{j}$ (choosing the systems of the unlabelled wiggly wires as $\C^2$). This concludes the argument and proves that $P$ self-tests $\tilde{S}$ according to causal simulation. \end{Example}

It is worth observing that the strategy $\tilde{S}$ considered in \cref{ex:UnconSelf} is essentially classical and can thus be considered as a strategy in the operational theory of classical information, $\CIT$ (cf.~\cref{sec:Operational}). Since our formulation of self-testing is operational, it is meaningful to ask whether it is also the case that the behaviour $P$ self-tests $\tilde{S}$ a.t.c.s.~\emph{in the theory $\CIT$}. This is indeed case; for a general discussion, see Sec.~4.4.C in Ref.~\cite{Hou21}, in particular the statement is implied by Cor.~4.4.24 (where the term `rigidity' is synonymous with `self-testing a.t.c.s.').\\

In the following, we explain how causal simulation among quantum strategies may be considered a special case of a much more general phenomenon.

\subsubsection{Causal Channels}

\label{subsubsec:Causchan}

Let $\Theory$ be any compositional theory (e.g.~$\Theory= \QIT$). In short, causal channels provide a compact way of dealing with circuits of channels in $\Theory$. We have already used pictures of circuits extensively when describing implementations and the various modes of simulation between quantum strategies, e.g.~as in \cref{eq:CausSim}. In such situations, we always understood that the picture represented a single channel from the input systems on the left to the output systems on the right, rather than representing the specific circuit structure and the individual channels appearing in it. 

For example, a depiction such as 
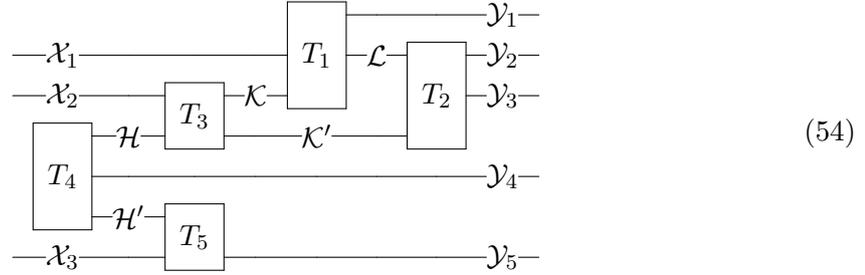
\begin{align} \label{eq:GenChan}
\myQ{0.7}{0.5}{
		&  &  & & &  \Nmultigate{2}{ T_1} & \qw & \qw & \push{\cY_1} \qw & \qw & \\
	 	& \push{\cX_1} \qw  & \qw & \qw& \qw &  \ghost{ T_1} & \push{\cL} \qw &  \multigate{2}{T_2} &   \push{\cY_2} \qw & \qw  \\
	& \push{\cX_2}\qw & \qw  &  \multigate{1}{T_3} &  \push{\cK} \qw & \ghost{T_1} & &  \Nghost{T_2} & \push{\cY_3}\qw & \qw \\
		& \Nmultigate{2}{T_4} &  \push{\cH} \qw & \ghost{T_3} & \qw & \push{\cK'} \qw & \qw & \ghost{T_2} \\
		& \Nghost{T_4} & \qw & \qw & \qw & \qw& \qw  & \qw &\push{\cY_4} \qw & \qw \\
		& \Nghost{T_2} & \push{\cH'} \qw & \multigate{1}{T_5}  \\
	 	&    \push{\cX_3} \qw  & \qw &\ghost{T_5} & \qw & \qw & \qw& \qw &   \push{\cY_5}  \qw &  \qw   \\
} 
\end{align}
would represent a channel $T$ from the system $\cX_1 \otimes \cX_2 \otimes \cX_3$ to the system $\cY_1 \otimes \ldots \otimes \cY_5$, rather than representing the structure of boxes and wires and the individual channels $T_1, \ldots, T_5$ filling these boxes. (This is completely analogous to the way in which an expression such as `$2 \times 3$' is usually taken to represent the number $6$ rather than the sequence of symbols `$2$', `$\times$' and `$3$'.) 

If we wish to address instead the \myuline{circuit} of channels depicted in \cref{eq:GenChan}, we must refer to two pieces of information. One of these is purely graph-theoretic, namely the \emph{stencil} 
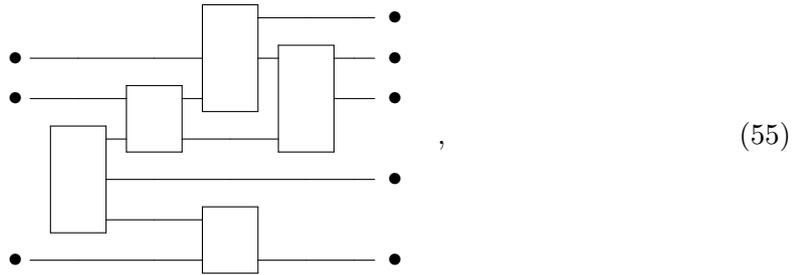
\begin{align} \label{eq:stencil}
\scalemyQ{1}{0.7}{0.5}{
	&	&  &  & \Nmultigate{2}{ \phantom{X}} & \qw & \qw & \qw & \bullet \\
	& \bullet  \quad 	&  \qw  & \qw & \ghost{ \phantom{X}} &  \multigate{2}{\phantom{X}} &  \qw & \qw & \bullet  \\
	&	  \bullet   \quad& \qw  &  \multigate{1}{\phantom{X}} & \ghost{\phantom{X}} &  \Nghost{\phantom{X}} & \qw & \qw &  \bullet\\
	&	& \Nmultigate{2}{\phantom{X}} &  \ghost{\phantom{X}} & \qw & \ghost{\phantom{X}} \\
	&	& \Nghost{\phantom{X}} & \qw & \qw & \qw & \qw & \qw &  \bullet \\
	&	& \Nghost{\phantom{X}} & \qw & \multigate{1}{\phantom{X}}  \\
	& 	  \bullet   \quad	&   \qw  & \qw &\ghost{\phantom{X}} & \qw &   \qw & \qw &  \bullet
}  \quad,
\end{align}
 which is formally a directed acyclic graph (the direction is from left to right) whose edges we think of as `wires' and whose nodes come in two types, which we think of respectively as `open ports' and `boxes'. The second piece of information is the \emph{filling}, which assigns to each wire in the stencil a system in $\Theory$ and to each box in the stencil a channel in $\Theory$ between the adjacent systems. In general, a \emph{circuit} is thus a pair $(\fF, G)$, where $G$ is a stencil described by a directed acyclic graph and where $\fF$ is a filling of the stencil $G$.
 
 To any circuit $(\fF, G)$ is an associated \emph{(input-output) behaviour}, denoted $\fF[G]$, namely the total channel from the system corresponding to the open input ports to the to the system corresponding to the open output ports. There is also an associated \emph{causal specification}, namely the map $\scrC_G$ which to every open output port in $G$ associates the set of input ports in $G$ which are ancestral to that output port. We think of these input ports as the \emph{causes} of the output port. (In the stencil \eqref{eq:stencil}, for example, the top three output ports each has the top two input ports as causes, the fourth output port has no causes, and the bottom output port has the bottom input port as cause.)
 
 The main idea behind causal channels is to consider the following equivalence relation among circuits:
 
 \begin{Definition}[Congruence of Circuits]
  Two circuits $(\fF, G)$ and $(\fF', G')$ are called \emph{congruent} if they have the same behaviour and the same causal specification, i.e.~if $\fF[G]=\fF'[G']$ and $\scrC_G = \scrC_{G'}$. 
 \end{Definition}

Clearly, there is a very natural way of representing the congruence class of a given circuit $(\fF, G)$, namely in terms of the pair $(T, \scrC) = (\fF[G], \scrC_G)$ consisting of the behaviour $T$ and the causal specification $\scrC$.

\begin{Definition}[Constructible Causal Channels]
A \emph{constructible causal channel} in $\Theory$ is a pair of the form $(T, \scrC)$ where $T=\fF[G]$ and $\scrC=\scrC_G$ for some circuit $(\fF,G)$ in $\Theory$.
\end{Definition}

\begin{Remark}[Constructibility] \label{rem:CausGen}The treatment in Ref.~\cite{Hou21} involves causal channels $(T, \scrC)$ of an abstract kind, which need not admit a representation in terms of circuits, but where $T$ merely complies to certain non-signalling conditions suggested by $\scrC$. The term `constructible' then precisely refers to those causal channels which arise from circuits. Here, we shall \myuline{only} consider constructible causal channels, but we often include the term for emphasis. 
\end{Remark}


A constructible causal channel is a mathematically much simpler object than a circuit. The reason why it is viable to make this simplification is that congruence of circuits respects the various operations one might wish to perform on the level of circuits: 

It is easy to see that if we compose two circuits in series or parallel, then the congruence class of the resulting composition can be determined from the congruence classes of the initial circuits. In other words, serial and parallel composition can be considered as operations on the level of constructible causal channels themselves. 

More remarkably, it is also true in most\footnote{Specifically, this is true under the so-called \emph{universality} axiom of Ref.~\cite{Hou21}, see Theorem 4.2.6 therein. The universality axiom holds in the theories $\QIT$ and $\CIT$.} theories that the congruence class of a circuit which arises from \myuline{contraction} of wires in an initial circuit can be derived from the congruence class of this initial circuit. (An output wire can be contracted with an input wire whenever the resulting graph remains acyclic.) In other words, contraction of wires can also be performed a the level of constructible causal channels. 

Though we shall not directly need the latter result here, the above facts are ultimately what justifies in manipulations the replacement of circuits by causal channels.

 \begin{Example}[Bell-Channels] \label{ex:Bell}
    Let $S= (\varrho, \Lambda_\sfA, \Lambda_\sfB)$ be a strategy for a Bell-scenario $\fB$. The strategy $S$ is really the description of a circuit, and the diagram
    \begin{align} \label{eq:Bellchan}
	\myQ{0.7}{0.7}{
		& \push{\what{X}_\sfA}  \qw 	& \qw  & \multigate{1}{\Lambda_\sfA} & \push{\what{Y}_\sfA}  \qw & \qw   \\
		& \Nmultigate{1}{\varrho}  & \push{\cH_\sfA} \qw & \ghost{\Lambda_\sfA} &  \\
		& \Nghost{\varrho} & \push{\cH_\sfB} \qw  & \multigate{1}{\Lambda_\sfB}  \\
	&   \push{\what{X}_\sfB} \qw  		& \qw  & \ghost{\Lambda_\sfB} & \push{\what{Y}_\sfB}  \qw & \qw \\
	}
	\end{align}
	gives rise to a constructible causal channel $(P, \scrC)$ which in addition to the behaviour $P$ of the strategy also holds a causal specification $\scrC$ which to the output port $\sfy_\sfP$ associates the set of input ports $\scrC(\sfy_\sfP) = \{\sfx_\sfP\}$, for $\sfP \in \{\sfA, \sfB\}$. Any other strategy $S'$ with behaviour $P$ gives rise to the same causal channel, but the circuit-representation is different.

A constructible causal channel $(P, \scrC)$ will be called a \emph{Bell-channel} if it admits a circuit-representation of the form \eqref{eq:Bellchan}. It can be shown \cite{Hou21} that if $(P, \scrC)$ is any constructible causal channel for which $P$ is embedded classical and $\scrC$ is given by $\scrC(\sfy_\sfP) = \{\sfx_\sfP\}$ for $\sfP \in \{\sfA, \sfB\}$, then $(P, \scrC)$ is in fact a Bell-channel.\end{Example}

\begin{Example}[Implementations] \label{ex:Imp}
Let $(\xi, \Phi_\sfA, \Phi_\sfB)$ be a components-wise dilation of a strategy $S=(\varrho, \Lambda_\sfA, \Lambda_\sfB)$ for a Bell-scenario $\fB$. Then, the appropriate way of thinking about the implementation
	\begin{align}
	    \myQ{0.7}{0.7}{
	& \push{\what{X}_\sfA}  \qw  & \qw & \multigate{1}{ \Phi_\sfA} & \push{\what{Y}_\sfA}  \qw & \qw  \\
	& \Nmultigate{2}{\xi}  & \push{\cH_\sfA} \qw & \ghost{\Phi_\sfA} & \push{\cE_\sfA} \ww & \ww \\
	& \Nghost{\xi}  & \push{\cE_0} \ww & \ww & \ww & \ww \\
	& \Nghost{\xi} & \push{\cH_\sfB} \qw & \multigate{1}{\Phi_\sfB} & \push{\cE_\sfB} \ww & \ww  \\
	&   \push{\what{X}_\sfB} \qw  & \qw&  \ghost{\Phi_\sfB} & \push{\what{Y}_\sfB}  \qw & \qw  \\
} 
	\end{align}
is as a constructible causal channel $(L, \scrE)$, where $L$ is the input-output behaviour, and where the causal specification $\scrE$ is given by $\scrE(\sfe_0)=\emptyset$ and $\scrE(\sfe_\sfP) = \scrE(\sfy_\sfP)= \{\sfx_\sfP\}$ for $\sfP \in \{\sfA, \sfB\}$.\end{Example}

\subsubsection{Causal Dilations and  the Derivability Pre-Order}
\label{subsubsec:Causdil}

The causal channel in \cref{ex:Imp} which corresponds to an implementation of a strategy $S$ has the following property: if we trace out the three ports $\sfe_\sfA$, $\sfe_\sfB$ and $\sfe_0$, then we obtain precisely the causal channel $(P, \scrC)$ from \cref{ex:Bell} describing the behaviour of the strategy. Specifically, the channel $L$ is a dilation of the channel $P$, and the causal specification $\scrE$ assigns to any output port $\sfy$, which is \myuline{not} traced out, the same set of causes as does $\scrC$. 

This is an instance of the following concept:

\begin{Definition}[Constructible Causal Dilations]Let $(T, \scrC)$ be a causal channel. A \emph{construcible causal dilation of $(T, \scrC)$} is a constructible causal channel $(L, \scrE)$ whose set of input ports coincides with that of $(T, \scrC)$, whose set of output ports contains that of $(T, \scrC)$, and such that $L$ is a dilation of $T$ and $\scrE(\sfy)=\scrC(\sfy)$ for any output port $\sfy$ in $(T, \scrC)$.

The set of output ports in $(L, \scrE)$ which are not output ports of $(T, \scrC)$ are called the \emph{environment of $(L, \scrE)$}. \end{Definition}

On the level of circuits, a constructible causal dilation of $(T, \scrC)= (\fF[G] ,\scrC_G)$ corresponds to a circuit which after having the output ports in its environment traced out becomes congruent to $(\fF, G)$. As such, the constructible causal dilations of $(T, \scrC)$ represent all the ways in which the circuit $(\fF, G)$ may possibly be realised in a larger environment. Side-information leaks to this environment according to a schedule prescribed by the causal dependence of outputs on inputs. 

There is a natural notion of relative strength between causal dilations:

\begin{Definition}[Derivability]
Let $(L, \scrE)$ and $(L', \scrE')$ be constructible causal dilations of $(T, \scrC)$. We say that \emph{$(L', \scrE')$ is derivable from $(L, \scrE)$}, or that \emph{$(L, \scrE)$ derives $(L',\scrE')$}, if we can obtain $(L', \scrE')$ by composing $(L, \scrE)$ with some constructible causal channel $(H, \scrD)$ which acts only in the environment. We write in this case $(L, \scrE) \cder (L', \scrE')$.\end{Definition}

   Again, derivability can be described on the level of circuits: $(L, \scrE)$  derives $(L', \scrE')$ if a circuit corresponding to $(L, \scrE)$ can be composed with some circuit, acting only in the environment, to yield a circuit corresponding to $(L', \scrE')$. As such, derivability of $(L', \scrE')$ from $(L, \scrE)$ means that all side-information which escapes to the environment according to $(L', \scrE')$ can be extracted --- even conforming to the correct schedule --- if we possess $(L, \scrE)$.

\begin{Example}
If $(T, \scrC)$ is a constructible causal channel, then $(T,\scrC)$ is a constructible causal dilation of itself. We call it the \emph{trivial} dilation of $(T, \scrC)$, and it is derivable from any other constructible causal dilation, by simply tracing out each output in the environment.  \end{Example}

\begin{Example}
Let $S$ and $S'$ be strategies for a Bell-scenario with behaviour $P$, and let $(\xi, \Phi_\sfA, \Phi_\sfB)$ and $(\xi', \Phi'_\sfA, \Phi'_\sfB)$ be respective component-wise dilations. We have seen in \cref{ex:Imp} that the corresponding implementations are causal dilations of the causal channel $(P, \scrC)$. The causal dilation corresponding to $(\xi', \Phi'_\sfA, \Phi'_\sfB)$ is derivable from the causal dilation corresponding to $(\xi, \Phi_\sfA, \Phi_\sfB)$ if there exist channels $\Gamma_0$, $\Gamma_\sfA$ and $\Gamma_\sfB$ such that 	
\begin{align}  \label{eq:CausSimImp}
\scalemyQ{1}{0.7}{0.5}{
	&   \push{\what{X}_\sfA} \qw  & \multigate{1}{\Phi_\sfA}  &   \qw &  \push{\what{Y}_\sfA}  \qw &  \qw  & \qw \\
	& \Nmultigate{4}{\xi}  & \ghost{\Phi_\sfA} &  \push{\cE_\sfA} \ww &  \Nmultigate{1}{\Gamma_\sfA}{\ww} & \push{\cE'_\sfA} \ww & \ww \\
	& \Nghost{\xi}   &  & \Nmultigate{2}{\Gamma_0} & \Nghost{\Gamma_\sfA}{\ww} &  \\  
	& \Nghost{\xi}  & \push{\cE_0} \ww &   \Nghost{\Gamma_0}{\ww}  & \ww & \push{\cE'_0} \ww & \ww \\
	& \Nghost{\xi}  &  & \Nghost{ \Gamma_0} &   \Nmultigate{1}{\Gamma_\sfB}{\ww}   \\
	& \Nghost{\xi}   & \multigate{1}{\Phi_\sfB} &  \push{\cE_\sfB} \ww & \Nghost{\Gamma_\sfB}{\ww} & \push{\cE'_\sfB}\ww & \ww \\ 
	&   \push{\what{X}_\sfB} \qw  & \ghost{\Phi_\sfB}   & \qw &  \push{\what{Y}_\sfB} \qw &   \qw & \qw} =\scalemyQ{1}{0.7}{0.5}{
	& \push{\what{X}_\sfA}\qw  & \multigate{1}{ \Phi'_\sfA} &  \push{\what{Y}_\sfB}  \qw & \qw  \\
	& \Nmultigate{2}{\xi'}  & \ghost{\Phi'_\sfA} &  \push{\cE'_\sfA} \ww & \ww\\
	& \Nghost{\xi'} & \ww & \push{\cE'_0} \ww & \ww\\
	& \Nghost{\xi'} & \multigate{1}{\Phi'_\sfB} & \push{\cE'_\sfB} \ww & \ww  \\
	&  \push{\what{X}_\sfA} \qw  & \ghost{\Phi'_\sfB} &  \push{\what{Y}_\sfB} \qw & \qw  \\
}  \quad. 
\end{align}
Indeed, the channels $\Gamma_0$, $\Gamma_\sfA$ and $\Gamma_\sfB$ make up a circuit in the environment which derives one implementation-circuit from the other, up to congruence of circuits (in particular preserving the causal dependence of outputs on inputs).
\end{Example}

\begin{Remark}[Modularity]\label{rem:Modular}
	The notion of derivability among causal dilations enjoys several `composability', or `modularity', features. For instance, if causal dilations $(L_1, \scrE_1)$ of $(T_1, \scrC_1)$ and $(L_2, \scrE_2)$ of $(T_2, \scrC_2)$ derive the causal dilations $(L'_1, \scrE'_1)$ and $(L'_2, \scrE'_2)$, respectively, then the parallel [serial] composition of $(L_1, \scrE_1)$ and $(L_2, \scrE_2)$ derives the parallel [serial] composition of $(L'_1, \scrE'_1)$ and $(L'_2, \scrE'_2)$. For precise and general versions of these statements, see \cite{Hou21} Theorems~4.3.24 and 4.3.25.\end{Remark}

Now, for any constructible causal channel $(T, \scrC)$, the derivability relation is a pre-order on the class of constructible causal dilations of $(T, \scrC)$. It is of interest to know this class and pre-order, since this provides an understanding of the various ways in which side-information may leak to the environment and how these relate to each other. In fact, it is already interesting to identify a \emph{dense} subclass of causal dilations, i.e.~a subclass such that any causal dilation may be derived from a member of this subclass. A very special case of this is when every dilation is derivable from a single one: 

\begin{Definition}[Complete Causal Dilations]
    A constructible causal dilation $(K, \scrF)$ of $(T, \scrC)$ is called \emph{(causally) complete} if it is a $\cder$-largest element, i.e.~if for any constructible causal dilation $(L, \scrE)$ of $(T, \scrC)$, we have $(K, \scrF) \cder (L, \scrE)$.
\end{Definition}

A complete causal dilation of $(T, \scrC)$ formalises the notion of a strongest possible information-leak to the environment. What we now wish to illustrate is that in the theory $\QIT$, in the special causal case of Bell-scenarios, causal completeness is essentially the occurrence of quantum self-testing according to causal simulation.

	\subsubsection{The Special Case of Bell-Scenarios}
	\label{subsubsec:SpecialBell}

	Let $\fB = (X_\sfA, X_\sfB, Y_\sfA, Y_\sfB)$ be a Bell-scenario. Recall from \cref{ex:Bell} that a Bell-channel is a constructible causal channel $(P, \scrC)$ which admits a circuit-representation of the form \eqref{eq:Bellchan}, or, equivalently, one whose behaviour $P$ is classical and whose causal specification is given by $\scrC(\sfy_\sfA)=\{\sfx_\sfA\}$ and $\scrC(\sfy_\sfB)=\{\sfx_\sfB\}$. We have the following result, which holds in any operational theory:

	\begin{Thm}[Density Theorem for Bell-Channels] \label{thm:BellDense} Let $(P, \scrC)$ be a Bell-channel. Then, any constructible causal dilation of $(P, \scrC)$ is derivable from a constructible causal dilation corresponding to a circuit of the form
\begin{align} \label{eq:belldil}
\myQ{0.7}{0.5}{
	&  \push{\what{X}_\sfA}  \qw& \multigate{1}{ \Phi_\sfA} &  \push{\what{Y}_\sfA} \qw & \qw  \\
	& \Nmultigate{2}{\xi}  & \ghost{\Phi_\sfA} &  \push{\cE_\sfA} \ww & \ww\\
	& \Nghost{\xi} & \ww & \push{\cE_0} \ww & \ww\\
	& \Nghost{\xi} & \multigate{1}{\Phi_\sfB} &  \push{\cE_\sfB}\ww & \ww  \\
	&   \push{\what{X}_\sfB} \qw & \ghost{\Phi_\sfB} & \push{\what{Y}_\sfB} \qw & \qw  \\
} 
\end{align}
for some state $\xi$ and some channels $\Phi_\sfA$ and $\Phi_\sfB$. 
\end{Thm}

\begin{Remark} \label{rem:DenseIso} In the theory $\QIT$, it follows from \cref{thm:StineComp} that we may strengthen the statement in \cref{thm:BellDense} by assuming the state $\xi$ to be pure and the channels $\Phi_\sfA$ and $\Phi_\sfB$ to be isometric. 
\end{Remark}

For a proof of \cref{thm:BellDense}, the reader is referred to Theorem 4.4.8 in Ref.~\cite{Hou21}. The proof is essentially graph-theoretic, and the significance of \cref{thm:BellDense} is that \myuline{any} circuit which implements $(P, \scrC)$ in a larger environment, however complicated, can be derived from a circuit congruent to one with the simple triple-structure displayed in \cref{eq:belldil}. \\

It would seem that \cref{thm:BellDense} \myuline{explains} the significance of triple-strategies relative to more general strategies, cf.~our earlier discussion in \cref{subsec:Aim}. This is not quite true, however, due to the following curious state of affairs: 

\emph{It may happen that the circuit \eqref{eq:belldil} is not the implementation of a strategy $S= (\varrho, \Lambda_\sfA, \Lambda_\sfB)$ with behaviour $P$}. The subtle reason for this circumstance is that that we defined a strategy $S$ to consist of a state $\varrho$ together with \myuline{measurement ensembles} $\Lambda_\sfA$ and $\Lambda_\sfB$. Whereas it is obviously true that $(\xi, \Phi_\sfA, \Phi_\sfB)$ in \cref{eq:belldil} is a component-wise dilations of some triple $(\varrho, \Lambda_\sfA,\Lambda_\sfB)$ for which
\begin{align} \label{eq:quasiBhv}
    \scalemyQ{1}{0.7}{0.7}{
		& \push{\what{X}_\sfA}  \qw 	& \qw  & \multigate{1}{\Lambda_\sfA} & \push{\what{Y}_\sfA}  \qw & \qw   \\
		& \Nmultigate{1}{\varrho}  & \push{\cH_\sfA} \qw & \ghost{\Lambda_\sfA} &  \\
		& \Nghost{\varrho} & \push{\cH_\sfB} \qw  & \multigate{1}{\Lambda_\sfB}  \\
	&   \push{\what{X}_\sfB} \qw  		& \qw  & \ghost{\Lambda_\sfB} & \push{\what{Y}_\sfB}  \qw & \qw \\
	} = \myQ{0.7}{0.7}{& \push{\what{X}_\sfA} \qw & \multigate{1}{P} & \push{\what{Y}_\sfA} \qw & \qw \\ & \push{\what{X}_\sfB} \qw & \ghost{P} & \push{\what{Y}_\sfB} \qw & \qw } \quad , 
	\end{align}
 the fact that $P$ is classical does not guarantee that the channels $\Lambda_\sfA$ and $\Lambda_\sfB$ are measurement ensembles, i.e.~have classical  outputs and classical inputs on $\what{X}_\sfA$ and $\what{X}_\sfB$. \\
 
 Let us call any triple $(\varrho, \Lambda_\sfA, \Lambda_\sfB)$ a \emph{quasi-strategy with behaviour $P$} if it satisfies  \cref{eq:quasiBhv}. \emph{Implementations} of a quasi-strategy $(\varrho, \Lambda_\sfA, \Lambda_\sfB)$ are defined in the obvious way, namely as causal dilations of the form \eqref{eq:belldil}, where $(\xi, \Phi_\sfA, \Phi_\sfB)$ is a component-wise dilation of $(\varrho, \Lambda_\sfA, \Lambda_\sfB)$. \emph{Causal simulation} can be also be defined in the obvious way between quasi-strategies.
 
 Now, \cref{thm:BellDense} expresses that every constructible causal dilation of a Bell-channel $(P, \scrC)$ is derivable from the implementation of some quasi-strategy with behaviour $P$. The appearance of quasi-strategies constitutes an annoyance in relating quantum self-testing to general causal dilations, and we shall return to this problem very shortly. First, however, we state the following two results, which again hold in any operational theory: 
 
  \begin{Thm}[Derivability within the Dense Class for Bell-Channels]\label{thm:BellDenseII} 

Let $(P, \scrC)$ be a Bell-channel, and let $(L, \scrE)$ and $(L', \scrE')$ be constructible causal dilations which are implementations of quasi-strategies, corresponding to the component-wise dilations $(\xi, \Phi_\sfA, \Phi_\sfB)$ and $(\xi', \Phi'_\sfA, \Phi'_\sfB)$, respectively. Then, $(L, \scrE) \cder (L', \scrE')$ if and only if there exist channels $\Gamma_0$, $\Gamma_\sfA$ and $\Gamma_\sfB$ such that 	
\begin{align}  \label{eq:ThinDense}
\scalemyQ{1}{0.7}{0.5}{
	&   \push{\what{X}_\sfA} \qw  & \multigate{1}{\Phi_\sfA}  &   \qw &  \push{\what{Y}_\sfA}  \qw &  \qw  & \qw \\
	& \Nmultigate{4}{\xi}  & \ghost{\Phi_\sfA} &  \push{\cE_\sfA} \ww &  \Nmultigate{1}{\Gamma_\sfA}{\ww} & \push{\cE'_\sfA} \ww & \ww \\
	& \Nghost{\xi}   &  & \Nmultigate{2}{\Gamma_0} & \Nghost{\Gamma_\sfA}{\ww} &  \\  
	& \Nghost{\xi}  & \push{\cE_0} \ww &   \Nghost{\Gamma_0}{\ww}  & \ww & \push{\cE'_0} \ww & \ww \\
	& \Nghost{\xi}  &  & \Nghost{ \Gamma_0} &   \Nmultigate{1}{\Gamma_\sfB}{\ww}   \\
	& \Nghost{\xi}   & \multigate{1}{\Phi_\sfB} &  \push{\cE_\sfB} \ww & \Nghost{\Gamma_\sfB}{\ww} & \push{\cE'_\sfB}\ww & \ww \\ 
	&   \push{\what{X}_\sfB} \qw  & \ghost{\Phi_\sfB}   & \qw &  \push{\what{Y}_\sfB} \qw &   \qw & \qw} =\scalemyQ{1}{0.7}{0.5}{
	& \push{\what{X}_\sfA}\qw  & \multigate{1}{ \Phi'_\sfA} &  \push{\what{Y}_\sfB}  \qw & \qw  \\
	& \Nmultigate{2}{\xi'}  & \ghost{\Phi'_\sfA} &  \push{\cE'_\sfA} \ww & \ww\\
	& \Nghost{\xi'} & \ww & \push{\cE'_0} \ww & \ww\\
	& \Nghost{\xi'} & \multigate{1}{\Phi'_\sfB} & \push{\cE'_\sfB} \ww & \ww  \\
	&  \push{\what{X}_\sfA} \qw  & \ghost{\Phi'_\sfB} &  \push{\what{Y}_\sfB} \qw & \qw  \\
}  \quad. 
\end{align}
\end{Thm}

\begin{Cor}[Characterisation of Completeness for Bell-Channels]  \label{cor:CompBell}
The Bell-channel $(P, \scrC)$ has a complete causal dilation if and only if there exists a quasi-strategy $\tilde{S}$ with behaviour $P$ which causally simulates every quasi-strategy $S$ with behaviour $P$.\end{Cor}

For a proof of \cref{thm:BellDenseII}, see Theorem 4.8.10 in Ref.~\cite{Hou21}. It consists in demonstrating that the most general causal channel which may derive $(L',\scrE')$ from $(L, \scrE)$ is of the form represented in \cref{eq:ThinDense} by the channels $\Gamma_0$, $\Gamma_\sfA$ and $\Gamma_\sfB$, and the argument is again graph-theoretic. \cref{cor:CompBell} follows immediately from \cref{thm:BellDenseII} and \cref{thm:BellDense}. \\
 
 Now, were it not for the appearance of quasi-strategies rather than ordinary strategies, \cref{cor:CompBell} would provide the final link between quantum self-testing and the general theory of derivability among causal dilations: Quantum self-testing (according to causal simulation) simply signifies the existence of a complete causal dilation of the Bell-channel. 
 
 As the following example shows, however, some quasi-strategies cannot be treated as ordinary strategies:

\begin{Example} (A Genuine Quasi-Strategy.) \label{ex:StrangeRel}
Consider the Bell-scenario $\fB$ with $X_\sfA = X_\sfB = \{0,1\}$ and $Y_\sfA = Y_\sfB=\{+1,-1\}$. Let $\tilde{S} = (\tilde{\psi}, \tilde{\Lambda}_\sfB, \tilde{\Lambda}_\sfB)$ be an ordinary strategy for $\fB$ whose state is a pure state on $\C^2 \otimes \C^2$ and whose measurement ensembles are projective, with each projection $\Pi^{x_\sfP}_\sfP(y_\sfP)$ having rank one. (For example, $\tilde{S}$ could be the ideal strategy for the CHSH-game \cite{CHSH69,SB19}.) From the ordinary strategy $\tilde{S}$, we may construct a quasi-strategy $\tilde{S}'$ which is not an ordinary strategy:

		First, observe that every projective measurement with rank-one projections may be realised as a unitary conjugation followed by a measurement in the computational basis. As such, the measurement ensemble $\tilde{\Lambda}_\sfA$ may be rewritten as $	\scalemyQ{.8}{0.7}{0.5}{& \push{\what{X}_{\sfA}} \qw & \multigate{1}{\tilde{\Lambda}_{\sfA}} & \push{\what{Y}_{\sfA}} \qw & \qw \\ &\push{\C^2} \qw & \ghost{\tilde{\Lambda}_{\sfA}}} = \scalemyQ{.8}{0.7}{0.5}{& \push{\what{X}_{\sfA}} \qw & \multigate{1}{\tilde{\Omega}_{\sfA}}& \push{\C^2} \qw & \gate{\Delta} & \push{\what{Y}_{\sfA}} \qw & \qw \\ &\push{\C^2} \qw & \ghost{\tilde{\Omega}_{\sfA} }}$ where $\tilde{\Omega}_\sfA$ is an $X_\sfA$-indexed ensemble of unitary conjugations $\End{\C^2} \to \End{\C^2}$, and where $\Delta$ is the measurement in the computational basis of $\C^2$ (equivalently, the dephasing channel on $\what{Y}_\sfA$ under some natural identification of $\what{Y}_{\sfA}$ with $\C^2$). 
		
		Next, the channel $\Delta$ equals the uniform mixture $\frac{1}{2} \id_{\C^2} + \frac{1}{2} Z$, where $Z : \End{\C^2} \to \End{\C^2}$ denotes the unitary conjugation by the Pauli matrix $\sigma_z = \left( \begin{array}{cc} 1 & 0 \\ 0 & -1 \end{array} \right)$. As such, we may substitute $\scalemyQ{.8}{0.7}{0.5}{& \push{\C^2} \qw  & \gate{\Delta} & \push{\what{Y}_{\sfA}} \qw & \qw}$ by $\scalemyQ{.8}{0.7}{0.5}{& \push{\C^2} \qw & \qw& \multigate{1}{\Omega'_\sfA} & \push{\what{Y}_{\sfA}} \qw & \qw \\ &\Ngate{\tau_{\sfA}} & \push{\C^2} \qw & \ghost{\Omega'_\sfA}}$, where $\tau_\sfA$ is the state $\frac{1}{2} \ketbra{0} + \frac{1}{2} \ketbra{1}$ and where $\Omega'_\sfA$ is the ensemble of unitary conjugations which when reading off $\ketbra{0}$ or $\ketbra{1}$ applies $\id_{\C^2}$ or $Z$, respectively.

The measurement ensemble $\tilde{\Lambda}_\sfB$ may be rewritten in a similar way, with ensembles $\tilde{\Omega}_\sfB$ and $\Omega'_\sfB$ of unitary conjugations, and the state $\tau_\sfB$ again given by $\frac{1}{2} \ketbra{0} + \frac{1}{2} \ketbra{1}$. Altogether, the behaviour $P$ of $\tilde{S}$ can therefore be written as 
	\begin{align} \label{eq:strange}
	\myQ{0.7}{0.5}{
		& \qw &  \push{\what{X}_\sfA}  \qw& \qw & \multigate{1}{ \tilde{\Omega}_\sfA} &  \push{\C^2} \qw & \multigate{2}{\Omega'_\sfA} & \push{\what{Y}_\sfA} \qw & \qw  \\
		& \Nmultigate{3}{\tilde{\psi}} & \qw  &  \push{\C^2} \qw & \ghost{\tilde{\Omega}_\sfA} & & \Nghost{\Omega'_\sfA}  \\
		& \Nghost{\tilde{\psi}} & \Ngate{\tau_\sfA}& \push{\C^{2}}\qw& \qw& \qw & \ghost{\Omega'_\sfA}\\
		& \Nghost{\tilde{\psi}} & \Ngate{\tau_\sfB}& \push{\C^{2}}\qw& \qw& \qw & \multigate{2}{\Omega'_\sfB} \\
		& \Nghost{\tilde{\psi}} &\qw &  \push{\C^2} \qw & \multigate{1}{\tilde{\Omega}_\sfB} &  \qw & \ghost{\Omega'_\sfB} \\
		&  \qw &  \push{\what{X}_\sfB} \qw &\qw & \ghost{\tilde{\Omega}_\sfB} &  \qw & \ghost{\Omega'_\sfB} &  \push{\what{Y}_\sfB} \qw &  \qw  \\
	}  \quad . \end{align}
	
	Now, we may regroup components so as to form a quasi-strategy $S'= (\tilde{\varrho}' , \tilde{\Lambda}'_\sfA, \tilde{\Lambda}'_\sfB)$, where $\tilde{\varrho}'$ is the state $\tau_\sfA \otimes \tilde{\psi} \otimes \tau_\sfB$ on $(\C^2 \otimes \C^2) \otimes (\C^2 \otimes \C^2)$, and where $\tilde{\Lambda}'_{\sfP} : \End{\what{X}_\sfP \otimes \C^2 \otimes \C^2} \to \End{\what{Y}_\sfP}$ is the channel given by the composition of $\tilde{\Omega}_\sfP$ and $\Omega'_\sfP$. The quasi-strategy $\tilde{S}'$ is not an ordinary strategy, since the channels $\tilde{\Lambda}'_{\sfP}$ are not measurement ensembles. In fact, for $x_\sfP \in  \{0,1\}$ and $k \in \{0,1\}$, the channel $F \mapsto \tilde{\Lambda}'_\sfP( \ketbra{x_\sfP} \otimes F \otimes \ketbra{k})$ is one of four possible unitary conjugations from $\C^2$ to $\what{Y}_{\sfP}$, so a suitable choice of input results in a state on $\what{Y}_\sfP$ which is not diagonal in the computational basis. \end{Example}

While mathematically sensible, the quasi-strategy in \cref{ex:StrangeRel} is physically somewhat obscure. Rather than performing a measurement to obtain an output, each party $\sfP \in \{\sfA, \sfB\}$ performs a random unitary rotation on the $\sfP$-part of the state $\tilde{\psi}$ (the randomness residing in the state $\tau_\sfP$) and returns the resulting state as output --- the total input-output behaviour then happens to be classical `on average'. The quasi-strategy is obscure for the same reason that a measurement in the computational basis of $\C^2$ (i.e.~the channel $\Delta$) cannot seriously be employed by flipping a coin and applying one of the unitary conjugations $\id_{\C^2}$ and $Z$, even though it is mathematically true that $\Delta= \frac{1}{2} \id_{\C^2} + \frac{1}{2} Z$. 

It could seem that the problematic aspect of considering this procedure as a measurement ultimately stems from the fact that the outcome of the coin flip may stay in memory (this is addressed more formally in \cite{Hou21} Example 4.3.9). We do not wish to discuss philosophical aspects of quasi-strategies here, but we need to somehow exclude them from consideration in order to establish a formal link to ordinary quantum self-testing. This is not only because quasi-strategies do not occur as valid strategies in ordinary self-testing, but also because including them would in some cases actually constitute an obstruction to self-testing in terms of complete causal dilations: whereas for example the canonical optimal strategy $\tilde{S}$ for the CHSH-game \cite{CHSH69,SB19} is self-tested a.t.r., and therefore causally simulates any ordinary strategy with the correct behaviour, it does \myuline{not} causally simulate the implementation corresponding to the quasi-strategy $\tilde{S}'$ in \cref{ex:StrangeRel}. Indeed, $\tilde{S}'$ has an implementation in which the states $\tau_\sfA$ and $\tau_\sfB$ are dilated by means of copying (this precisely corresponding to `keeping the coin flip in memory', as mentioned above). Conditioned on the value of these copies, the output on $\what{Y}_\sfA \otimes \what{Y}_\sfB$ is one of four pure entangled states, and since their mixture is a classical state the copies must be correlated non-trivially with these states (i.e.~the states must be different). On the other hand, the state $\tilde{\psi}$ in the ordinary canonical strategy $\tilde{\psi}$ is pure, and therefore $\tilde{S}$ cannot causally simulate any strategy with this feature (see also \cref{subsec:SecExt}). 

In Ref.~\cite{Hou21}, the exclusion of non-ordinary quasi-strategies is achieved by introducing a notion of \emph{classically bound} causal dilations, in an attempt to formalise the occurrence of a measurement which cannot be causally dilated e.g.~by keeping a record of a unitary conjugation (see {Hou21} Definition 5.2.10 and the surrounding). We shall not explain this notion further here, but merely remark that the classically bound causal dilations are precisely those which can be derived from implementations of \myuline{ordinary} strategies (\cite{Hou21} Proposition 5.2.12). It thus achieves the desired restriction of \cref{thm:BellDense}, and allows for the following result which concludes our presentation of how to recover quantum self-testing from the larger perspective of causal channels and causal dilations:

\begin{Thm}[Self-Testing a.t.c.s. as an Instance of Causal Completeness] \label{thm:BridgeSC}
Let $(P, \scrC)$ be a Bell-channel. Then the following are equivalent:
\begin{enumerate}
\item There is a classically bound constructible causal dilation of $(P, \scrC)$ from which every classically bound constructible dilation can be derived. 
\item $P$ self-tests some strategy according to causal simulation.
\end{enumerate}
In fact, if $\tilde{S}$ is a strategy self-tested by $P$ a.t.c.s., then its associated Stinespring implementation defines a classically bound constructible causal dilation from which all others can be derived. 
\end{Thm}

{\begin{Remark}
When $\tilde{S}$ is pure-state, causal simulation may in \cref{thm:BridgeSC} be replaced by local assisted simulation (see \cref{rem:CausPure}). In the language of causal dilations, the pure-state assumption on $\tilde{S}$ corresponds to the requirement that side-information in the environment which has no causes --- i.e.~which may be leaked in advance of the inputs --- must factor from the remaining information. (See \cref{subsec:SecExt} and Chapter 5 in Ref.~\cite{Hou21} for more formal discussions of this.)
\end{Remark}
}

\section{Utility and Applications of the Operational Formulation}
\label{sec:Utility}

In this final section of the paper, we explore how our operational notions of simulation from \cref{subsec:Sim}, \cref{subsec:AssSim} and \cref{subsubsec:Caussim} can be used to prove features related to quantum self-testing. Some of these features are already known, but their proofs have previously relied on operator-algebra. The proofs presented in this section are all operational. Moreover, they clarify that some of the results hold even when replacing local simulation (reducibility) by the weaker notions of local assisted simulation or causal simulation.

In \cref{subsec:StateExt} we show how \myuline{local} simulation of $\tilde{S}$ by $S$ implies that the state $\tilde{\varrho}$ is locally extractable from the state $\varrho$ (\cref{prop:StateExtract}). This is well-understood in the operator-algebraic framework \cite{SB19} (see \cref{rem:ConvStateExt} below), albeit under the assumption that $\tilde{\varrho}$ is pure. \cref{prop:StateExtract} is phrased for general states and implies rigorously that only pure-state strategies can be self-tested according to local simulation (\cref{cor:SelfPure}).  

In \cref{subsec:MeasExt}, we show how \myuline{local assisted} simulation of $S$ by $\tilde{S}$ implies that the measurement ensembles of $S$ can be extracted from those of $\tilde{S}$ (\cref{prop:MeasExtract}). This result is related to certain formulations of self-testing of measurements alone \cite{Kan17,SB19}, but we believe that it has not been stated in this form before.

In \cref{subsec:SecExt}, we demonstrate that \myuline{causal} simulation has implications for the structure of convex decompositions of behaviours (\cref{lem:visible} and \cref{prop:SelfSec}). This in particular recovers from a general perspective the result of Ref.~\cite{Goh18} that a behaviour which self-tests some strategy a.t.r.~must be extremal in the convex set of realisable behaviours. In fact, our version of this result (\cref{cor:extremal}) holds under the weaker assumption of self-testing a.t.l.a.s.

Finally, in \cref{subsec:Exhausted}, we report a feature of self-testing which to our best knowledge has not been observed before. We will focus our view on the \emph{environment} of strategies implementing a certain behaviour. In the case of pure projective rank one strategies, we find that the respective environments of the players will hold exactly a copy of their input - no more and no less. We do not currently know of any concrete applications of this feature, but rather consider it as illustrative in its own right.

\subsection{Local Simulation and State Extractability}
\label{subsec:StateExt}

It is well-known that the operator-algebraic reducibility relation $S \geq_{red.} \tilde{S}$ (\cref{def:Redattl}) implies that the state $\tilde{\psi}$ of $\tilde{S}$ is locally extractible from the state $\varrho$ of $S$. Indeed, if in the reducibility conditions $[	W \Pi^{x}(y)   \otimes \bone_\cP ] \ket{\psi} =  \tilde{\Pi}^{x}(y) \tilde{\ket{\psi}} \otimes \ket{\psi^\res}$ we sum over $y$, then we obtain the relation $[	W  \otimes \bone_\cP ] \ket{\psi} =   \tilde{\ket{\psi}} \otimes \ket{\psi^\res}$, or, unfolding $W$ as $W_\sfA \otimes W_\sfB$, 
\begin{align} \label{eq:exteq}
[	W _\sfA  \otimes W_\sfB   \otimes \bone_\cP ] \ket{\psi} = \tilde{\ket{\psi}} \otimes \ket{\psi^\res} .
\end{align}

Below, we prove that a similar state extractibility result follows from the relation $S \geq_{l.s.} \tilde{S}$. Our proof is entirely compositional (it is in terms of diagrams), and it works even when the involved strategies are not pure-state or projective:

\begin{Prop}[State Extractibility] \label{prop:StateExtract} Let $S=(\varrho, \Lambda_\sfA, \Lambda_\sfB)$ and $\tilde{S} = (\tilde{\varrho}, \tilde{\Lambda}_\sfA, \tilde{\Lambda}_\sfB)$ be strategies. Let $\psi$ and $\tilde{\psi}$ be purifications of $\varrho$ and $\tilde{\varrho}$, respectively. If $S \geq_{l.s.} \tilde{S}$, then there exist channels $\Xi_0$, $\Xi_\sfA$ and $\Xi_\sfB$ such that 
\begin{align} \label{eq:StateExtract}
\myQ{0.7}{0.7}{& \Nmultigate{2}{\psi} & \push{\cH_\sfA} \qw & \gate{\Xi_\sfA} & \push{\tilde{\cH}_\sfA} \qw & \qw \\
& \Nghost{\psi} & \push{\cP} \ww & \Ngate{\Xi_0}{\ww} & \push{\tilde{\cP}} \ww & \ww\\& \Nghost{\psi} & \push{\cH_\sfB} \qw & \gate{\Xi_\sfB} & \push{\tilde{\cH}_\sfB} \qw & \qw} 
\quad  = \quad
\myQ{0.7}{0.7}{& \Nmultigate{2}{\tilde{\psi}} & \push{\tilde{\cH}_\sfA} \qw & \qw \\ &\Nghost{\tilde{\psi}} & \push{\tilde{\cP}} \ww & \ww \\ 
& \Nghost{\tilde{\psi}} & \push{\tilde{\cH}_\sfB} \qw & \qw}.
\end{align}
\end{Prop}

We make a few remarks before presenting the proof.

\begin{Remark}[Relation to Conventional State Extractibility] \label{rem:ConvStateExt}If the state $\tilde{\varrho}$ is already pure we may choose $\tilde{\psi}= \tilde{\varrho}$, and the channel $\Xi_0$ is then the trace. Employing our usual trick of replacing $\Xi_{\sfA}$ and $\Xi_\sfB$ by Stinespring dilations $\breve{\Xi}_{\sfA}$ and $\breve{\Xi}_{\sfA}$ and replacing the trace by an identity, we thus conclude from  \cref{eq:StateExtract} that
\begin{align} \label{eq:StateExtractPure}
\myQ{0.7}{0.7}{& \Nmultigate{4}{\psi} & \push{\cH_\sfA} \qw & \multigate{1}{\breve{\Xi}_\sfA} & \push{\tilde{\cH}_\sfA} \qw & \qw \\ & \Nghost{\psi} & & \Nghost{\breve{\Xi}_\sfA} & \push{\cH^\res_\sfA} \ww & \ww \\ & \Nghost{\psi} & \ww & \push{\cP} \ww & \ww & \ww \\& \Nghost{\psi} & & \Nmultigate{1}{\breve{\Xi}_\sfB} & \push{\cH^\res_\sfB} \ww & \ww\\ & \Nghost{\psi} & \push{\cH_\sfB} \qw & \ghost{\breve{\Xi}_\sfB} & \push{\tilde{\cH}_\sfB} \qw & \qw} 
\quad = \quad  \myQ{0.7}{0.7}{& \Nmultigate{4}{\tilde{\psi}} & \push{\tilde{\cH}_\sfA} \qw & \qw \\& \Nghost{\tilde{\psi}} & \Nmultigate{2}{\psi^\res} & \push{\cH^\res_\sfA} \ww & \ww  \\ & \Nghost{\tilde{\psi}}& \Nghost{\psi^\res} & \push{\cP} \ww & \ww \\ & \Nghost{\tilde{\psi}} & \Nghost{\psi^\res} & \push{\cH^\res_\sfB} \ww & \ww \\ & \Nghost{\tilde{\psi}} & \push{\tilde{\cH}_\sfB} \qw & \qw}  
\end{align}
for some pure state $\psi^\res$, and this is precisely \cref{eq:exteq}, with $W_\sfA$ and $W_\sfB$ representing $\breve{\Xi}_\sfA$ and $\breve{\Xi}_\sfB$.\end{Remark}

\begin{Remark}[Generalisation to the Approximate Case] It will be clear from the proof of \cref{prop:StateExtract} that the result translates without loss to the case of approximate simulation (cf.~\cref{rem:Approx}) provided the metric is monotone, i.e.~distance does not increase under post-processing.
\end{Remark}

\begin{proof}[Proof (of \cref{prop:StateExtract})]
Choose Stinespring dilations $\Sigma_{\sfP}$ of $\Lambda_{\sfP}$ and $\tilde{\Sigma}_{\sfP}$ of $\tilde{\Lambda}_{\sfP}$, for $\sfP \in \{\sfA, \sfB\}$. By \cref{prop:Rechar}, there exist channels $\Gamma_0$, $\Gamma_\sfA$ and $\Gamma_\sfB$ such that 
     \begin{align}   \label{eq:locstart}
\myQ{0.7}{0.7}{
	& \push{\what{X}_\sfA}  \qw  & \multigate{1}{ \Sigma_\sfA} & \qw &  \push{\what{Y}_\sfA}  \qw & \qw  \\
	& \Nmultigate{2}{\psi}  & \ghost{\Sigma_\sfA} & \push{\cE_\sfA} \ww &  \Ngate{\Gamma_\sfA}{\ww} &\push{\tilde{\cE}_A} \ww& \ww\\
	& \Nghost{\psi} & \ww & \push{\cP} \ww &\Ngate{\Gamma_0}{\ww}& \push{\tilde{\cP}} \ww& \ww \\
	& \Nghost{\psi} & \multigate{1}{\Sigma_\sfB} & \push{\cE_\sfB} \ww &  \Ngate{\Gamma_\sfB}{\ww} &\push{\tilde{\cE}_B} \ww& \ww \\
	&   \push{\what{X}_\sfB} \qw  & \ghost{\Sigma_\sfB} & \qw &  \push{\what{Y}_\sfB}  \qw & \qw  \\
} 
\quad = \quad 
\myQ{0.7}{0.7}{
	& \push{\what{X}_\sfA}  \qw  & \multigate{1}{ \tilde{\Sigma}_\sfA} & \qw &  \push{\what{Y}_\sfA}  \qw & \qw  \\
	& \Nmultigate{2}{\tilde{\psi}}  & \ghost{\tilde{\Sigma}_\sfA} & \push{\tilde{\cE}_A} \ww& \ww\\
	& \Nghost{\tilde{\psi}} & \ww & \push{\tilde{\cP}} \ww&\ww  \\
	& \Nghost{\tilde{\psi}} & \multigate{1}{\tilde{\Sigma}_\sfB} & \push{\tilde{\cE}_\sfB} \ww&\ww  \\
	&   \push{\what{X}_\sfB} \qw  & \ghost{\tilde{\Sigma}_\sfB} & \qw &  \push{\what{Y}_\sfB}  \qw & \qw  \\
} .
\end{align}
The channels $\tilde{\Sigma}_\sfA$ and $\tilde{\Sigma}_\sfB$ are isometric and thus reversible by \cref{lem:IsoRev}. Applying left-inverses on both sides of \cref{eq:locstart} yields
	
	\begin{align}
	\myQ{0.7}{0.7}{
		& \push{\what{X}_\sfA}  \qw  & \multigate{1}{ \Xi'_\sfA} & \push{\what{X}_\sfA}  \qw & \qw & \qw  \\
		& \Nmultigate{2}{\psi}  & \ghost{\Xi'_\sfA} &  \push{\tilde{\cH}_\sfA} \qw & \qw & \qw \\	& \Nghost{\psi} & \push{\cP} \ww &\Ngate{\Gamma_0}{\ww} & \push{\tilde{\cP}}\ww & \ww \\
		& \Nghost{\psi} & \multigate{1}{\Xi'_\sfB} & \push{\tilde{\cH}_\sfB} \qw & \qw & \qw \\
		&   \push{\what{X}_\sfB} \qw  & \ghost{\Xi'_\sfB} & \push{\what{X}_\sfB}  \qw & \qw  & \qw 
	} 
\quad = \quad 
	\myQ{0.7}{0.7}{& \qw & \push{\what{X}_\sfA} \qw & \qw & \qw \\& \Nmultigate{2}{\tilde{\psi}} & \push{\tilde{\cH}_\sfA} \qw & \qw \\& \Nghost{\tilde{\psi}}& \push{\tilde{\cP}} \ww & \ww \\& \Nghost{\tilde{\psi}} & \push{\tilde{\cH}_\sfB} \qw & \qw\\& \qw & \push{\what{X}_\sfB} \qw & \qw & \qw }  ,
	\end{align}
	for some channels $\Xi'_\sfA$ and $\Xi'_\sfB$. If we now, for $\sfP \in \{\sfA, \sfB\}$, trace out the output systems $\what{X}_{\sfP}$ and insert arbitrary states $\sigma_{\sfP}$ into the input systems $\what{X}_{\sfP}$, then the desired identity \eqref{eq:StateExtract} follows by defining $\Xi_{\sfP} := (\tr_{\what{X}_{\sfP}} \otimes \id_{\tilde{\cH}_{\sfP}}) \circ \Xi'_{\sfP} \circ (\sigma_{\sfP} \otimes \id_{\cH_{\sfP}})$ for ${\sfP} \in \{\sfA, \sfB\}$ and $\Xi_0 := \Gamma_0$.  
\end{proof}

We immediately conclude the following:

\begin{Cor} \label{cor:SelfStateExtract}
	Suppose that the strategy $\tilde{S} = (\tilde{\varrho}, \tilde{\Lambda}_\sfA, \tilde{\Lambda}_\sfB)$ is self-tested according to local simulation by its behaviour $P$. Let $\tilde{\psi}$ be a purification of $\tilde{\varrho}$. For any strategy $S= (\varrho, \Lambda_\sfA, \Lambda_\sfB)$ with behaviour $P$, and for any purification $\psi$ of $\varrho$, there exist channels $\Xi_0$, $\Xi_\sfA$ and $\Xi_\sfB$ such that \cref{eq:StateExtract} holds.
\end{Cor}

Since \cref{cor:SelfStateExtract} is a statement about canonical strategies with arbitrary state, it can be used to rigorously show that certain states simply cannot appear in self-testable strategies:

\begin{Cor}[All a.t.l.s.~Self-Testable Strategies are Pure-State] \label{cor:SelfPure}
Suppose that $\tilde{S} = (\tilde{\varrho}, \tilde{\Lambda}_\sfA, \tilde{\Lambda}_\sfB)$ is a strategy which is self-tested a.t.l.s.~by its behaviour. Then, the state $\tilde{\varrho}$ must be pure.
\end{Cor}

\begin{proof}
There is \myuline{some} strategy $S$ with behaviour $P$ whose state, say $\psi$, is pure. Indeed, such a strategy can be obtained if in any given strategy $S_0$ we purify the state and hand the purification to one of the parties, say $\sfA$. By \cref{cor:SelfStateExtract}, a purification $\tilde{\psi}$ of $\tilde{\varrho}$ now satisfies 
\begin{align}
\myQ{0.7}{0.7}{& \Nmultigate{2}{\tilde{\psi}} & \push{\tilde{\cH}_\sfA} \qw & \qw \\ &\Nghost{\tilde{\psi}} & \push{\tilde{\cP}} \ww & \ww \\ 
& \Nghost{\tilde{\psi}} & \push{\tilde{\cH}_\sfB} \qw & \qw} \quad  =  \quad \myQ{0.7}{0.7}{& \Nmultigate{2}{\psi} & \push{\cH_\sfA} \qw & \gate{\Xi_\sfA} & \push{\tilde{\cH}_\sfA} \qw & \qw \\
& \Nghost{\psi} &  & \Ngate{\xi_0} & \push{\tilde{\cP}} \ww & \ww\\& \Nghost{\psi} & \push{\cH_\sfB} \qw & \gate{\Xi_\sfB} & \push{\tilde{\cH}_\sfB} \qw & \qw} 
\end{align}
for some channels $\Xi_\sfA$, $\Xi_\sfB$ and some state $\xi_0$ (since the purifying system for $\psi$ may be taken trivial). In other words, $\tilde{\psi}$ factors across its purifying system and it follows that $\tilde{\varrho}$ must already be pure. 
\end{proof}

\subsection{Local Assisted Simulation and Measurement Extractibility}

\label{subsec:MeasExt}

Next, we prove that local assisted simulation implies an extractibility result for the measurement ensembles under a suitable rank-assumption. As mentioned, this result connects to alternative formulations of self-testing \cite{Kan17,SB19}, but as far as we know the general statement below (\cref{prop:MeasExtract}) has not been presented before. Like for state extractibility, our proof is purely compositional in terms of diagrams. 

\begin{Prop}[Measurement Extractibility] \label{prop:MeasExtract}
Let $S=(\varrho, \Lambda_\sfA, \Lambda_\sfB)$ and $\tilde{S} = (\tilde{\varrho}, \tilde{\Lambda}_\sfA, \tilde{\Lambda}_\sfB)$ be strategies such that $S \leq_{l.a.s.} \tilde{S}$. Let $\sfP \in \{\sfA, \sfB\}$, and let $\Sigma_{\sfP}$ and $\tilde{\Sigma}_{\sfP}$ be Stinespring dilations of $\Lambda_{\sfP}$ and $\tilde{\Lambda}_{\sfP}$, respectively. If the $\sfP$-marginal $\tilde{\varrho}_{\sfP}$ of $\tilde{\varrho}$ has full rank, then there exist channels $\tilde{\Xi}_{\sfP}$ and $\tilde{\Psi}_{\sfP}$ such that
	\begin{align} \label{eq:MeasExt}
	\myQ{0.7}{0.7}{
		& \qw	& \push{\what{X}_{\sfP}} \qw    &  \qw &  \multigate{1}{\Sigma_{\sfP}}  &  \qw &  \push{\what{Y}_{\sfP}} \qw &  \qw & \qw & \qw \\
		& \push{\tilde{\cH}_{\sfP}}\qw &   \gate{\tilde{\Xi}_{\sfP}} & \push{\cH_{\sfP}} \qw &  \ghost{\Sigma_{\sfP}}  & \push{\cE_\sfP} \ww &  \Ngate{\tilde{\Psi}_{\sfP}}{\ww} & \push{\tilde{\cE}_\sfP} \ww & \ww 
	} 
	\quad =\quad 
	\myQ{0.7}{0.7}{
		&  \push{\what{X}_{\sfP}}   \qw  & \multigate{1}{\tilde{\Sigma}_{\sfP}}  &    \push{\what{Y}_{\sfP}} \qw &  \qw   \\
		& \push{\tilde{\cH}_{\sfP}}\qw & \ghost{\tilde{\Sigma}_{\sfP}} &\push{\tilde{\cE}_\sfP} \ww  & \ww 
	}.
	\end{align}
\end{Prop}

Before presenting the proof, we again make a few remarks.

\begin{Remark}
Tracing out the environment on each side of \cref{eq:MeasExt} leaves the identity  
\begin{align}
    \scalemyQ{1}{0.7}{0.7}{
		& \qw	& \push{\what{X}_{\sfP}} \qw    &  \qw &  \multigate{1}{\Lambda_{\sfP}}  &   \push{\what{Y}_{\sfP}} \qw &  \qw & \qw \\
		& \push{\tilde{\cH}_{\sfP}}\qw &   \gate{\tilde{\Xi}_{\sfP}} & \push{\cH_{\sfP}} \qw &  \ghost{\Lambda_{\sfP}}  
	} \quad = \quad 
	\scalemyQ{1}{0.5}{0.5}{
		&  \push{\what{X}_{\sfP}}   \qw  & \multigate{1}{\tilde{\Lambda}_{\sfP}}  &    \push{\what{Y}_{\sfP}} \qw &  \qw   \\
		& \push{\tilde{\cH}_{\sfP}}\qw & \ghost{\tilde{\Lambda}_{\sfP}} 
	}, \end{align}
	which expresses that the ensemble $\tilde{\Lambda}_{\sfP}$ can be extracted from the ensemble $\Lambda_{\sfP}$.  \cref{eq:MeasExt} is thus a strong version of extractability, according to which even \myuline{dilations} of $\tilde{\Lambda}_\sfP$ can be extracted from dilations of $\Lambda_\sfP$. (The same, incidentally, can be said about \cref{prop:StateExtract} regarding dilations of states.)
\end{Remark}

\begin{Remark}
The rank-assumption in \cref{prop:MeasExtract} is necessary since there is no way of determining the action of $\tilde{\Lambda}_{\sfP}$ outside the support of $\tilde{\varrho}_{\sfP}$ given only a Stinespring implementation of $\tilde{S}$.
\end{Remark}

\begin{Remark}[Generalisation to the Approximate Case] As follows from the very last argument in the proof of \cref{prop:MeasExtract}, approximate versions of this result might lose factors, with the loss depending on the spectrum of the density operator $\tilde{\varrho}_{\sfP}$. \end{Remark}

\begin{proof}[Proof (of \cref{prop:MeasExtract})]
Let us assume without loss of generality that $\sfP= \sfA$. Let us also assume for simplicity that the state $\tilde{\varrho} =: \tilde{\psi}$ is pure (the idea of the proof is not conceptually different in the general case, but less clearly exposed). Since $S \leq_{l.a.s.} \tilde{S}$, we know that
\begin{align} \label{eq:asim}
\myQ{0.7}{0.5}{
	& \push{\what{X}_\sfA} \qw  & \multigate{1}{\Sigma_\sfA}  &   \qw & \push{\what{Y}_\sfA} \qw &  \qw  \\
	& \Nmultigate{2}{\psi}   & \ghost{\Sigma_\sfA} &  \push{\cE_\sfA} \ww &  \ww   \\
	& \Nghost{\psi}  &   \ww & \push{\cE_0}\ww &  \ww \\
	& \Nghost{\psi}  & \multigate{1}{\Sigma_\sfB} & \push{\cE_\sfB} \ww  & \ww   \\ 
	&   \push{\what{X}_\sfB} \qw  & \ghost{\Sigma_\sfB}  &   \qw & \push{\what{Y}_\sfB} \qw &   \qw  \\ 
}
\quad =
\quad
\myQ{0.7}{0.5}{
	&  \push{\what{X}_\sfA}   \qw  & \multigate{1}{\tilde{\Sigma}_\sfA}  &   \qw &  \push{\what{Y}_\sfA} \qw &  \qw  & \qw \\
	& \Nmultigate{4}{\tilde{\psi}}  & \ghost{\tilde{\Sigma}_\sfA} & \push{\tilde{\cE}_\sfA} \ww &  \Nmultigate{1}{\Gamma_\sfA}{\ww} & \push{\cE_\sfA} \ww  \\
	& \Nghost{\tilde{\psi}}   &  & \Nmultigate{2}{\alpha} & \Nghost{\Gamma_\sfA}{\ww} &  \\  
	& \Nghost{\tilde{\psi}}  &   & \Nghost{\alpha}  & \ww  & \push{\cE_0} \ww & \ww \\
	& \Nghost{\tilde{\psi}}  &  & \Nghost{\alpha} &   \Nmultigate{1}{\Gamma_\sfB}{\ww}   \\
	& \Nghost{\tilde{\psi}}   & \multigate{1}{\tilde{\Sigma}_\sfB} & \push{\tilde{\cE}_\sfB} \ww & \Nghost{\Gamma_\sfB}{\ww} & \push{\cE_\sfB}\ww & \ww  \\ 
	&  \push{\what{X}_\sfB}  \qw  & \ghost{\tilde{\Sigma}_\sfB}   & \qw &  \push{\what{Y}_\sfB} \qw &   \qw & \qw \\ 
}
\quad 
\end{align}
for some state $\alpha$ and some channels $\Gamma_\sfA$ and $\Gamma_\sfB$ (the channel $\Gamma_0$ appearing in the definition of local assisted simulation can be absorbed into the assistant state $\alpha$ because it has no input from a purifying system for $\tilde{\psi}$). We may assume without loss of generality that $\Gamma_\sfB$ is reversible. (If not, we replace $\Gamma_\sfB$ by a Stinespring dilation and absorb the pure state arising on the left hand side into $\psi$, thus effectively replacing the strategy $S$ by a strategy $S'$ for which the measurement ensemble corresponding to $\sfA$ is still $\Lambda_\sfA$.)\footnote{In the case of a general state $\tilde{\varrho}$, there would be as well a channel $\Gamma_0$ in \cref{eq:asim} which may by a similar argument also be assumed reversible.}

Now, by \cref{lem:IsoRev} the isometric channel $\Sigma_\sfA$ is reversible. If we apply a left-inverse on both sides of \cref{eq:asim}, the resulting top output system is $\what{X}_\sfA$. When we trash this output and insert an arbitrary state on input $\what{X}_\sfA$, we obtain
\begin{align} 
\myQ{0.7}{0.5}{
	& \Nmultigate{2}{\psi}   &  \qw &  \push{\cH_\sfA}\qw &  \qw   \\
	& \Nghost{\psi}  &   \ww & \push{\cE_0} \ww &  \ww \\
	& \Nghost{\psi}  & \multigate{1}{\Sigma_\sfB} & \push{\cE_\sfB} \ww  & \ww   \\ 
	&   \push{\what{X}_\sfB} \qw  & \ghost{\Sigma_\sfB}  &   \qw & \push{\what{Y}_\sfB} \qw &   \qw  \\ 
}
\quad =
\quad
\myQ{0.7}{0.5}{
	& \Nmultigate{4}{\tilde{\psi}}  & \qw & \push{\tilde{\cH}_\sfA} \qw &   \multigate{1}{\tilde{\Xi}'_\sfA} & \push{\cH_\sfA}\qw & \qw  \\
	& \Nghost{\tilde{\psi}}   &  & \Nmultigate{2}{\alpha} & \Nghost{\tilde{\Xi}'_\sfA}{\ww} &  \\  
	& \Nghost{\tilde{\psi}}  &   & \Nghost{\alpha}    & \ww & \push{\cE_0}\ww & \ww \\
	& \Nghost{\tilde{\psi}}  &  & \Nghost{\alpha} &   \Nmultigate{1}{\Gamma_\sfB}{\ww}   \\
	& \Nghost{\tilde{\psi}}   & \multigate{1}{\tilde{\Sigma}_\sfB} & \push{\tilde{\cE}_\sfB} \ww & \Nghost{\Gamma_\sfB}{\ww} & \push{\cE_\sfB}\ww & \ww  \\ 
	&  \push{\what{X}_\sfB}  \qw  & \ghost{\tilde{\Sigma}_\sfB}   & \qw &  \push{\what{Y}_\sfB} \qw &   \qw & \qw \\ 
}
\end{align}
for some channel $\tilde{\Xi}'_\sfA$. By inserting this fragment back into the left hand side of \cref{eq:asim}, we find 
\begin{align}
\myQ{0.7}{0.5}{
	& \qw	& \push{\what{X}_\sfA} \qw    & \qw & \qw &  \multigate{1}{\Sigma_\sfA}  & \push{\what{Y}_\sfA} \qw &  \qw  \\
	& \Nmultigate{4}{\tilde{\psi}}  & \qw & \push{\tilde{\cH}_\sfA} \qw &   \multigate{1}{\tilde{\Xi}'_\sfA} & \ghost{\Sigma_\sfA}  & \push{\cE_\sfA}\ww & \ww\\
	& \Nghost{\tilde{\psi}}   &  & \Nmultigate{2}{\alpha} & \Nghost{\tilde{\Xi}'_\sfA}{\ww} &  \\  
	& \Nghost{\tilde{\psi}}  &   & \Nghost{\alpha}   & \ww & \push{\cE_0} \ww & \ww \\
	& \Nghost{\tilde{\psi}}  &  & \Nghost{\alpha} &   \Nmultigate{1}{\Gamma_\sfB}{\ww}   \\
	& \Nghost{\tilde{\psi}}   & \multigate{1}{\tilde{\Sigma}_\sfB} & \push{\tilde{\cE}_\sfB} \ww & \Nghost{\Gamma_\sfB}{\ww} & \push{\cE_\sfB} \ww  & \ww \\ 
	&  \push{\what{X}_\sfB}  \qw  & \ghost{\tilde{\Sigma}_\sfB}   & \qw & \qw  &  \push{\what{Y}_\sfB} \qw &   \qw & \qw \\ 
}
\quad  = \quad 
\myQ{0.7}{0.5}{
	&  \push{\what{X}_\sfA}   \qw  & \multigate{1}{\tilde{\Sigma}_\sfA}  &   \qw &  \push{\what{Y}_\sfA} \qw &  \qw  & \qw \\
	& \Nmultigate{4}{\tilde{\psi}}  & \ghost{\tilde{\Sigma}_\sfA} & \push{\tilde{\cE}_\sfA} \ww &  \Nmultigate{1}{\Gamma_\sfA}{\ww} & \push{\cE_\sfA}\ww & \ww  \\
	& \Nghost{\tilde{\psi}}   &  & \Nmultigate{2}{\alpha} & \Nghost{\Gamma_\sfA}{\ww} &  \\  
	& \Nghost{\tilde{\psi}}  &   & \Nghost{\alpha}    & \ww &  \push{\cE_0}\ww & \ww \\
	& \Nghost{\tilde{\psi}}  &  & \Nghost{\alpha} &   \Nmultigate{1}{\Gamma_\sfB}{\ww}   \\
	& \Nghost{\tilde{\psi}}   & \multigate{1}{\tilde{\Sigma}_\sfB} & \push{\tilde{\cE}_\sfB} \ww & \Nghost{\Gamma_\sfB}{\ww} & \push{\cE_\sfB}\ww & \ww  \\ 
	&  \push{\what{X}_\sfB}  \qw  & \ghost{\tilde{\Sigma}_\sfB}   & \qw &  \push{\what{Y}_\sfB} \qw &   \qw & \qw \\ 
}.
\end{align}

Now, $\Gamma_\sfB$ was assumed reversible, so we may cancel it by a left-inverse.\footnote{In the case of general $\tilde{\varrho}$, we would cancel $\Gamma_0$ as well.} Afterwards we may then cancel the reversible channel $\tilde{\Sigma}_\sfB$, which yields the identity
\begin{align}
\myQ{0.7}{0.5}{
	& \qw	& \push{\what{X}_\sfA} \qw    & \qw & \qw &  \multigate{1}{\Sigma_\sfA}   & \push{\what{Y}_\sfA} \qw &  \qw  \\
	& \Nmultigate{4}{\tilde{\psi}}  & \qw & \push{\tilde{\cH}_\sfA} \qw &   \multigate{1}{\tilde{\Xi}'_\sfA} & \ghost{\Sigma_\sfA}  & \push{\cE_\sfA} \ww & \ww\\
	& \Nghost{\tilde{\psi}}   &  & \Nmultigate{2}{\alpha} & \Nghost{\tilde{\Xi}'_\sfA}{\ww} &  \\  
	& \Nghost{\tilde{\psi}}  &   & \Nghost{\alpha}    & \push{\cE_0}\ww & \ww  \\
	& \Nghost{\tilde{\psi}}  &  & \Nghost{\alpha}  &\push{\cK_\sfB} \ww & \ww \\
	& \Nghost{\tilde{\psi}}   & \qw & \push{\tilde{\cH}_\sfB}\qw & \qw  \\ &\qw &  \push{\what{X}_\sfB} \qw & \qw & \qw
}
 \quad = \quad 
\myQ{0.7}{0.5}{
	&  \push{\what{X}_\sfA}   \qw  & \multigate{1}{\tilde{\Sigma}_\sfA}  &   \qw &  \push{\what{Y}_\sfA} \qw &  \qw  & \qw \\
	& \Nmultigate{4}{\tilde{\psi}}  & \ghost{\tilde{\Sigma}_\sfA} &  \push{\tilde{\cE}_\sfA}\ww &  \Nmultigate{1}{\Gamma_\sfA}{\ww} & \push{\cE_\sfA}\ww & \ww  \\
	& \Nghost{\tilde{\psi}}   &  & \Nmultigate{2}{\alpha} & \Nghost{\Gamma_\sfA}{\ww} &  \\  
	& \Nghost{\tilde{\psi}}  &   & \Nghost{\alpha}    & \push{\cE_0}\ww & \ww\\
	& \Nghost{\tilde{\psi}}  &  & \Nghost{\alpha}  & \push{\cK_\sfB}\ww & \ww \\
	& \Nghost{\tilde{\psi}}   & \qw & \push{\tilde{\cH}_\sfB} \qw & \qw  \\ & \qw &  \push{\what{X}_\sfB} \qw & \qw & \qw
}.
\quad 
\end{align}
If we now insert an arbitrary state into $\what{X}_\sfB$ and trace out the output $\what{X}_\sfB$ along with the systems $\cE_0$ and $\cK_\sfB$, we finally obtain 
\begin{align} \label{eq:choi}
\myQ{0.7}{0.5}{
		& \push{\what{X}_\sfA} \qw   & \qw & \qw &  \multigate{1}{\Sigma_\sfA}   & \push{\what{Y}_\sfA} \qw &  \qw  \\
	& \Nmultigate{2}{\tilde{\psi}}  & \push{\tilde{\cH}_\sfA} \qw &   \gate{\tilde{\Xi}_\sfA} & \ghost{\Sigma_\sfA}  & \push{\cE_\sfA} \ww & \ww\\
	& \Nghost{\tilde{\psi}}    \\  
	& \Nghost{\tilde{\psi}}   & \push{\tilde{\cH}_\sfB}\qw & \qw
}
\quad =
\quad
\myQ{0.7}{0.5}{
	&  \push{\what{X}_\sfA}   \qw  & \qw &  \multigate{1}{\tilde{\Sigma}_\sfA}  &   \qw & \qw & \push{\what{Y}_\sfA} \qw &  \qw   \\
	& \Nmultigate{2}{\tilde{\psi}}  & \push{\tilde{\cH_\sfA}} \qw & \ghost{\tilde{\Sigma}_\sfA} &  \ww &  \Ngate{\tilde{\Psi}_\sfA}{\ww} & \push{\cE_\sfA} \ww & \ww  \\
	& \Nghost{\tilde{\psi}}     \\ 
	& \Nghost{\tilde{\psi}}   & \push{\tilde{\cH}_\sfB}\qw & \qw \\ 
} \quad 
\end{align}
for some channels $\tilde{\Xi}_\sfA$ and $\tilde{\Psi}_\sfA$.\footnote{For general $\tilde{\varrho}$, the identity would at this point be $\scalemyQ{.8}{0.7}{0.5}{
	& \push{\what{X}_\sfA} \qw  	& \qw  & \qw &  \multigate{1}{\Sigma_\sfA}   & \push{\what{Y}_\sfA} \qw &  \qw  \\
	& \Nmultigate{2}{\tilde{\psi}}  & \push{\tilde{\cH}_\sfA}\qw &   \gate{\tilde{\Xi}_\sfA} & \ghost{\Sigma_\sfA}  & \push{\cE_\sfA} \ww & \ww\\
	& \Nghost{\tilde{\psi}} & \push{\tilde{\cP}}\ww & \ww    \\  
	& \Nghost{\tilde{\psi}}   & \push{\tilde{\cH}_\sfB} \qw & \qw
}
\quad =
\quad
\scalemyQ{.8}{0.7}{0.5}{
	&  \push{\what{X}_\sfA}   \qw  & \qw &  \multigate{1}{\tilde{\Sigma}_\sfA}  &   \qw & \push{\what{Y}_\sfA} \qw &  \qw \\
	& \Nmultigate{2}{\tilde{\psi}}  & \push{\tilde{\cH}_\sfA} \qw & \ghost{\tilde{\Sigma}_\sfA} &   \Ngate{\tilde{\Psi}_\sfA}{\ww} & \push{\cE_\sfA}\ww & \ww  \\
	& \Nghost{\tilde{\psi}}    &\push{\tilde{\cP}} \ww & \ww \\ 
	& \Nghost{\tilde{\psi}}   & \push{\tilde{\cH}_\sfB}\qw & \qw \\ 
}$ for some purification $\tilde{\psi}$ of $\tilde{\varrho}$.} Lastly, since the $\sfA$-marginal of $\tilde{\psi}$ has full rank, the state $\tilde{\psi}$ gives rise to a Choi-Jamio\l{}kowski-like isomorphism, and for any input to $\what{X}_\sfA$ \cref{eq:choi} expresses the equality of two such Choi-Jamio\l{}kowski states. We thus conclude by injectivity that the two channels in \cref{eq:MeasExt} acts identically on any input to $\what{X}_\sfA$, and therefore must be identical, as desired. \end{proof}

We immediately conclude the following:
	
\begin{Cor} \label{cor:SelfMeasExtract}
	Suppose that $\tilde{S} = (\tilde{\varrho}, \tilde{\Lambda}_\sfA, \tilde{\Lambda}_\sfB)$ is a full-rank strategy which is self-tested according to local assisted simulation by its behaviour $P$. Let $\tilde{\Sigma}_\sfA$, $\tilde{\Sigma}_\sfB$ be Stinespring dilations of $\tilde{\Lambda}_\sfA$, respectively $\tilde{\Lambda}_\sfB$. For any strategy $S= (\varrho, \Lambda_\sfA, \Lambda_\sfB)$ with behaviour $P$, and for any Stinespring dilation $\Sigma_{\sfP}$ of $\Lambda_{\sfP}$, $\sfP \in \{\sfA, \sfB\}$, there exist channels $\tilde{\Xi}_{\sfP}$ and $\tilde{\Psi}_{\sfP}$ such that \cref{eq:MeasExt} holds.
\end{Cor}

\begin{Remark}[Equivalence of Measurement Ensembles]
Since assisted simulation is in quantum theory an equivalence relation (\cref{thm:AssSym}), the extractibility converse to \cref{prop:MeasExtract} can be deduced as well if the state $\varrho$ has full local rank. In other words, if $\tilde{S}$ is full-rank and self-tested a.t.l.a.s.~by the behaviour $P$, then for any full-rank strategy $S$ with behaviour $P$ the measurement ensembles of $S$ are \emph{equivalent} to the canonical measurement ensembles, in the sense that either can be extracted from the other. 
\end{Remark}

\subsection{Causal Simulation and Convex Decompositions --- Extremality}
\label{subsec:SecExt}

It is known that a necessary condition for a behaviour $P$ to self-test a strategy a.t.r.~is extremality of $P$ in the convex set of all behaviours realisable by quantum strategies. A formal proof of this has been given in Ref.~\cite{Goh18}, based on operator-algebraic considerations. 

Below, we present a new proof of the statement in operational language. Our proof branches into pieces which collectively reveal that a more general result holds for arbitrary operational theories, and in fact under the weaker notion of self-testing a.t.c.s.~(\cref{prop:SelfSec}). The extremality result for quantum self-testing is obtained as an immediate consequence of this (\cref{cor:extremal}). \\

Let $P$ be a realisable behaviour. By a \emph{convex decomposition of $P$} we mean formally a sequence $\big((p_k, P_k)\big)_{k=1, \ldots, n}$ such that $p_1, \ldots, p_n$ are strictly positive real numbers and such that $P_1,\ldots, P_n$ are realisable behaviours with $\sum_{k=1}^n p_k P_k = P$ (it follows that necessarily $\sum_{k=1}^n p_k =1$). 

Given a strategy $S=(\varrho, \Lambda_\sfA, \Lambda_\sfB)$ with behaviour $P$, every convex decomposition of the state $\varrho$ gives rise to a convex decomposition of $P$. Specifically, if $\varrho = \sum_{k=1}^n p_k \varrho_k$ and if we define $P_k$ as the behaviour of the strategy $(\varrho_k, \Lambda_\sfA, \Lambda_\sfB)$, then clearly $P= \sum_{k=1}^n p_k P_k$. Let us say that a convex decomposition $\big((p_k, P_k)\big)_{k=1, \ldots, n}$ is \emph{visible from $S$} if it arises in this way. (If for example $S$ is pure-state, then all convex decompositions visible from $S$ are trivial, meaning with $P_1 = \ldots = P_n$; indeed, any convex decomposition of the pure state must have $\varrho_1 = \ldots = \varrho_n$.)

Visibility of convex decompositions is linked to causal simulation as follows: 

\begin{Lem} \label{lem:visible}
Let $S$ and $\tilde{S}$ be strategies with common behaviour $P$. If $S \leq_{c.s.} \tilde{S}$, then any convex decomposition of $P$ visible from $S$ is also visible from $\tilde{S}$.    
\end{Lem}

\begin{proof}
Let $S=(\varrho, \Lambda_\sfA, \Lambda_\sfB)$ and $\tilde{S}=(\tilde{\varrho}, \tilde{\Lambda}_\sfA, \tilde{\Lambda}_\sfB)$. The main observation to make is that for any dilation $\xi$ of $\varrho$ with environment $\cE_0$, there exists a dilation $\tilde{\xi}$ of $\tilde{\varrho}$, with environment $\cE_0$, such that 
\begin{align} \label{eq:statedil}  
\myQ{0.7}{0.7}{
	& \push{\what{X}_\sfA}  \qw  & \multigate{1}{ \tilde{\Lambda}_\sfA} & \push{\what{Y}_\sfA}  \qw & \qw  \\
	& \Nmultigate{2}{\tilde{\xi}}  & \ghost{\tilde{\Lambda}_\sfA} \\
	& \Nghost{\tilde{\xi}} & \push{\cE_0} \ww  & \ww \\
	& \Nghost{\tilde{\xi}} & \multigate{1}{\tilde{\Lambda}_\sfB}  \\
	&   \push{\what{X}_\sfB} \qw  & \ghost{\tilde{\Lambda}_\sfB} & \push{\what{Y}_\sfB}  \qw & \qw  \\
}\quad = \quad 
\myQ{0.7}{0.7}{
	& \push{\what{X}_\sfA}  \qw  & \multigate{1}{ \Lambda_\sfA} & \push{\what{Y}_\sfA}  \qw & \qw  \\
	& \Nmultigate{2}{\xi}  & \ghost{\Lambda_\sfA} \\
	& \Nghost{\xi} & \push{\cE_0} \ww  & \ww \\
	& \Nghost{\xi} & \multigate{1}{\Lambda_\sfB} \\
	&   \push{\what{X}_\sfB} \qw  & \ghost{\Lambda_\sfB} & \push{\what{Y}_\sfB}  \qw & \qw  \\
} .
\end{align}
By completeness of purifications (\cref{thm:StineComp}), it suffices to show this statement when $\xi = \psi$ is a purification of $\varrho$. This case however follows from the condition which defines the relation $S \leq_{c.s.} \tilde{S}$ in terms of Stinespring implementations (\cref{eq:CausSim}, with $\tilde{S}$ in place of $S$ and $S$ in place of $S'$), by tracing out the environments corresponding to the Stinespring dilations $\Sigma_{\sfP}$ and $\tilde{\Sigma}_{\sfP}$ of the measurements ensembles. 

Having established this fact, let $\big((p_k, P_k)\big)_{k=1, \ldots, n}$ be a convex decomposition visible from $S$. Pick states $\varrho_1, \ldots, \varrho_n$ such that $\varrho = \sum_{k=1}^n p_k \varrho_k$ and $(\varrho_k, \Lambda_\sfA, \Lambda_\sfB)$ has behaviour $P_k$. Now, consider the dilation $\xi := \sum_{k=1}^n p_k \varrho_k \otimes \ketbra{k}$ of $\varrho$, with environment $\cE_0 = \C^n$. By the introductory observation, we find a dilation $\tilde{\xi}$ and $\tilde{\varrho}$ such that \cref{eq:statedil} holds. The channel on the right hand side of \cref{eq:statedil} is given by $\sum_{k=1}^n p_k P_k \otimes \ketbra{k}$. Since the left hand side must be identical, the $\cE_0$-marginal of $\tilde{\xi}$ is $\sum_{k=1}^n p_k \ketbra{k}$, so $\tilde{\xi}$ itself must be of the form $\sum_{k=1}^n p_k \tilde{\varrho}_k \otimes \ketbra{k}$ for some states $\tilde{\varrho}_1, \ldots, \tilde{\varrho}_n$ with $\sum_{k=1}^n p_k \tilde{\varrho}_k = \tilde{\varrho}$. Letting $\tilde{P}_k$ denote the behaviour of $(\tilde{\varrho}_k, \tilde{\Lambda}_\sfA, \tilde{\Lambda}_\sfB)$, the decomposition $\big((\tilde{P}_k, p_k) \big)_{k=1, \ldots, n}$ is obviously visible from $\tilde{S}$. But since the left hand side of \cref{eq:statedil} is given by $\sum_{k=1}^n p_k \tilde{P}_k \otimes \ketbra{k}$ and must coincide with the right hand side, we have $\tilde{P}_k = P_k$ for all $k=1, \ldots, n$ since $p_1, \ldots, p_n \neq 0$. Altogether we conclude that the decomposition $\big((p_k, P_k)\big)_{k=1, \ldots, n}$ is visible from $\tilde{S}$, as desired. 
\end{proof}

\begin{Remark}[Ekert-Implementations]
 The implementations appearing in \cref{eq:statedil} have a natural operational interpretation, especially in light of \cref{subsec:Causal}. Indeed, they correspond to causal dilations in which the side-information in the environment requires no causes thus thus exists in advance of seeing the inputs $x_\sfA$ and $x_\sfB$ to the Bell-scenario. Let us call an implementation of $S= (\varrho, \Lambda_\sfA, \Lambda_\sfB)$ an \emph{Ekert-implementation} if it corresponds to a component-wise dilation of the form $(\xi, \Lambda_\sfA, \Lambda_\sfB)$, where the environment of $\xi$ is a classical embedded system on which the state $\xi$ is classical and has full rank. (Essentially, Ekert's insight \cite{Ek91} concerns the leak of information to the environment in Ekert-implementations.) It is an easy exercise to show that Ekert-implementations of $S$ correspond to convex decompositions of $P$ visible from $S$, in the sense that they as channels are precisely given by $\sum_{k=1}^n p_k \ketbra{k} \otimes P_k$ with $\big((p_k, P_k)\big)_{k=1, \ldots, n}$ a convex decomposition visible from $S$. The essence of \cref{lem:visible} consists in observing that if $S \leq_{c.s.} \tilde{S}$ then every Ekert-implementation of $S$ is also an Ekert-implementation of $\tilde{S}$.\end{Remark}

\cref{lem:visible} can be used to show the following: 

\begin{Prop} \label{prop:SelfSec}
Suppose that $P$ self-tests $\tilde{S}$ a.t.c.s. Then, every convex decomposition of $P$ is visible from $\tilde{S}$.
\end{Prop}

\begin{proof}
    By \cref{lem:visible}, it suffices to argue that every possible convex decomposition of $P$ is visible from \myuline{some} strategy $S$ with behaviour $P$. But this is easy to see: Given a convex decomposition $\big((p_k, P_k) \big)_{k=1, \ldots, n}$, choose strategies $S_1 = (\varrho_1, \Lambda_{\sfA, 1}, \Lambda_{\sfB, 1}), \ldots, S_n= (\varrho_n, \Lambda_{\sfA, n}, \Lambda_{\sfB, n})$ such that $S_k$ has behaviour $P_k$. We may assume without loss of generality that the systems $\cH_{\sfA,k}$ are all equal to a fixed system, say $\cH_\sfA$,  by embedding them into a common system if necessary. Likewise, we may assume all the systems $\cH_{\sfB, k}$ equal to a fixed system $\cH_\sfB$. Now, let $S$ be the strategy whose state is $\sum_{k=1}^n p_k \ketbra{k}_\sfA \otimes \ketbra{k}_\sfB \otimes \varrho_k$ on the bipartite system $(\C^n \otimes \cH_\sfA) \otimes (\C^n \otimes \cH_\sfB)$, and whose measurement ensemble $\Lambda_{\sfP}$ is defined by reading off the classical value $k$ on the subsystem $\C^n$ and applying the according ensemble $\Lambda_{\sfP, k}$. The behaviour of the strategy $S$ is evidently $\sum_{k=1}^n p_k P_k = P$, and the convex decomposition $\big((p_k, P_k) \big)_{k=1, \ldots, n}$ is visible from $S$ by construction.  \end{proof}

\cref{prop:SelfSec} constitutes our key result about self-testing and convex decompositions: it expresses that a strategy which is self-tested according to causal simulation fully determines all possible decompositions of its behaviour. 

In quantum theory, we may use the result to infer \myuline{extremality} of behaviours which self-test strategies (even according to \myuline{local assisted} simulation):

\begin{Cor} \label{cor:extremal}
    If $P$ self-tests some quantum strategy a.t.l.a.s., then $P$ is extremal in the convex set of quantum realisable behaviours. 
\end{Cor}

\begin{proof}
If $P$ self-tests some quantum strategy a.t.l.a.s., then all quantum strategies with behaviour $P$ are equivalent under local assisted simulation by \cref{thm:AssSym}. Therefore, $P$ self-tests \myuline{any} strategy with behaviour $P$ a.t.l.a.s. By the purification argument in the proof of \cref{cor:SelfStateExtract}, there is some pure-state strategy $\tilde{S}$ with behaviour $P$, and it holds in particular that $P$ self-tests $\tilde{S}$ a.t.l.a.s. By \cref{prop:SelfSec}, any convex decomposition of $P$ is visible from $\tilde{S}$, and as $\tilde{S}$ is pure-state every convex decomposition visible from $\tilde{S}$ is trivial. Thus, $P$ is extremal. \end{proof}

We note that there are known examples of extremal behaviours which do not self-test any strategy according to reducibility \cite{Jed21}. However, some or all of these might self-test a strategy according to local assisted simulation (a.t.l.a.s.); indeed the converse to \cref{cor:extremal} might be true. We leave this question open for future investigations.

\subsection{Exhausted Environments} 
\label{subsec:Exhausted}

In this final subsection, we point out a feature of self-testing which we believe has so far gone unnoticed. We do not currently know any concrete applications of this feature, but it shows up naturally from our operational reformulation and we suspect that it might be useful.\\

Suppose that parties $\sfA$ and $\sfB$ manage to produce a behaviour $P$ which we know to self-test some strategy $\tilde{S}$ according to causal simulation~(in particular, this holds if $P$ self-tests $\tilde{S}$ in the conventional sense, i.e.~according to reducibility). From the perspective of a referee who probes the behaviour by sending inputs $x_\sfA, x_\sfB$ and receiving outputs $y_\sfA, y_\sfB$, the information-processing is described by the channel $P$. From a global perspective including the environment, however, the information-processing is governed by the various implementations 
\begin{align} \label{eq:Imple}
	    \myQ{0.7}{0.7}{
	& \push{\what{X}_\sfA}  \qw  & \multigate{1}{ \Phi_\sfA} & \push{\what{Y}_\sfA}  \qw & \qw  \\
	& \Nmultigate{2}{\xi}  & \ghost{\Phi_\sfA} & \push{\cE_\sfA} \ww & \ww \\
	& \Nghost{\xi}  & \push{\cE_0} \ww & \ww & \ww  \\
	& \Nghost{\xi}  & \multigate{1}{\Phi_\sfB} & \push{\cE_\sfB} \ww & \ww  \\
	&   \push{\what{X}_\sfB} \qw  &  \ghost{\Phi_\sfB} & \push{\what{Y}_\sfB}  \qw & \qw  \\
} 
	\end{align}
	of a strategy $(\varrho, \Lambda_\sfA, \Lambda_\sfB)$. (By the discussion in \cref{subsubsec:SpecialBell}, this is even true when strategies corresponding to general circuits are employed; any possible information-processing is derivable from an implementation of the form \eqref{eq:Imple}.) 
	
	Now, let us consider what the information-processing looks like \emph{from the perspective of the environment}. This amounts to tracing out the system $\what{Y}_\sfA$ and $\what{Y}_\sfB$ and considering the channel 
	\begin{align} \label{eq:Comple}
	    \myQ{0.7}{0.7}{
	& \push{\what{X}_\sfA}  \qw   & \multigate{1}{ \Phi_\sfA} & \push{\what{Y}_\sfA}  \qw & \gate{\tr}  \\
	& \Nmultigate{2}{\xi}  &  \ghost{\Phi_\sfA} & \push{\cE_\sfA} \ww & \ww \\
	& \Nghost{\xi}  & \push{\cE_0} \ww & \ww & \ww  \\
	& \Nghost{\xi} &  \multigate{1}{\Phi_\sfB} & \push{\cE_\sfB} \ww & \ww  \\
	&   \push{\what{X}_\sfB} \qw  &  \ghost{\Phi_\sfB} & \push{\what{Y}_\sfB}  \qw & \gate{\tr}  \\
} .
	\end{align}
When the implementation is a Stinespring implementation, the channel \eqref{eq:Comple} is a \emph{complementary channel} \cite{Devetak05} to the behaviour channel $P$, however as a \myuline{causal} channel (cf.~\cref{subsec:Causal}) it has additionally a specific structure according to which the information in $\cE_0$ is available from the beginning and the information in $\cE_\sfP$ becomes available upon the input to $\what{X}_\sfP$. Since $P$ self-tests $\tilde{S}$ a.t.c.s., the channel \eqref{eq:Comple} must be of the form 	
\begin{align}
\scalemyQ{1}{0.7}{0.5}{
	&   \push{\what{X}_\sfA} \qw  & \multigate{1}{\tilde{\Sigma}_\sfA}  &   \qw &  \push{\what{Y}_\sfA}  \qw &  \gate{\tr} \\
	& \Nmultigate{4}{\tilde{\psi}}  & \ghost{\tilde{\Sigma}_\sfA} &  \push{\tilde{\cE}_\sfA} \ww &  \Nmultigate{1}{\Gamma_\sfA}{\ww} & \push{\cE_\sfA} \ww & \ww \\
	& \Nghost{\tilde{\psi}}   &  & \Nmultigate{2}{\Gamma_0} & \Nghost{\Gamma_\sfA}{\ww} &  \\  
	& \Nghost{\tilde{\psi}}  & \push{\tilde{\cE}_0} \ww &   \Nghost{\Gamma_0}{\ww}  & \ww & \push{\cE_0} \ww & \ww \\
	& \Nghost{\tilde{\psi}}  &  & \Nghost{ \Gamma_0} &   \Nmultigate{1}{\Gamma_\sfB}{\ww}   \\
	& \Nghost{\tilde{\psi}}   & \multigate{1}{\tilde{\Sigma}_\sfB} &  \push{\tilde{\cE}_\sfB} \ww & \Nghost{\Gamma_\sfB}{\ww} & \push{\cE_\sfB}\ww & \ww \\ 
	&   \push{\what{X}_\sfB} \qw  & \ghost{\tilde{\Sigma}_\sfB}   & \qw &  \push{\what{Y}_\sfB} \qw &   \gate{\tr}}
\end{align}
for some channels $\Gamma_0$, $\Gamma_\sfA$ and $\Gamma_\sfB$, with $(\tilde{\psi}, \tilde{\Sigma}_\sfA, \tilde{\Sigma}_\sfB)$ a component-wise Stinespring dilation of $\tilde{S}$. As such, when the referee proves the behaviour by supplying only classical inputs to $\what{X}_\sfP$, the strongest possible leak of information to the environment is represented by the channel
	\begin{align} \label{eq:Complement}
	    \myQ{0.7}{0.7}{
& \push{\what{X}_\sfA}  \qw & \gate{\Delta_{X_\sfA}} 	& \push{\what{X}_\sfA}  \qw   & \multigate{1}{ \tilde{\Sigma}_\sfA} & \push{\what{Y}_\sfA}  \qw & \gate{\tr}  \\
&&	& \Nmultigate{2}{\tilde{\psi}}  &  \ghost{\Phi_\sfA} & \push{\tilde{\cE}_\sfA} \ww & \ww \\
&&	& \Nghost{\tilde{\psi}}  & \push{\tilde{\cE}_0} \ww & \ww & \ww  \\
&&	& \Nghost{\tilde{\psi}} &  \multigate{1}{\tilde{\Sigma}_\sfB} & \push{\tilde{\cE}_\sfB} \ww & \ww  \\
& \push{\what{X}_\sfB}  \qw & \gate{\Delta_{X_\sfB}} 	&   \push{\what{X}_\sfB} \qw  &  \ghost{\tilde{\Sigma}_\sfB} & \push{\what{Y}_\sfB}  \qw & \gate{\tr}  \\
} ,
	\end{align}
and this channel can be \myuline{computed} by the referee, because the strategy $\tilde{S}$ is known. Specifically, the information in $\tilde{\cE}_0$ is the maximal information the environment may hold before the probing begins, and the information in $\tilde{\cE}_\sfP$ is the maximal information which may leak to the environment after receiving the input to $\what{X}_\sfP$. 

It is no surprise that, when computing the channel \eqref{eq:Complement}, the system $\tilde{\cE}_\sfP$ will at least hold a copy of the input $x_\sfP$ and the output $y_\sfP$, or information equivalent to these copies; indeed these two classical values may always be copied by the environment and kept in memory. It is however perfectly conceivable that the environment will after the interaction with the referee hold information additional to these memories, information which we may think of as undisclosed to the referee. (For example, one could imagine that $\cE_\sfA$ would hold a copy of the output $y_\sfB$, or that $\cE_0$ would hold non-trivial information about either output.) 

We will show now that this does \myuline{not} happen when $\tilde{S}$ is a pure-state projective strategy whose measurements are described by rank-one projections. This is the case e.g.~in the optimal canonical strategy for the CHSH-game \cite{CHSH69,SB19}. In other words, two parties $\sfA$ and $\sfB$ who play this game optimally essentially cannot know anything in advance of the game, and will be forced to reveal as output to the referee everything they learn upon receiving their inputs. Hence, they are effectively \emph{exhausted} of information. \\

The proof of this is fairly simple. Suppose that the canonical strategy $\tilde{S}=(\tilde{\psi}, \tilde{\Pi}_\sfA, \tilde{\Pi}_\sfB)$ is pure-state and projective, and that all of the projections $\tilde{\Pi}^{x_{\sfP}}_{\sfP}(y_{\sfP})$ have rank one (this implies, incidentally, than, $\dim \tilde{\cH}_{\sfP} = \abs{Y_{\sfP}}$). If we choose the Stinespring dilations $\tilde{\Sigma}_\sfA$ and $\tilde{\Sigma}_\sfB$ conveniently as in the proof of \cref{prop:OpSim}, then the channel 
\begin{align} \label{eq:Comp3}
	    \myQ{0.7}{0.7}{
& \push{\what{X}_\sfA}  \qw & \gate{\Delta_{X_\sfA}} 	& \push{\what{X}_\sfA}  \qw   & \multigate{1}{ \tilde{\Sigma}_\sfA} & \push{\what{Y}_\sfA}  \qw & \gate{\Delta_{Y_\sfA}} & \push{\what{Y}_\sfA} \qw  & \qw \\
& & 	& \Nmultigate{2}{\tilde{\psi}}  &  \ghost{\Phi_\sfA} & \push{\tilde{\cE}_\sfA} \ww & \ww \\
& & 	& \Nghost{\tilde{\psi}}   \\
& & 	& \Nghost{\tilde{\psi}} &  \multigate{1}{\tilde{\Sigma}_\sfB} & \push{\tilde{\cE}_\sfB} \ww & \ww  \\
& \push{\what{X}_\sfB}  \qw & \gate{\Delta_{X_\sfB}} 	&   \push{\what{X}_\sfB} \qw  &  \ghost{\tilde{\Sigma}_\sfB} & \push{\what{Y}_\sfB}  \qw & \gate{\Delta_{Y_\sfB}} & \push{\what{Y}_\sfB} \qw & \qw \\
} ,
	\end{align}
where $\tilde{\cE}_\sfA = \what{X}_\sfA \otimes \tilde{\cH}_\sfA$ and $\tilde{\cE}_\sfB = \what{X}_\sfB \otimes \tilde{\cH}_\sfB$ (and where $\tilde{\cE}_0$ is trivial by purity of the state), is given by  
\begin{align} \label{eq:XiCh}
    F & \mapsto  \sum_{x \in X, y \in Y} \Big(  \bra{x}F\ket{x}  \ketbra{x} \otimes  \tilde{\Pi}^x(y) \ketbra*{\tilde{\psi}} \tilde{\Pi}^{x}(y) \Big)\otimes \ketbra{y}{y}
\end{align}
for $F \in \End{\what{X}_\sfA \otimes \what{X}_\sfB}$, parenthesising the parts which belong to the environment. Now, if we let $\ket{\phi^{x_\sfP}_\sfP(y_\sfP)}$ denote a unit vector corresponding to the rank-one projection $\tilde{\Pi}^{x_{\sfP}}_\sfP(y_{\sfP})$, and if $\ket{\phi^x(y)} = \ket{\phi^{x_\sfA}_\sfA(y_\sfA)} \otimes \ket{\phi^{x_\sfB}_\sfB(y_\sfB)}$ for $x=(x_\sfA, x_\sfB)$ and $y=(y_\sfA, y_\sfB)$, then we have the identity $\tilde{\Pi}^x(y) \ket*{\tilde{\psi}} = \braket*{\phi^x(y)}{\tilde{\psi}} \ket{\phi^x(y)}$, with $\abs*{\braket*{\phi^x(y)}{\tilde{\psi}}}^2$ equal to the probability of output $y$ on input $x$, namely $P^x(y)$. Hence, the channel \eqref{eq:XiCh} can be re-expressed as 
\begin{align}
    F& \mapsto  \sum_{x \in X, y \in Y} \bra{x}F\ket{x}  P^x(y) \Big(  \ketbra{x} \otimes  \ketbra{\phi^x(y)} \Big) \otimes \ketbra{y}{y} .
\end{align}

For any $\sfP \in  \{\sfA, \sfB\}$ and any $x_{\sfP} \in X_{\sfP}$, the pure states $(\phi^{x_{\sfP}}_{\sfP}(y_{\sfP}))_{y_{\sfP} \in Y_{\sfP}}$ are perfectly distinguishable, indeed their vector representatives form an orthonormal basis of $\tilde{\cH}_{\sfP}$. They can therefore be viewed as an encoding of $y_{\sfP} \in Y_{\sfP}$ in terms of pure states on $\tilde{\cH}_{\sfP} \cong \what{Y}_\sfP$. More precisely, there exist unitaries $U_\sfP : \what{X}_\sfP \otimes \tilde{\cH}_\sfP \to \what{X}_\sfP \otimes \what{Y}_\sfP$ which map $\ket{x_\sfP} \otimes \ket{\phi^{x_{\sfP}}_{\sfP}(y_{\sfP})}$ to $\ket{x_\sfP} \otimes \ket{y_\sfP}$. As such, the channel \eqref{eq:Comp3} is equivalent (under application of unitaries $U_\sfA$ and $U_\sfB$ to the systems $\tilde{\cE}_\sfA$ and $\tilde{\cE}_\sfB$) to the channel given by
\begin{align}
   F \mapsto  \sum_{x \in X, y \in Y} \bra{x}F \ket{x}  P^x(y) \Big(  \ketbra{x} \otimes  \ketbra{y} \Big) \otimes \ketbra{y}{y} .
\end{align}
In other words, when executing the behaviour $P$ the only possible information which can possibly reside in the environment is, upon input $x_\sfP$, a copy of $x_\sfP$ and the corresponding output $y_\sfP$.

\section{Conclusions}
\label{sec:Conclusions}

Quantum self-testing is an important phenomenon in quantum cryptography which arose from Bell's considerations about experimental scenarios with separated parties. Today, it has many theoretical and practical ramifications in quantum information theory. However, the conventional mathematical formalisation of self-testing is phrased in the operator-algebraic formalism for quantum theory, and therefore resists an operational interpretation and immediate generalisation e.g.~to alternative theories, formalisms or causal scenarios. \\

After discussing operational structure in physical theories (\cref{sec:Operational}) and reviewing the conventional operator-algebraic narrative of quantum self-testing (\cref{sec:Standard}), we have turned in \cref{sec:Rework} to rework the perception of self-testing and cast it anew as an operational phenomenon. 

In \cref{subsec:Sim} we have proposed a new definition of self-testing (\cref{def:Sim}) according to the concept of local simulation (\cref{def:OpSelf}). It specialises in the case of quantum theory to the standard operator-algebraic definition (\cref{thm:SelfRvsS}), it answers an unresolved issue regarding a restriction to projective strategies (\cref{prop:Proj}), and it naturally lends itself to definitions of approximate self-testing (\cref{rem:Approx}).

Our reformulation shows how to think of self-testing operationally, in terms of side-information which gradually leaks to an environment. By meditating on this idea, we have proposed two relaxations of local simulation, namely local assisted simulation (\cref{def:LocAss}) and causal simulation (\cref{def:CausSim}). In \cref{subsec:AssSim}, we have studied local assisted simulation and seen that it is in quantum theory the equivalence relation generated by local simulation (\cref{thm:AssSym} and \cref{rem:SelfEquiv}). This sheds new light on quantum self-testing by uncovering a two-fold character of the phenomenon. In particular, it suggests a notion of self-testing with emphasis on upper bounding the abilities of adversaries (\cref{def:SelfTestAtlas}). Causal simulation yields a further relaxation of self-testing, and we have explained in \cref{subsec:Causal} how this ultimately realises quantum self-testing within a larger framework of causally structured channels and dilations (most important in this regard are \cref{thm:BellDense}, \cref{thm:BellDenseII} and \cref{thm:BridgeSC}). It is our hope that this realisation may eventually be beneficial in cryptographic contexts and point to modular, composable features of self-testing (see also Section 4.3 in Ref.~\cite{Hou21}).

In \cref{sec:Utility}, we have shown that our operational reformulation of self-testing may be used to give conceptually simple proofs of several features of self-testing. Specifically, we have considered extractability of canonical states (\cref{prop:StateExtract}) and canonical measurements (\cref{prop:MeasExtract}), limitations on convex decompositions of behaviours (\cref{prop:SelfSec} and \cref{cor:extremal}), and a novel feature of information exhaustion (\cref{subsec:Exhausted}). The fact that we are able to recover well-known features from our reformulation indicates that it is at least as potent as the conventional formulation --- we hope, of course, that our reformulation will eventually prove itself even more potent by easing more arguments which involve self-testing. \\

The current work opens up several avenues for future investigation. 
A rather concrete problem is the further study of \emph{approximate} simulation and \emph{robust} self-testing. As mentioned in \cref{rem:Approx}, the Bures distance $\beta$ \cite{KSW08math,KSW08phys} and purified diamond distance $P_\diamond$ \cite{Hou21} have properties which render them suitable for quantifying approximate simulation, but we have not determined the distances in concrete cases. In particular, if $\tilde{S}$ is the canonical optimal strategy for the CHSH-game \cite{CHSH69,SB19}, we do not know, given $\delta > 0$, what is the optimal value of $\varepsilon > 0$ such that any strategy $S$ whose behaviour is $\delta$-close to the behaviour of $\tilde{S}$ locally $\varepsilon$-simulates $\tilde{S}$ (with closeness measured either by $\beta$ or $P_\diamond$).\footnote{Likewise, we do not know the optimal value of $\varepsilon$ when instead of assuming the behaviours to be $\delta$-close it is merely assumed that the winning probability of $S$ is $\delta$-close to the optimal winning probability.} An important insight from our work is that the local isometries which relate one strategy to another according to local simulation can be chosen dependently on the local inputs $x_\sfA$ and $x_\sfB$ (cf.~\cref{prop:OpSim}); this turns out to be invisible in the case of exact local simulation (\cref{lem:technical}), but we suspect that it will result in differences in the approximate case. 

Another open problem is the following. Usually, a proof that a concrete behaviour $P$ self-tests a concrete strategy $\tilde{S}$ (e.g.~the optimal strategy for the CHSH-game) is obtained by deriving certain operator-algebraic identities which the constituents of a general strategy $S$ must satisfy if its behaviour is $P$, and then from these identities concluding the existence of local isometries and a residual state witnessing reducibility to the canonical strategy $\tilde{S}$. Given that we have been able to reformulate self-testing operationally, the following question is natural: \emph{Is it possible to find an operational \myuline{proof} that $\tilde{S}$ is self-tested by its behaviour?} Such a proof would not directly rely on operator-algebra, but would rather use operational features of the theory $\QIT$. Even for well-studied cases such as the canonical optimal strategy in the CHSH-game, an operational proof of self-testing (according to any mode of simulation) would yield new insights. It would even be interesting to find an operational proof of the fact that the CHSH-game has the property of exhausting quantum players of information, as described in \cref{subsec:Exhausted}.\footnote{It is worth observing that this proof would necessarily involve aspects of quantum theory and cannot be a feature of the CHSH-game by itself, since the game is not exhaustive for \myuline{classical} adversaries who play optimally --- for example, in the optimal classical strategy which is the equal probabilistic mixture of optimal deterministic strategies, both parties $\sfA$ and $\sfB$ know the function which the other party uses to determine an output, while this function need not be disclosed to the referee.}

Finally, it would be natural to investigate self-testing statements in other causal scenarios and other operational theories.\footnote{A few ideas in this direction are presented in Ref.~\cite{Hou21}, in particular as regards the theory $\CIT$.} A more systematic study of this might clarify the special case of quantum self-testing and eventually bring new insights to the field.

\section*{Acknowledgements}  
    
We thank J\k{e}drzej Kaniewski for comments on an earlier version of this manuscript. We acknowledge financial support from the European Research Council (ERC Grant Agreement No.~81876), VILLUM FONDEN via the QMATH Centre of Excellence (Grant No.~10059), the QuantERA ERA-NET Cofund in Quantum Technologies implemented within the European Union’s Horizon 2020 Programme (QuantAlgo project) via the Innovation Fund Denmark and the Novo Nordisk Foundation (Grant NNF20OC0059939 ‘Quantum for Life’). LM is also supported by VILLUM FONDEN (Young Investigator Grant No. 37532).

   \section*{Authors' Contributions}
   
This paper is based on the PhD thesis of NGHL \cite{Hou21}. NGHL is the main author of the paper, including conceptualisation and methodology of the work and the formal analysis and writing. MC and LM contributed to the conceptualisation and methodology of the work and reviewed and edited the draft. All authors gave final approval for publication and agree to be held accountable for the work performed therein.

\newpage


\bibliographystyle{amsalpha} 
\addcontentsline{toc}{section}{References}
\bibliography{OperationalEnvironment}

\end{document}